\RequirePackage{fix-cm} 
\documentclass[ titlepage,numbers=noenddot,headinclude,
                footinclude=true,cleardoublepage=empty,abstractoff, 
                BCOR=5mm,fontsize=11pt,american]{scrreprt}
\PassOptionsToPackage{utf8}{inputenc}
	\usepackage{inputenc}
\PassOptionsToPackage{eulerchapternumbers,listings,
					 pdfspacing,
					 subfig,beramono,eulermath,parts}{classicthesis}

\newcommand{\myTitle}{A Type System for Julia\xspace}

\newcommand{\myName}{Benjamin Chung\xspace}

\newcommand{\myFaculty}{}

\newcommand{\myUni}{Northeastern University\xspace}

\newcommand{\myTime}{May 2023\xspace}
\newcommand{\myVersion}{version 1.2\xspace}

\newcounter{dummy} 
\providecommand{\mLyX}{L\kern-.1667em\lower.25em\hbox{Y}\kern-.125emX\@}

\newcommand{\Ie}{I.\,e.}

\newcommand{\Eg}{E.\,g.}

\usepackage{babel}

\PassOptionsToPackage{fleqn}{amsmath}       
    \usepackage{amsmath}

\PassOptionsToPackage{T1}{fontenc} 
    \usepackage{fontenc}     
\usepackage{textcomp} 
\usepackage{scrhack} 
\usepackage{xspace} 
\usepackage{mparhack} 
\usepackage{fixltx2e} 
\PassOptionsToPackage{printonlyused,smaller}{acronym} 
    \usepackage{acronym} 
    
\usepackage{tabularx} 
\usepackage{caption}
\usepackage{subfig}  

\usepackage{listings} 
\PassOptionsToPackage{pdftex,hyperfootnotes=false,pdfpagelabels}{hyperref}
    \usepackage{hyperref}  
\PassOptionsToPackage{pdftex}{graphicx}
    \usepackage{graphicx}

\hypersetup{%
    colorlinks=true, linktocpage=true, pdfstartpage=3, pdfstartview=FitV,%
    breaklinks=true, pdfpagemode=UseNone, pageanchor=true, pdfpagemode=UseOutlines,%
    plainpages=false, bookmarksnumbered, bookmarksopen=true, bookmarksopenlevel=1,%
    hypertexnames=true, pdfhighlight=/O,
    urlcolor=webbrown, linkcolor=RoyalBlue, citecolor=webgreen, 
    pdftitle={\myTitle},%
    pdfauthor={\textcopyright\ \myName, \myUni, \myFaculty},%
    pdfsubject={},%
    pdfkeywords={},%
    pdfcreator={pdfLaTeX},%
    pdfproducer={LaTeX with hyperref and classicthesis}%
}   

\makeatletter
\@ifpackageloaded{babel}%
    {%
       \addto\extrasamerican{%
                }%
       \addto\extrasngerman{%
                }%
    }{\relax}
\makeatother

\usepackage{classicthesis} 

\usepackage{jlcode,amsthm,wrapfig,changepage,float,amsmath,nccmath,multirow,tabularx,pdfpages}
\usepackage{graphicx,csquotes,mathpartir,amssymb,stmaryrd,aliascnt,annotate-equations,changebar,caption}
\newtheorem{theorem}{Theorem}
\newaliascnt{definition}{theorem}
\newtheorem{definition}[definition]{Definition}
\aliascntresetthe{definition}

\theoremstyle{definition}
	\newaliascnt{lemma}{theorem}
	\newtheorem{lemma}[lemma]{Lemma}
	\aliascntresetthe{lemma}
	
\renewcommand{\c}[1]{\jlinl{#1}\xspace}

\newcommand\blfootnote[1]{%
  \begingroup
  \renewcommand\thefootnote{}\footnote{#1}%
  \addtocounter{footnote}{-1}%
  \endgroup
}
\newcommand{\eval}{\c{eval}\xspace}
\newcommand{\Eval}{\c{Eval}\xspace}
\nochangebars
\begin{document}
\frenchspacing
\raggedbottom
\pagenumbering{roman}
\pagestyle{plain}
\begin{titlepage}
    \begin{addmargin}[-1cm]{-3cm}
    \begin{center}
        \large  

        \hfill

        \vfill

        \begingroup
            \color{Maroon}\spacedallcaps{\myTitle} \\ \bigskip
        \endgroup

        \spacedlowsmallcaps{\myName}

        \vfill



        \myTime\ -- \myVersion

        \vfill                      

    \end{center}  
  \end{addmargin}       
\end{titlepage}   

\cleardoublepage
\pdfbookmark[1]{Abstract}{Abstract}
\begingroup
\let\clearpage\relax
\let\cleardoublepage\relax
\let\cleardoublepage\relax

\chapter*{Abstract}
The Julia programming language was designed to fill the needs of scientific computing by combining the benefits of productivity and performance languages. Julia allows users to write untyped scripts easily without needing to worry about many implementation details, as do other productivity languages. If one just wants to get the work done--regardless of how efficient or general the program might be---such a paradigm is ideal. Simultaneously, Julia also allows library developers to write efficient generic code that can run as fast as implementations in performance languages such as C or Fortran. This combination of user-facing ease and library developer-facing performance has proven quite attractive, and the language has increasing adoption.

With adoption comes combinatorial challenges to correctness. One of Julia's calling cards is multiple dispatch, a mechanism for solving the expression problem. On one hand, multiple dispatch allows many libraries to compose "out of the box:" for example, you can automatically differentiate an integrator merely by importing and using the appropriate packages. On the other hand, it creates bugs where one library expects features that another does not provide. Typing is one solution to this problem---mechanically ensuring that methods are used correctly---but would make Julia harder to use as a scripting language.

I developed a ``best of both worlds'' solution: gradual typing for Julia. My system forms the core of a gradual type system for Julia, laying the foundation for improving the correctness of Julia programs while not getting in the way of script writers. My framework allows methods to be individually typed or untyped, allowing users to write untyped code that interacts with typed library code and vice versa. Typed methods then get a soundness guarantee that is robust in the presence of both dynamically typed code and dynamically generated definitions. I additionally describe protocols, a mechanism for typing abstraction over concrete implementation that accommodates one common pattern in Julia libraries, and describe its implementation into my typed Julia framework.

\endgroup			

\vfill
\cleardoublepage
\pdfbookmark[1]{Acknowledgments}{acknowledgments}

\bigskip

\begingroup
\let\clearpage\relax
\let\cleardoublepage\relax
\let\cleardoublepage\relax
\chapter*{Acknowledgments} 

First, I want to thank my parents. I cannot sufficiently express my thanks for their love and support.

My colleagues Artem Pelenitsyn, Julia Belyakova, Francesco Zappa Nardelli, Celeste Hollenbeck, and Lionel Zoubritzky all played a vital part in the work that went into this thesis; this would not have been possible without you. 

Leif Andersen, Ming-Ho Yee, Aviral Goel, and Andrew Cobb were an essential part of my degree. They helped me through both the most challenging and best parts of my tenure as a student and were always there for me when I needed support.

Jonathan Aldrich and Joshua Sunshine were key to me pursuing programming languages and research generally; they helped the undergraduate me discover a love of research and taught me how to communicate in an academic setting. 

Several people both in and around the Julia team played a major part in this project. The whole collaboration is rooted in a chance encounter with Jean Yang, who I met flying back from SPLASH 2015 in Pittsburgh (after I already moved to Boston!) who was the introduction of my group to the Julia team. The Julia development team were very welcoming; in particular, I would like to thank Jeff Bezanson and Jameson Nash for advice, support, and making measured performance values obsolete extraordinarily quickly.

Last but not least, I want to thank my advisor Jan Vitek for his guidance, advice, support, and willingness to bear with my many, many side projects.

\endgroup

\pagestyle{scrheadings}
\cleardoublepage
\refstepcounter{dummy}
\pdfbookmark[1]{\contentsname}{tableofcontents}
\setcounter{tocdepth}{2} 
\setcounter{secnumdepth}{3} 
\manualmark
\markboth{\spacedlowsmallcaps{\contentsname}}{\spacedlowsmallcaps{\contentsname}}
\tableofcontents 
\automark[section]{chapter}
\renewcommand{\chaptermark}[1]{\markboth{\spacedlowsmallcaps{#1}}{\spacedlowsmallcaps{#1}}}
\renewcommand{\sectionmark}[1]{\markright{\thesection\enspace\spacedlowsmallcaps{#1}}}

\clearpage

\cleardoublepage\pagenumbering{arabic}

\chapter{Introduction}\label{ch:introduction}

Computer programming languages provide abstractions for expressing our
computational intent.  Historically, languages have come in one of two
flavors. First are productivity languages, designed to enable rapid and
interactive changes to the system at the price of performance,  facilitating
the development of software systems.  Second, performance languages are
designed to allow users to write efficient code at the cost of higher
development effort. Increasing performance often requires taking more control
over  how exactly a computation is performed which might otherwise be
abstracted away. Of course, there are many dimensions to modern computer
languages and the distinction between the two flavors can be blurred.

This thesis focuses on a feature that often distinguishes the two language
flavors: static type checking. A statically type-checkable language can be
augmented with a tool, a type checker, that catches potential errors before
the program is run. Furthermore, a program that has been checked can often be
more efficient as the compiler is able to rely on type information when making
optimization decisions. Productivity languages frequently lack a type checker,
while performance languages are much more wont to possess one. Why?

\cbstart
Most producitivity languages are untyped for two reasons. First, these
languages often include features that are inherently difficult to check before
execution of the program. Second, enforcing a static type system renders
programs more rigid, inhibiting the very exploratory programming that
productivity language users want.
\cbend

I extend a productivity language with a static type system, but do so in a
non-invasive way. In my system, statically typed code should coexist with
untyped code allowing programmers to pick whether to write code that will be 
checked or to ignore the type system when it gets in the way. The idea is that
developers of large systems, like libraries, can use static checking to
provide some measure of confidence while users who just want to get their
script working might not. Choice then allows the benefits of both productivity
and performance languages within the same framework.

A number of criteria should hold in order to reap the benefits of static
checking. First, there should be a clear and well-understood guarantee about
which errors may not occur in the checked parts of the code. Additionally, no
(or very few) performance regressions should occur due to typing. Finally, it
should be possible to opt-in or opt-out of typing so as to provide the
requisite agency to the developer.

\paragraph{Julia.} I have chosen the Julia language~\cite{bezansonthesis} as my concrete target
language. Julia is a relatively young system with a combination of
productivity and performance features. It was developed at MIT in the last
decade by a small academic team,  and has since become increasingly popular
within scientific computing.  Today, Julia is used in applications ranging
from climactic modeling to numerical optimization, with a steadily growing
user base.  There are flies in the adoption and scale ointment, however. 

Julia libraries have grown to be of considerable size, exceeding hundreds of
thousands of lines in some cases. The developers and users of these libraries
are frustrated by simple type bugs, unreliability, and confusing error
messages. One selling point of Julia is that it allows easy composition of large
libraries; however, when one composes libraries, one also composes their bugs. 
Composing libraries causes the potential for error to expand combinatorially;
each inter-library interaction begets the opportunity for some new and exciting
edge case to spring up. Julia, presently, has no formalized mechanism for checking
or even describing these interactions. Consequently, it has gained some notoriety
for being buggy due to one of its signature features\footnote{\url{https://yuri.is/not-julia/}, \url{https://github.com/JuliaStats/Distances.jl/issues/201}}.

Types are the standard answer for building abstraction over implementation.
Abstract types, in particular, allow specification of some generic behavior
that should hold for any potential instantiation of the type. Typing most
untyped languages is difficult, though, for a notion of type must be added
on top of the existing language concept. Julia is different.

In spite of being an untyped productivity language, Julia programs are full of
types. Almost all methods in the most popular Julia libraries have at least on
type annotated argument and the overwhelming majority are fully
typed~\cite{oopsla18a}. Moreover, most struct fields carry a type annotation.
If Julia is untyped, why should programmers write all of these types?

Julia programmers write types because Julia's runtime uses them. Julia's
implementation relies on types for two key applications: First, type annotated
field writes are checked. These checks are performed each time the field is
updated during program execution. Stored values can therefore be relied upon
to be type-correct. In turn, this correctness guarantee lets compiled code
safely read from typed fields. Second, types are used for dispatch, wherein
Julia decides what implementation should be invoked from a given call site.
Julia uses a mechanism called multiple dispatch to decide what method to
invoke. Each Julia function can have multiple implementations, or methods,
with each method having differently type annotated arguments. Julia will
dynamically dispatch function invocations to the method with the most
applicable type annotations. Julia therefore dynamically guarantees that
arguments will be of their declared types.

These guarantees mean little, however, when they are not exposed to the
programmer. As the scale of the Julia ecosystem increases, so does the
likelihood that some argument, somewhere, is misused. It does not matter if
the argument is a member of the statically-declared type annotation if the
method body calls a function on it that does not exist. Thus, there may be
cases where no applicable method exists for some call site. Only when code is
exercised are these errors discovered, but not all code is exercised with all
possible types during testing. Moreover, Julia's design allows users to
introduce new types and thread them through existing code, further creating
edge cases whereby invocations can go wrong.  As a result, bugs in Julia code
can hide from developers, hidden under a thin layer of passing tests and
faulty unstated assumptions, only to jump out when a user is least expecting
it. A static type system can help identify these issues ahead of time,
improving reliability of Julia code and easing development.

\paragraph{Composition.}

One of Julia's selling points is the composition of libraries using
multiple dispatch.  Suppose we wanted to differentiate the function
\c{f(x)=x*sin(x)+x*cos(x)} about \c{5}.  We could differentiate the
function's implementation on paper, but if \c{f} was even slightly
complex this becomes impractical.  We could also use finite
differencing, but this is imprecise.  In Julia, however, computing the
derivative is as simple as:
\begin{jllisting}
 > using ForwardDiff
 > ForwardDiff.derivative(f, 5)
 5.537670211431911
\end{jllisting}
The library works by replacing the argument of \c{f} with a dual
number.  The dual number extends every operation on numbers to operate
on the stored derivative value.

Consider the operation \c{x+y}, an invocation of the function \c{+} on
the variables \c{x} and \c{y}.  To evaluate this simple addition,
Julia must select which method to invoke out of over 190 options:
\begin{jllisting}
# 190 methods for generic function "+":
[1] +(x::T, y::T) where T<:Union{Int128, Int16, Int32, Int64, Int8, UInt128, UInt16, UInt32, UInt64, UInt8} in Base
[2] +(c::Union{UInt16, UInt32, UInt8}, x::BigInt) in Base.GM
[3] +(c::Union{Int16, Int32, Int8}, x::BigInt) in Base.GMP
[4] +(c::Union{UInt16, UInt32, UInt8}, x::BigFloat) in Base.MPFR
[5] +(c::Union{Int16, Int32, Int8}, x::BigFloat) in Base.MPFR
...
\end{jllisting}
The selection is based on two criteria. First, the dispatch algorithm
selects methods that are \emph{applicable}, in other word, methods
whose argument type annotations include the types of the actual values
provided.  Second, the dispatch algorithm must, out of all applicable
methods, find whichever one is the most \emph{specific}, the one that
most accurately describes the provided arguments.  All arguments
participate in dispatch; no preference is given to any single
argument.

As an example, consider the case where both \c{x} and \c{y} are \c{Int64}. 
The following methods are all \emph{applicable}:
\begin{jllisting}
[1] +(x::T, y::T) where T<:Union{Int128, Int16, Int32, Int64, Int8, UInt128, UInt16, UInt32, UInt64, UInt8} in Base
[2] +(a::Integer, b::Integer) in Base
[3] +(y::Integer, x::Rational) in Base
[4] +(x::T, y::T) where T<:Number in Base
[5] +(x::Number, y::Number) in Base
\end{jllisting}
The algorithm chooses the first of these implementations as when the
variable \c{T} is bound to \c{Int64}, the first method describes
exactly the arguments at the call site.  All other implementations
handle some additional value types, and are thus less specific.

As the number of libraries in the Julia ecosystem increases, errors become
more frequent.  Julia relies on undocumented interfaces such as \jlinl{+}. 
Programmers expect to be able to add two number-like-things together, but
there is no specification for either what a number-like-thing is, or what
\jlinl{+} should actually do.

To illustrate this, let us consider a case where the intuition for what a
number-like-thing is breaks down. Suppose that instead of using forward
differentiation to calculate the derivative of \c{f}, we wanted to do it
symbolically using computer algebra. We can do this in Julia  by using the
\jlinl{Symbolics} library as follows:
\begin{jllisting}
 > using Symbolics: variable, solve_for
 > x = variable(:x)
 > eqn = f(x)
 x*cos(x) + x*sin(x)
 > Symbolics.derivative(eqn, x)
 x*cos(x) + cos(x) + sin(x) - x*sin(x)
\end{jllisting}
Now, instead of passing a dual number to \jlinl{f} for automatic
differentiation, we pass the symbolic variable \jlinl{x}.  This
symbolic variable accumulates the operations performed on it,
constructing an equation describing the computation. We can then
symbolically differentiate this equation to determine the derivative
of the function as a whole. This can readily go wrong on a
function like \c{g}:
\begin{jllisting}
function g(x) 
   result = 0
   for i=1:x 
     result += i
   end
   return result
end
\end{jllisting}
Symbolically differentiating \c{g} with respect to \c{x} yields the error:
\vspace{-1em}
\begin{jllisting}
ERROR: TypeError: non-boolean (Num) used in boolean context
Stacktrace:
 [1] unitrange_last(start::Num, stop::Num)
   @ Base .\range.jl:294
 [2] UnitRange{Num}(start::Num, stop::Num)
   @ Base .\range.jl:287
 [3] Colon
   @ .\range.jl:5 [inlined]
 [4] Colon
   @ .\range.jl:3 [inlined]
 [5] f(x::Num)
   @ Main .\REPL[19]:3
\end{jllisting}
This is an example of where the lack of agreement on what it means to
be a number-like-thing bites us.  Here, should all number-like-things
be usable as iteration bounds?  The function \c{g} implicitly assumes
that yes, numbers should be usable as an iteration bound. Integers
satisfy this unstated assumption but abstract variables such as those
introduced by Symbolics cannot.  Moreover, there is no way to specify
that \c{g} takes arguments that can serve as loop bounds.  This
fundamental lack of agreement on what it should be possible to do to a
number can create unexpected errors from deep within programs that are
difficult to understand; why, in this example, is the error referring
to using a non-boolean as a boolean even though the real issue is that
we are trying to bound iteration with an algebraic variable? Deep
inspection of the library's source code can answer this question, but
most programmers would prefer not to have to do this to identify
simple errors of this nature.

Julia possesses no means of specifying or enforcing the existence or usage of
abstractions. Programmers introduce generic notions (such as ``this should act
like \jlinl{+}'') but then have no way to ensure that every implementation of
the notion is complete nor that usages are correct. As a result, this problem
of composition running wild only grows with the size of the Julia ecosystem.

A large number of Julia packages have been written for describing
implementations of various abstractions. None have gotten substantial
adoption. I argue that this failure is in part because none of them actually
solve the problem: they allow users to easily describe abstractions, but do
not guarantee that implementations are correct to any standard nor that usages
are safe. As a result, the benefit of using such a system is very limited.

The problem underlying these efforts to canonize some abstraction is that
checking correctness of Julia code against any given standard is hard. Julia
itself provides no mechanisms for source-level static analysis; the only tools
on offer are for an intermediate language with little direct correspondence to
source code. The only available source code analyzers are so deeply integrated
into one particular code editor (as part of the Julia Visual Studio Code
extension) and limited that they are not practical for this application.

I aim to address this gap by providing a framework for sound static type
checking in Julia. I built a system that can type check Julia code that
fits within the existing Julia type paradigms and can be extended to support
various sorts of abstractions. Moreover, I propose one particular approach to
abstraction (which I refer to as protocols, more on this later). However,
static typing is not a great fit for all use cases of Julia.

\paragraph{Gradual Typing.}
A type system answers the question of ``how can we ensure that
independently-developed code can interoperate?'' If both sides type checked
their code and the type system is sound then composition should then 
be straightforward. The problem is that, in practice, not all code is typed.

As I mentioned earlier, Julia straddles the line between performance and
productivity language. Some Julia developers such as scientists and analysts
want to use it as a productivity language, geting results fast without
worrying too much about if their code is correct. Other developers like
library programmers want to have every assurance possible that their code will
work in the ``real world.'' Where the analyst might be annoyed by type errors
the library developer might love the red squiggly line. Julia accommodates both
use cases---and a type system for Julia should as well.

A type system that can meet the needs of both the productivity and performance
developer is one that lets both typed and untyped code coexist. This mixing of
 soundly typed and untyped code is the popular concept of gradual typing.
Gradual type systems allow soundly statically typed code to run in the same
system as untyped code. 

One concern is that a key Julia feature is performance. Gradually typed Julia,
then, needs to be just as fast as base Julia if it is going to be useful
particularly for library developers---and past gradual type
systems~\cite{takikawa2016sound} have had severe performance impacts. My
type system for Julia needs to have minimal to no runtime impact.

As shown earlier, Julia code has a lot of semantically meaningful type
annotations. Normally, these annotations are used for dispatch, not for static
checking. The existence and semantics of type annotations then means that
programmers have an intuition about what it means to inhabit a type in Julia, which
is a critical issue for gradual typing. However, if I use Julia's existing type
system, then I need to be able to reason about subtyping in Julia.

\paragraph{Subtyping.}
Julia's type language inlcudes nominal single subtyping as well as union,
parametric invariant existential, and covariant tuple types. Subtyping of
unions and tuples is distributive, including over parametric types, which
Julia augments with the so-called diagonal rule.

All of these features cause Julia's type system to be theoretically
challenging. I was able to prove that subtyping in Julia is undecidable by
reduction from System $F_{<:}$~\cite{Pierce92}, showing that type checking in
Julia is potentially nonterminating.

Undecidability of subtyping is not fatal for type checking. However, it illustrates
that reasoning about and deriving useful properties from subtyping is difficult.
The type system cannot rely on any specific definition of subtyping, instead using
various approximations that aim to be sound and useful rather than complete. The type
checker needs to be parametric over operations on types.

\paragraph{Dynamic Code Generation.}
Dynamically generated code is another issue for a type system. Other untyped
languages make extensive use of features like \eval. Prior gradual
type systems for such languages ignore this dynamism as it is very difficult
to reason about code that does not yet exist. Moreover, dynamically checking
the newly-generated code is frequently impractical.

Julia, like other dynamic lanugages, supports \eval. Ideally, a type
system for Julia should then be able to support dynamically typed code. The
problem is that multiple dispatch allows functions to be extended with new
methods anywhere and anywhen---including from \eval. As a result,
\emph{no} call site is safe; any function call could have an untyped method
slipped underneath it at any time. While uncommon in practice, a sound type
system for Julia should be able to retain soundness even in the presence of
\eval.

\section{Thesis}

I posit that a static type system can be designed for Julia such that
\begin{itemize}
  \item statically typed code can interoperate with untyped methods,
  \item static type annotations do not introduce new dynamic checks,
  \item dynamically generated code does not break the whole system.
\end{itemize}
My static type system will guarantee that statically typed methods do not
go wrong while leaving the semantics of untyped code unchanged.

I additionally demonstrate one kind of abstraction for Julia programs:
protocols. A protocol defines the interface of a function that must be
implemented for all of the subtypes of the declared argument. Protocols
demonstrate how one common pattern in Julia can be typed and how future work
might use the type system to provide checked abstraction.

This work builds on a number of papers that I have co-written; papers in
\textbf{bold} are directly used or extended in this work, while those in
\textit{italics} serve as background.
\begin{itemize}
  \item \textit{Type Stability in Julia: Avoiding Performance Pathologies in JIT Compilation} \\ OOPSLA 2021
  \item \textbf{World Age in Julia: Optimizing Method Dispatch in the Presence of Eval} \\ OOPSLA 2020
  \item \textit{Julia's Efficient Algorithm for Subtyping Unions and Covariant Tuples} \\ ECOOP 2019
  \item \textbf{Julia Subtyping: a Rational Reconstruction} \\ OOPSLA 2018
  \item \textbf{Julia: Dynamism and Performance Reconciled by Design} \\ OOPSLA 2018
  \item \textit{KafKa: Gradual Typing for Objects} \\ ECOOP 2018
\end{itemize}

\renewcommand{\t}{\ensuremath{\tau}\xspace}
\newcommand{\g}{\ensuremath{\sigma}\xspace}
\newcommand{\jsub}[2]{\ensuremath{#1 <: #2}}
\chapter{The Julia Lanugage}\label{ch:julialang} 

\renewcommand{\c}[1]{\jlinl{#1}\xspace}
\renewcommand{\Ie}{\emph{i.e.}\xspace}
\renewcommand{\Eg}{\emph{e.g.}\xspace}
\newcommand{\Etal}{\emph{et al.}\xspace}
\newcommand{\Cf}{\emph{cf.}\xspace}
\newcommand{\m}{\jlinl}
\newcommand{\code}[1]{\c{#1}\xspace}

Before I can examine too deeply how I should type Julia I first need to
consider what Julia \emph{is}, exactly. Superficially, of course, Julia is a
programming language, but what was it designed to do, what  are its
distinguishing features, and how does the community take advantage of them? 
The answers to each of these questions are critical to answering how a type
system should work.

\blfootnote{This work was previously published in OOPSLA 2018~\cite{oopsla18a} as Julia: Dynamism
and performance reconciled by design. This section is an updated version of that presentation.}

The purpose of Julia is straightforward. Julia aims to bridge the gap between
productivity and performance languages. Previously, a scientist might have
written the interface to their library in Python and the backend in C++. In
contrast, Julia aims to allow programmers to write both the easy-to-use
interface and the high-performance implementation in a single homogeneous
Julia codebase.

\begin{figure}
\begin{lstlisting}
mutable struct Node
    val
    nxt
end

function insert(list, elem)
    if list isa Void
        return Node(elem, nothing)
    elseif list.val > elem
        return Node(elem, list)
    end
    list.nxt = insert(list.nxt, elem)
    list
end
\end{lstlisting}
\caption{Linked list}\label{lst}
\end{figure}

The fact that Julia delivers on its promise of having both performance and
ease of use is surprising. Dynamic languages like Python or R typically suffer
from at least an order of magnitude slowdown over C and often more. 
Fig.~\ref{lst} illustrates that Julia is indeed a dynamic language. Just as in
in Python, one can declare a \c{Node} datatype containing two untyped fields,
\c{val} and \c{nxt}, and an untyped \c{insert} function that takes a sorted
list and performs an ordered insertion. While this code will be optimized by
the Julia compiler, it is not going to run fast without some additional
programmer intervention.

The key to performance in Julia lies in the synergy between language design,
implementation techniques and programming style. Julia's design was carefully
tailored so that a very small team of language implementers could create an
efficient compiler. The key to this relative ease is to leverage the 
combination of language features and programming idioms to reduce overhead,
but what language properties enable easy compilation to fast code?

\emph{Language design:} Julia includes a number of features that are common
to many productivity languages, namely dynamic types, optional type
annotations, reflection, dynamic code loading, and garbage collection.  A
slightly less common feature is symmetric multiple dispatch~\cite{Bobrow86}.
In Julia a function can have multiple implementations, called methods,
distinguished by the type annotations added to parameters of the function.
At run-time, a function call is dispatched to the most specific method
applicable to the types of the arguments. Type annotations can be attached
to datatype declarations as well, in which case they are checked whenever
typed fields are assigned to. The language design does impose limits on some
of those features, for instance the \code{eval} function does not run in
local scope, but instead is evaluated at the top-level. Another significant
choice for optimizations is the difference between concrete and abstract
types: the former can have fields and can be instantiated while the latter
can be extended by subtypes.

\emph{Language implementation:} Performance in Julia does not arise from great
feats of compiler engineering: Julia's implementation is simpler than that of
many dynamic languages.  The Julia compiler has three main optimizations that
are performed on a high-level intermediate representation; native code
generation is then delegated to the LLVM compiler infrastructure. The
optimizations performed in Julia are (1) \emph{method inlining} which
devirtualizes multi-dispatched calls and inline the call target; (2)
\emph{object unboxing} to avoid heap allocation; and (3) \emph{method
specialization} where code is special-cased to its actual argument types. The
compiler does not support the kind of speculative compilation and
deoptimizations common in dynamic language implementations, but supports
dynamic code loading from the interpreter and with \c{eval()}.

The synergy between language design and implementation is in evidence in the
interaction between the three optimizations. Each call to a function that
has, as arguments, a combination of concrete types not observed before
triggers specialization. A data-flow analysis algorithm uses the type of the
arguments (and if these are user-defined types, the declared type of their
fields) to approximate the types of all variables in the specialized
function.  This enables both unboxing and inlining. The specialized method
is added to the function's dispatch table so that future calls with the same
combination of argument types can reuse the generated code.

\emph{Programming style:} To assist the implementation, Julia programmers
need to write idiomatic code that can be compiled effectively. Programmers
are keenly aware of the optimizations that the compiler performs and shape
their code accordingly. For instance, adding type annotations to fields of
datatypes is viewed as good practice as it provides information to the
compiler to estimate the size of instances and may allow unboxing. Another
good practice is to write methods that are \emph{type stable}. A method is
type stable if, when it is specialized to a set of concrete types, data-flow
analysis can assign concrete types to all variables in the function. This
property should hold for all specializations of the same method. Type
instability can stem from methods that can return values of different types,
from assignment of different types to the same variable depending on
branches of the function, or from functions that cannot devirtualized and
analyzed.

Julia's design for (easy implementation of) performance is critical to
understanding how the language itself and code written in Julia might be
typed. I will provide an overview of this design and how it facilitates
performance by synergizing the design with the compilation pipeline. This
close coupling provides good performance (between 0.9x and 6.1x of optimized C
code on a small suite of 10 benchmarks) and makes many of the choices in type
checking obvious. Towards this last point, I will consider a corpus of 50 popular
Julia projects on GitHub to examine how Julia's features and their underlying
design choices are exercised by the broader community and discuss how they lend
themselves towards type-ability.

\section{Related work}

\paragraph{Scientific computing languages.} R~\cite{R} and
MATLAB~\cite{MATLAB} are the two languages superficially closest to Julia.
Both languages are dynamically typed, garbage collected, vectorized and
offer an integrated development environment focused on a read-eval-print
loop.  However, the languages' attitudes towards vectorization differ. In R
and MATLAB, vectorized functions are more efficient than iterative code
whereas the contrary stands for Julia. In this context I use
``vectorization'' to refer to code that operates on entire
vectors\footnote{This discussion should not be confused with hardware-level
  vectorization, e.g. SIMD operations, which are available to Julia at the
  LLVM level.}, so for instance in R, all operations are implicitly
vectorized.  The reason vectorized operations are faster in R and MATLAB is
that the implicit loop they denote is written in a C library, while
source-level loops are interpreted and slow. In comparison, Julia can
compile loops very efficiently, as long as type information is present.

While there has been much research in compilation of
R~\cite{vee14,Graal,riposte} and MATLAB~\cite{mat1,mat2}, both languages
are far from matching the performance of Julia.  The main difference, in
terms of performance, between MATLAB or R, and Julia comes from language
design decisions. MATLAB and R are more dynamic than Julia, allowing, for
example, reflective operations to inspect and modify the current scope and
arbitrary redefinition of functions. Other issues include the lack of type
annotations on data declarations, crucial for unboxing in Julia.

Other languages have targeted the scientific computing space, most notably
IBM's X10~\cite{X10} and Oracle's Fortress~\cite{fortress}. The two languages
are both statically typed, but differ in their details. X10 focuses on
programming for multicore machines that have partitioned global addressed
spaces; its type system is designed to track the locations of values.
Fortress, on the other hand, had multiple dispatch like Julia, but never
reached a stage where its performance could be evaluated due to the complexity
of its type system. In comparison, Julia's multi-threading is still in its
infancy, and it does not have any support for partitioned address spaces.

\paragraph{Multiple dispatch.} Multiple dispatch goes back to~\cite{Bobrow86}
and is used in languages such as CLOS~\cite{Gabriel87}, Perl~\cite{Randal03}
and R~\cite{Chambers14}. Lifting explicit programmatic type tests into
dispatch requires an expressive annotation sublanguage to capture the same
logic; expressiveness that has created substantial research challenges.
Researchers have struggled with how to provide expressiveness while ensuring
type soundness. Languages such as Cecil~\cite{lit98} and
Fortress~\cite{Allen11} are notable for their rich type systems; but, as
mentioned in Guy Steele's retrospective talk, finding an efficient,
expressive and sound type system remains an open
challenge.\footnote{JuliaCon 2016,
  \url{https://www.youtube.com/watch?v=EZD3Scuv02g}.} The language design
trade-off seems to be that programmers want to express relations between
arguments that require complex types, but when types are rich enough, static
type checking becomes difficult. The Fortress designers were not able to
prove soundness, and the project ended before they could get external
validation of their design.
Julia side-steps many of the problems encountered in previous work on typed
programming languages with multiple dispatch. It makes no attempt to
statically ensure invocation soundness or prevent ambiguities, falling back to
dynamic errors in these cases.

\paragraph{Static type inference.} At heart, despite the allure of types and
the optimizations they allow, type inference for untyped programs is
difficult. Flow typing tries to propagate types through the program at
large, but sacrifices soundness in the process. Soft
typing~\cite{fagan1991soft} applies Hindley-Milner type inference to untyped
programs, enabling optimizations. This approach has been applied practically
in Chez Scheme~\cite{wright1994practical}. However, Hindley-Milner type
inference is too slow to use on practically large code bases. Moreover, many
language features (such as subtyping) are incompatible with it. Constraint
propagation or dataflow type inference systems are a commonly used
alternative to Hindley-Milner inference. These systems work by propagating
types in a data flow analysis~\cite{aiken1993type}. No unification is
needed, and it is therefore much faster and more flexible than soft
typing. Several inference systems based on data flow have been proposed for
JavaScript~\cite{chaudhuri2017fast}, Scheme~\cite{Shiv90}, and others.

\paragraph{Dynamic type inference for JIT optimizations.} Feeding dynamic type
information into a type propagation type inference system is not a technique
new to Julia. The first system to use dataflow type inference inside a JIT
compiler was RATA~\cite{RATA}. RATA relies on abstract interpretation of
dynamically-discovered intervals, kinds, and variations to infer extremely
precise types for JavaScript code; types which enable JIT optimizations. The
same approach was then used by Hackett~\cite{hackett2012fast}, which used a
simplified type propagation system to infer types for more general JavaScript
code, providing performance improvements. In comparison to dynamic type
inference systems for JavaScript, Julia's richer type annotations and multiple
dispatch allow it to infer more precise types.
Another related project is the StaDyn~\cite{GarciaDyn} language. StaDyn was
designed specifically with hybrid static and dynamic type inference in mind.
However, StaDyn does not have many of Julia's features that enable precise
type inference, including typed fields and multiple dispatch.

\paragraph{Dynamic language implementation.} Modern dynamic language
implementation techniques can be traced back to the work on
the Self language, that pioneered the ideas of run-time specialization and
deoptimization~\cite{Holzle94}. These ideas were then transferred into the Java HotSpot
compiler~\cite{hotspot}; in HotSpot, static type information can be used to
determine out object layout, and deoptimization is used when inlining
decisions were invalidated by newly loaded code. Implementations of JavaScript
have increased the degree of specialization, for instance allowing unboxed
primitive arrays at the more complex guards and potentially wide-ranging
deoptimization~\cite{Graal}.

\section{Julia in action}\label{act}

To introduce Julia, let's consider an example function. This code started as an
attempt to replicate the R language's multi-dimensional summary
function. For explanatory reason, I shortened the code somewhat, the
shortened version simply computes the sum of a vector. Just like the R
function that inspired it, the Julia code is polymorphic over vectors of
integer, float, boolean, and complex values. Furthermore, since R supports
missing values in every data type, I encode \c{NA}s in
Julia.\footnote{Since Julia v1.0 support for missing values is native, this
  example shows how it could be encoded.}

\begin{figure}[H]
\begin{minipage}{.394\columnwidth}
\begin{lstlisting}
function vsum(x)
    sum = zero(x)
    for i = 1:length(x)
        @inbounds v = x[i]
        if !is_na(v)
            sum += v
        end
    end
    sum
end
\end{lstlisting}
\caption{Compute vector sum}\label{sum}
\vspace{1em}

{\scriptsize
\begin{verbatim}
    push    %rbp
    mov     %rsp, %rbp
    mov    (%rdi), %rcx
    mov     8(%rdi), %rdx
    xor     %eax, %eax
    test    %rdx, %rdx
    cmove   %rax, %rdx
    movl    $1, %esi
    movabs  $0x8000000000000000, %r8
    jmp     L54
    nopw    %cs:(%rax,%rax)
L48:add     %rdi, %rax
    inc     %rsi
L54:dec     %rsi
    nopl    (%rax)
L64:cmp     %rsi, %rdx
    je      L83
    mov     (%rcx,%rsi,8), %rdi
    inc     %rsi
    cmp     %r8, %rdi
    je      L64
    jmp     L48
L83:pop     %rbp
    ret 
    nopw    %cs:(%rax,%rax)
\end{verbatim}}
\caption{\protect\jlinl{@code_native vsum([1])} (X86-64)}\label{native}

\end{minipage}
\hfill
\begin{minipage}{.6\columnwidth}
\begin{adjustwidth}{1em}{0em}
\begin{lstlisting}[linewidth=0.98\columnwidth]
zero(::Array{T}) where {T<:AbstractFloat} = 0.0
zero(::Array{T}) where {T<:Complex} = complex(0.0,0.0)
zero(x) = 0

is_na(x::T) where T = x == typemin(T)

typemin(::Type{Complex{T}}) where {T<:Real} 
      = Complex{T}(-NaN)
\end{lstlisting}

\caption{\protect\c{zero} yields the zero matching the element type, by default the
  integer \protect\c{0}.  \protect\c{is_na} checks for missing values encoded as the
  smallest element of a type (returned by the builtin function \protect\c{typemin}).
  \protect\c{typemin} is extended with a method to return the smallest complex
  value}\label{zero}

\vspace{2.3em}

\begin{lstlisting}[linewidth=0.98\columnwidth]
primitive type RBool 8 end 

RBool(x::UInt8) = reinterpret(RBool, x)
convert(::Type{T},x::RBool) where{T<:Real} 
    = T(reinterpret(UInt8,x))

const T  = RBool(0x1)
const F  = RBool(0x0)
const NA = RBool(0xff)

typemin(::Type{RBool}) = NA

+(x::Union{Int,RBool}, y::RBool) = Int(x) + Int(y)
\end{lstlisting}
\cbstart
\caption{\protect\c{RBool} is a an 8-bit primitive type representing boolean values
  extended with a missing value \protect\c{NA}. The constructor takes an 8-bit unsigned
  integer. Conversion casts any number into an \protect\c{RBool} by \protect\c{reinterpret}ing the
  in-memory representation as a \protect\c{RBool}. A new
  method is added to \protect\c{typemin} to return \protect\c{NA}}
\cbend 
\label{rbool}
\end{adjustwidth}
\end{minipage}
\end{figure}

Fig.~\ref{sum} shows how to sum values of a vector \c{x} of any type. As the
Julia syntax is straightforward, little explanation is required to
understand the programmer's intent. In this case, type annotations are not
needed for the compiler to optimize the code, so I omit them. Variables are
lexically scoped; an initial assignment defines them. Fig.~\ref{native} is
the output of \c{@code_native(vsum([1]))} (a call to the function with a
vector of integers). It shows the x86 machine code generated for
the specialized method. It is noteworthy that the generated machine code does not
contain object allocation or method invocation, nor does it invoke any
language runtime components. The machine code is similar to code one would
expect to be emitted by a C compiler.

Type stability is key to performant Julia code, allowing the compiler to
optimize using types. An expression is type stable if, in a given type
context, it always returns a value of the same type. Function \c{vsum(x)}
always returns a value that is either of the same type as the element type
of \c{x} (for floating point and complex vectors) or \c{Int64}. For the call
\c{vsum([1])}, the method returns an \c{Int64}, as its argument is of type
\c{Array\{Int64,1\}}. When presented with such a call, the Julia compiler
specializes the method for that type. Specialization provides enough
information to determine that all values manipulated by the computation are
of the same type, \c{Int64}. Thus, no boxing is required; moreover, all
calls are devirtualized and inlined. The \c{@inbounds} macro elides array
bounds checking.

Type stability may require cooperation from the developer. Consider variable
\c{sum}: its type has to match the element type of \c{x}. In our case,
\c{sum} must be appropriately initialized to support any of the possible
argument types integer, float, complex or boolean. To ensure type stability,
the programmer leverages dispatch and specialization with the definition of
the function \c{zero} shown in Fig.~\ref{zero}. It dispatches on the type of
its argument. If the argument is an array containing subtypes of float, the
function returns float \c{0.0}. Similarly, if passed an array containing
complex numbers, the function returns a complex zero. In all other cases, it
returns integer \c{0}. All three methods are trivially type stable, as they
always return the same value for the same types.

Missing values also require attention. Each primitive type needs its own
representation---yet the code for checking whether a value is missing must
remain type stable. This can be achieved by leveraging dispatch. I add a
function \c{is_na(x)} that returns true if \c{x} is missing. I select the
smallest value in each type to use as its missing value (obtained by calling
the builtin function \c{typemin}).

The solution outlined so far fails for booleans, as their minimum is
\c{false}, which I can't steal. Fig.~\ref{rbool} shows how to add a new
boolean data type, \c{RBool}. Like Julia's boolean, \c{RBool} is represented
as an 8-bit value; but like R's boolean, it has three values, true, false
and missing. Defining a new data type entails providing a constructor and a
conversion function. Since our data type has only three useful values, we
enumerate them as constants. I then add a method to \c{typemin} to return
\c{NA}.  Finally, since the loop adds booleans to integers, I need to
extend addition to integer and \c{RBool}, this is done by interpreting
true as 1 and false as zero.

\section{Evaluating relative performance}\label{rel}

\begin{minipage}{\textwidth}
\begin{wrapfigure}{r}{.5\textwidth}
\includegraphics[width=.48\textwidth]{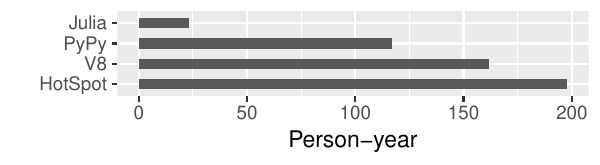}
\caption{Time spent on implementations}\label{my}
\end{wrapfigure}

Julia has to be fast to compete against other languages used for scientific
computing, but it also has to be easy to develop and maintain.  Programming
languages are notoriously expensive propositions in terms of the level of
expertise required during development and the effort required to achieve
production-quality outcomes. Fig.~\ref{my} shows a very rough estimate of the
the person-years invested in several language implementations. These blunt
approximations were obtained using commit histories: two commits made by the
same developer in one week were counted as one person-week of effort. While
approximate, this figure suggests that performance comes at a substantial cost
in engineering. For example, V8 for JavaScript and HotSpot for Java, two
high-performance implementations, have nearly two person-centuries invested into
their respective implementations. Even PyPy, an academic project, has over
one century of work by our metric. Given the difference in implementation
effort, the fact that Julia's performance is competitive is surprising.
\end{minipage}

To estimate the languages' relative performance, I selected 10 small
programs for which implementations in C, JavaScript, and Python are
available in the programming language benchmark game (PLBG)
suite~\cite{PLBG}. The suite consists of small but non-trivial benchmarks
which stress either computational or memory performance. I started with
PLBG programs written by the Julia developers and fixed some performance
anomalies. The benchmarks are written in an idiomatic style, using the same
algorithms as the C benchmarks. Their code is largely untyped, with type
annotations only appearing on structure fields. Over the 10 benchmark
programs, 12 type annotations appear, all on structs and only in the nbody,
binary\_trees, and knucleotide. The \c{@inbounds} macro eliding bounds
checking is the only low-level optimization used, leveraged only in revcomp.
Using the PLBG methodology, I measured the size of the programs by removing
comments and duplicate whitespace characters, then performing the minimal
GZip compression. The combined size of all the benchmarks is $6$ KB for
Julia, $7.4$ KB for JavaScript, $8.3$ KB for Python and $14.2$ KB for C.

\begin{figure}[H]
\includegraphics[width=\columnwidth]{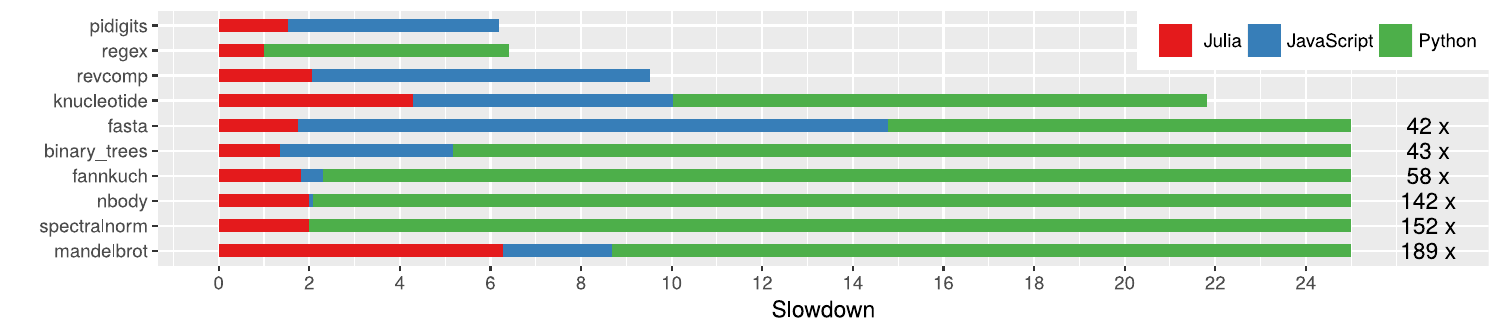}
\caption{Slowdown of Julia, JavaScript, and Python relative to C}\label{perf}
\end{figure}

Fig.~\ref{perf} compares the performance of the four languages with the
results normalized to the running time of the C programs. Measurements were
obtained using Julia v0.6.2, CPython 3.5.3, V8/Node.js v8.11.1, and GCC 6.3.0 -O2
for C, running on Debian 9.4 on a Intel i7-950 at 3.07GHz with 10GB of RAM.
All benchmarks ran single threaded. No other optimization flags were used.

The results show Julia consistently outperforming Python and JavaScript
(with the exception of spectralnorm).  Julia is mostly within 2x of
C. Slowdowns are likely due to memory operations. Like other high level
dynamically-typed programming languages, Julia relies on a garbage collector
to manage memory. It prohibits the kind of explicit memory management tricks
that C allows. In particular, it allocates structs on the heap. Stack
allocation is only used in limited circumstances. Moreover, Julia disallows
pointer arithmetic.

Three programs fall outside of this range: two programs (knucleotide and
mandelbrot) have slowdowns greater than 2x over C, while one (regex) is
faster than C. The knucleotide benchmark was written for clarity over
performance; it makes heavy use of abstractly-typed struct fields (which
cause the values they denote to be boxed).  In the case of mandelbrot, the C
code is manually vectorized to compute the fractal image 8 pixels at a time;
Julia's implementation, however, computes one pixel at a time. Finally,
regex, which was within the margin of error of C, simply calls into the same
regex library C does.

Julia is fast on tiny benchmarks, but this may not be representative of
real-world programs. I lack the benchmarks to gauge Julia's performance at
scale. Some libraries have published comparisons. JuMP, a large embedded
domain specific language for mathematical optimization, is one such library.
JuMP converts numerous problem types (e.g. linear, integer linear, convex,
and nonlinear) into standard form for solvers. When compared to equivalent
implementations in C++, MATLAB, and Python, JuMP is within 2x of C++. For
comparison, MATLAB libraries are between 4x and 18x slower than C++, while
Python's optimization frameworks are at least 70x slower than
C++~\cite{Miles13}. This provides some evidence that Julia's performance on
small benchmarks may carry over to larger programs.

\section{The Julia programming language}\label{over}

The designers of Julia set out to develop a language specifically for the
needs of scientific computation, and they chose a finely tuned set of
features to support this use case. Antecedent languages, like R and MATLAB,
illustrate scientific programmers' desire to write high-level scripts, which
motivated Julia's adoption of an optionally typed surface language.
Likewise, these languages drove home the importance of flexibility:
programmers regularly extend core language functionalities to fit
their needs. Julia provides this extensibility mechanism through multiple
dispatch.

\subsection{Values, types, and annotations}

\subsubsection{Values}

Values can be either instances of \emph{primitive types}, represented as
sequences of bits, or \emph{composite types}, represented as a collection of
fields holding values.  Logically, every value is tagged by its full type
description; in practice, however, tags are often elided when they can be
inferred from context.  Composite types are immutable by default, thus
assignment to their fields is not allowed. This restriction is lifted when
the \c{mutable} keyword is used.

\subsubsection{Types declarations}

Programmers can declare three kinds of types: \emph{abstract types},
\emph{primitive types}, and \emph{composite types}. Types can be parametrized
by bounded type variables and have a single supertype.  The type \c{Any} is
the root of the type hierarchy, or the greatest supertype
(top). \emph{Abstract} types cannot be instantiated; \emph{concrete} types
can.

\begin{jllisting}
abstract type Number end

abstract type Real <: Number end

primitive type Int64 <: Signed 64 end

struct Polar{T<:Real} <: Number
    r::T
    t::T
end
\end{jllisting}

The code shown is an extract of Julia's numeric tower. \c{Number} is an
abstract type with no declared supertype, which means \c{Any} is its
supertype. \c{Real} is also abstract but has \c{Number} as its supertype.
\c{Int64} is a primitive type with \c{Signed} as its supertype; it is
represented in 64 bits. The struct \c{Polar\{T<:Real\}} is a subtype of
\c{Number} with two fields of type \c{T} bounded by \c{Real}. Run-time
checks ensure that values stored in these fields are of the declared type.
When types are omitted from field declarations, fields can hold values of
\c{Any} type.  Julia does not make a distinction between reference and value
types as Java does. Concrete types can be manipulated either by value or by
reference; the choice is left to the implementation. Abstract types,
however, are always manipulated by reference. It is noteworthy that
composite types do not admit subtypes; therefore, types such as \c{Polar}
are final and cannot be extended with additional fields.

\subsubsection{Type annotations}

Julia offers a rich type annotation language to express constraints on
fields, parameters, local variables, and method return 
types~\cite{oopsla18b}. The \c{::} operator
ascribes a type to a definition. The annotation language includes union
types, written \c{Union\{A,...\}}; tuple types, written \c{Tuple\{A,...\}};
iterated union types, written \c{TExp where A<:T<:B}; and singleton types,
written \c{Type\{T\}} or \c{Val\{V\}}. The distinguished type \c{Union\{\}},
with no argument, has no value and acts as the bottom type.

Union types are abstract types which include, as values, all instances of
their arguments. Thus, \code{Union\{Integer, String\}} denotes the set of
all integers and strings.  Tuple types describe the types of the elements
that may be instantiated within a given tuple, along with their order. They
are parametrized, immutable types.  Additionally, they are \emph{covariant}
in their parameters. The last parameter of a tuple type may optionally be
the special type \c{Vararg}, which denotes any number of trailing elements.

Julia provides iterated union types to allow quantification over a range of
possible instantiations.  For example, the denotation of a polar coordinate
represented using a subtype \c{T} of real numbers is
\c{Polar\{T\}}~\c{where Union\{\}<:T<:Real}. Each \c{where} clause
introduces a single type variable. The type application syntax \c{A\{B\}}
requires \c{A} to be a \c{where} type, and substitutes \c{B} for the
outermost type variable in \c{A}. Type variable bounds can refer to outer
type variables. For example,

\begin{center}
\c{Tuple\{T, S\}}~\c{where S<:AbstractArray\{T\}}~\c{where T<:Real}
\end{center}

\noindent
refers to 2-tuples whose first element is some \c{Real}, and whose second
element is an array whose element type is the type of the first tuple
element, \c{T}.

A singleton type is a special kind of abstract type, \c{Type\{T\}}, whose
only instance is the object \c{T}.

\subsubsection{Subtyping}

In Julia, the subtyping relation between types, written \c{<:}, is used in
run-time casts, as well as method dispatch.  Semantic subtyping partially
influenced Julia's subtyping~\cite{Frisch02}, but practical considerations
caused Julia to evolve in a unique direction.  Julia has an original
combination of \emph{nominal subtyping}, \emph{union types}, \emph{iterated
  union types}, \emph{covariant} and \emph{invariant} constructors, and
\emph{singleton types}, as well as the \emph{diagonal rule}. Parametric
types are invariant in their parameters because this allows the 
Julia compiler to perform optimizations dependent on the memory
representation of values.  Arrays of dissimilar values box each of their
arguments, for consistent element size, under type
\c{Array\{Any\}}. However, if all the values are statically determined to be
of the same kind, they are stored inside of the array itself. 
Tuple types represent both tuples of values and function
arguments. They are covariant as this allows Julia
to compute dispatch using subtyping of tuples. Subtyping of union types is
asymmetrical but intuitive.  Whenever a union type appears on the left-hand
side of a judgment, as in \c{Union\{T1,...\}}~\c{<: T}, all the types
\c{T1,...}  must be subtypes of \c{T}.  In contrast, if a union type appears
on the right-hand side, as in \c{T <: Union\{T1,...\}}, then only one type,
\c{Ti}, needs to be a supertype of \c{T}.  Covariant tuples are distributive
with respect to unions, so \c{Tuple\{Union\{A,B\}, C\} \ <:
  Union\{Tuple\{A,C\}, Tuple\{B,C\}\}}.  The iterated union construct
\c{TExp where A<:T<:B}, as with union types, must have either a ``forall'' or an
``exist'' semantics, according to whether the union appears on the left or right
of a subtyping judgment.  Finally, the \emph{diagonal rule} states that if a
variable occurs more than once in covariant position, it is restricted to
ranging over only concrete types. For example, \c{Tuple\{T,T\}}~\c{where T}
can be seen as \c{Union\{Tuple\{Int8,Int8\}, Tuple\{Int16,Int16\}, ...\}},
where \c{T} ranges over all concrete types.  The details of  subtyping
are intricate and the interactions between  features can be
surprising, described in the paper~\cite{oopsla18b}.

\subsubsection{Dynamically-checked type assertions}

Type annotations in method arguments are guaranteed by the language
semantics. A method executes only if all of its arguments have types that
match their declarations.  However, Julia allows type annotations elsewhere
in the program, these act as checked type assertions. For example, to
guarantee that variable \c{x} has type \c{Int64}, the assertion \c{x::Int64}
can be inserted into its declaration.  Likewise, functions can assert a
return type, as in \c{f()::Int = ...} for example. Fields and expressions
can also be annotated.  These annotations check the type of the expression's
or field's value.  If that type is not a subtype of the declared type, Julia
will try to convert it to the declared type. If this conversion fails, 
an exception is thrown. 

\subsection{Multiple dispatch}

Julia uses multiple dispatch extensively, allowing extension of
functionality by means of overloading. Each function (for example \c{+}) can
consist of an arbitrarily large number of methods (in the case of \c{+},
180). Each of these methods declares what types it can handle, and Julia
will dispatch to whichever method is most specific for a given call.  As
hinted at with addition, multiple dispatch is omnipresent.  Virtually every
operation in Julia involves dispatch.  New methods can be added to existing
functions, extending them to work with new types.

\noindent\begin{minipage}{\textwidth}
\subsubsection{Example}
\noindent\begin{wrapfigure}{r}{7cm}
\begin{lstlisting}[linewidth=7cm]
struct Dual{T}
    re::T
    dx::T
end

function Base.:(+)(a::Dual{T},b::Dual{T}) where T
    Dual{T}(a.re+b.re, a.dx+b.dx)
end
function Base.:(*)(a::Dual{T},b::Dual{T}) where T
    Dual{T}(a.re*b.re, a.dx*b.re+b.dx*a.re)
end
function Base.:(/)(a::Dual{T},b::Dual{T}) where T
    Dual{T}(a.re/b.re, (a.dx*b.re-a.re*b.dx)/(b.re*b.re))
end
\end{lstlisting}

\end{wrapfigure}

\cbstart
Consider forward differentiation, a technique that allows derivatives to be
calculated for arbitrary programs. It is implemented threading a value
together with its derivative through a program. In many languages, the code
being differentiated would have to be aware of forward differentiation as the
dual numbers need new definitions of arithmetic.  Multiple dispatch allows to
implement a library that works for existing functions, as I can simply extend
arithmetic operators.  Suppose I want to compute the derivative of \c{f(a,b)
= a*b/(b+b*a+b*b)} about \c{a}, with \c{a =} \c{1} and \c{b =} \c{3}. Forward
differentiation works by introducing a concept of dual numbers as shown in the
example. Dual numbers consist of a real component (the actual value being
computed) and the derivative of that number (\c{dx}, in the example).
Differentiation is then performed by implementing the chain rule for whatever
operation is then performed. In the case of addition, for instance, the real
and derivative components of the two dual numbers are simply added. Multiplication, on
the other hand scales the derivatives of the terms by the opposing real component to determine
the derivative of the final value.

I can implement forward differentiation in Julia very easily by overloading
arithmetic. As seen in the example, I can simply add new definitions for the
same operators that are used for all other arithmetic operations. Since we
covered all of the operations used in the function \c{f}, I can now figure
out the derivative of \c{f} by simply calling it with dual numbers:
\c{f(Dual(1.0,1.0), Dual(3.0,0.0)).dx} yields \c{0.16}.
\cbend

\end{minipage}

\subsubsection{Semantics}
Dispatching on a function \c{f} for a call with argument type \c{T} consists
in picking a method $m$ from all the methods of \c{f}. The selection filters
out methods whose types are not a supertype of \c{T} and takes the method
whose type \c{T'} is the most specific of the remaining ones. In contrast to
single dispatch, every position in the tuples \c{T} and \c{T'} have the same
role---there is no single receiver position that takes precedence.
Specificity is required to disambiguate between two or more methods which
are all supertypes of the argument type. It extends subtyping with extra
rules, allowing comparison of dissimilar types as well. The specificity
rules are defined by the implementation and lack a formal semantics. In
general, \c{A} is more specific than \c{B} if \c A~\c{!=}~\c B and either \c
A \c{<:} \c B or one of a number of special cases hold:

\begin{enumerate}
\item[(a)] \c A \c = \c{R\{P\}} and \c B = \c{S\{Q\}}, and there exist
  values of \c P and \c Q such that \c R \c{<:} \c S. This allows us to
  conclude that \c{Array\{U\}} \c{where U} is more specific than
  \c{AbstractArray\{String\}}.
\item[(b)] Let \c C be the non-empty meet (approximate intersection) of \c A
  and \c B, and \c C is more specific than \c B and \c B is not more
  specific than \c A. This is used for union types: \c{Union\{Int32,
    String\}} is more specific than \c{Number} because the meet, \c{Int32},
  is clearly more specific than \c{Number}.
\item[(c)] \c A and \c B are tuple types, \c A ends with a \c{Vararg} type
  and \c A would be more specific than \c B if its \c{Vararg} was expanded
  to give it the same number of elements as \c B. This tells us that
  \c{Tuple\{Int32,Vararg\{Int32\}\}} is more specific than
  \c{Tuple\{Number,Int32,Int32\}}.
\item[(d)] \c A and \c B have parameters and compatible structures, \c A
  provides a consistent assignment of non-\c{Any} types to replace \c B's
  type variables, regardless of the diagonal rule. This means that
  \c{Tuple\{Int, Number, Number\}} is more specific than \c{Tuple\{R,S,S\}}
  \c{where \{R, S<:R\}}.
\item[(e)] \c A and \c B have parameters and compatible structures and \c
  A's parameters are equal or more specific than \c B's. As a consequence,
  \c{Tuple\{Array\{R\}} \c{where R, Number\}} is more specific than
  \c{Tuple\{AbstractArray\{String\}, Number\}}.
\end{enumerate}

\noindent\begin{minipage}{\textwidth}
\begin{wrapfigure}[7]{r}{6cm}
\begin{lstlisting}[linewidth=6cm,aboveskip=0cm]
ntuple(f, ::Type{Val{0}}) = (@_inline_meta; ())
ntuple(f, ::Type{Val{1}}) = (@_inline_meta; (f(1),))
ntuple(f, ::Type{Val{2}}) = (@_inline_meta; (f(1), f(2)))
\end{lstlisting}
\end{wrapfigure}

\noindent
One interesting feature is dispatch on type objects and on primitive
values.  For example, the Base library's \c{nutple} function is defined as a
set of methods dispatching on the value of their second argument. Thus a call to \c{ntuple(id, Val\{2\})} yields \c{(1,2)} where \c{id} is the
identity function.  The \c{@_inline_meta} macro is used to force inling. 
\end{minipage}

\subsection{Metaprogramming}

Julia provides various features for defining functions at compile-time and
run-time and has a particular definition of visibility for these
definitions.

\noindent\begin{minipage}{\textwidth}
\subsubsection{Macros}
\begin{wrapfigure}[8]{r}{6cm}
\begin{lstlisting}[linewidth=6.15cm,aboveskip=0cm]
macro assert(ex, msgs...)
    msg_body = isempty(msgs) ? ex : msgs[1]
    msg = string(msg_body)
    return :($ex ? nothing
       : throw(AssertionError($msg)))
end
\end{lstlisting}
\end{wrapfigure}

Macros provide a way to generate code and reduce the need for \c{eval()}.  A
macro maps a tuple of arguments to an expression which is compiled
directly. Macro arguments may include expressions, literal values, and
symbols. The example on the right shows the definition of the \c{assert}
macro which either returns \c{nothing} if the assertion is true or throws an
exception with an optional message provided by the user. The \c{:(...)}
syntax denotes quotation, that is the creation of an expression. Within it,
values can be interpolated: \c{\$x} will be replaced by the value of \c{x}
in the expression. Once defined, this macro can then be used to make assertions
like \c{@assert 1 + 1 == 2}.
\end{minipage}
\vspace{1em}

\cbstart
Another form of macro available in Julia is the string macro. String macros
allow static compilation of string literals. One example in the Julia standard
library is regular expressions: \c{r".*"} defines a regular expression that
matches a string of any characters and length, for example. String literal
macros are implemented very much like normal macros: they are only distinguished
by a \c{_str} suffix. As an example, the regular expression
macro is defined as \c{macro r_str(p) Regex(p) end}. A string macro implementation is
then simply passed the string literal which it can then analyze or otherwise process
as part of expansion.

Macros in Julia are unhygenic: macro developers can easily bind to and
introduce new external syntactic forms if they so wish. Julia implements
a sort of superficial hygine, wherein macro-introduced symbols are by default
analyzed and rewritten with generated unique identifiers. However, this system
can be opted out of using the \c{esc} expression form. Macros, as a result, can
introduce new forms. For instance, the JuMP library introduces the
\c{@variable} macro which defines a new optimization variable and binds it into scope.
As an example, if I wanted to introduce a variable \c{x} to a model with an upper
bound of 2, I could do it with \c{@variable(model, x <= 2)}. The variable \c{x} will
now be in scope and be inititalized with the desired constraint.

Macros have found a wide range of use cases in Julia. A few common patterns are:
\begin{itemize}
    \item \textbf{Sugar}: Macros like \c{@assert} or \c{@debug} encapsulate some simple but extremely common and otherwise tedious operation, such as asserting that an expression is true or logging a message at debug level.
    \item \textbf{Semantic}: Macros such as \c{@inline} mark expressions with metadata to alter how they are compiled. Similarly,
    macros such as \c{@.} modify the semantics of the expression they're given. The \c{@.} macro, for instance, automatically vectorizes the expression it's given.
    \item \textbf{DSL}: As seen in JuMP, another use case of macros is to define DSLs. DSLs in Julia can sometimes reuse the existing Julia grammar (as in the case of JuMP) or deviate wholly from it (as seen in the example of the regular expression macro).
\end{itemize}
\cbend
\subsubsection{Reflection}

Julia provides methods for run-time introspection.  The names of fields may
be interrogated using \c{fieldnames()} and their types, with
\c{fieldtype()}. Types are themselves
represented as a structure called \c{DataType}. The direct subtypes of any
\c{DataType} may be listed using \c{subtypes()}. The internal representation
of a \c{DataType} is important when interfacing with C code and several
functions are available to inspect these details. \c{isbits(T::DataType)}
returns true if \c{T} is stored with C-compatible alignment. The builtin
function \c{fieldoffset(T::DataType, i::Integer)} returns the offset for
field \c{i} relative to the start of the type. The methods of any function
may be listed using \c{methods()}. The method dispatch table may be searched
for methods accepting a given type using \c{methodswith()}.

\noindent\begin{minipage}{\textwidth}
\begin{wrapfigure}[5]{r}{6cm}
\begin{lstlisting}[linewidth=6.15cm,aboveskip=0cm]
for op in (:+, :*, :&, :|)
    eval(:($op(a,b,c) = $op($op(a,b),c)))
end
\end{lstlisting}
\end{wrapfigure}

More powerful is the \c{eval()} function which takes an expression object
and evaluates it in the global scope of the current module. For example
\c{eval(:(1+2))} will take the expression \c{:(1+2)} and evaluate it
yielding the expected result. 
\end{minipage}
\vspace{1em}

When combined with an invocation to the parser, any arbitrary string can be
evaluated, so for instance \c{eval(parse("function id(x) x end"))} adds an
identity method. One important difference from languages such as JavaScript is
that \c{eval()} does not have access to the current scope. This is crucial for
optimizations as it means that local variables are protected from
interference. The \c{eval()} function is sometimes used as part of code
generation. Here for example is a generalization of some of the basic binary
operators to three arguments. This generates four new methods of three
arguments each.

\subsubsection{World Age}

World age is a critical component of Julia's design for performance. It
arose out of a problem encountered with the one-shot JIT compilation strategy:
what happens if the set of methods changes?

\begin{figure}[H]
 \begin{minipage}{.27\textwidth}
  \begin{lstlisting}
> x = 3.14
> f(x) = (
    eval(:(x = 0)),
    x * 2,
    Main.x)

> f(42) # (0, 84, 0)
> x     # 0
\end{lstlisting}\vspace{-2mm}
\caption{Scope of eval in Julia}\label{fig:eval-scope}
\end{minipage}
\begin{minipage}{.05\textwidth}
\hspace{1em}
\end{minipage}
\begin{minipage}{.31\textwidth}\begin{lstlisting}[morekeywords={defn}]
> (defn g [] 2)
> (defn f [x]
    (eval `(defn g [] ~x))
    (* x (g)))
> (f 42) ; 1764
> (g)    ; 42
> (f 42) ; 1764
\end{lstlisting}\vspace{-2mm}
    \caption{Eval in Clojure}\label{fig:clojure}
  \end{minipage}
\begin{minipage}{.31\textwidth}\begin{lstlisting}
> g() = 2
> f(x) = (
    eval(:(g()=$x));
    x * g())
> f(42) # 84
> g()   # 42
> f(42) # 1764
\end{lstlisting}\vspace{-2mm}
\caption{Eval in Julia}\label{fig:eval-methods}
\end{minipage}
\end{figure}

Using \c{eval} Julia programs can modify the global state at any time
including both simple variables but functions too. Modifications are
limited to the global scope---as shown in~\ref{fig:eval-scope}, local
variables are not changed by \c{eval}---but any global reference may be
altered using \c{eval} at any time. As seen in the example, the local
reference to \c{x} was unaffected but explicitly referring to the outer
global \c{Main.x} shows the new value immediately.

In most dynamic languages global function lookup works the same way as
these variable assignements do. For example, in~\ref{fig:clojure} I see
that if I use \c{eval} to define a new implementation of \c{g} mid-\c{f}
then \c{f} ``picks up'' that definition of \c{g} immediately. This semantics
gives Julia serious problems, however, for the one-shot JITting model means
that \c{f} would have a compiled implementation that is now referring to the
wrong \c{g}; Julia would have to either make all method invocations dynamic or 
implement on-stack-replacement to back out the compilation in order to support
it.

World age is Julia's solution to this problem. If faced with a hard problem
one can either face it head on (and implement deoptimization/on-stack
replacement, in this example) or define it out of existence. Julia took the
latter approach: world age concretizes what definitions running code has
access to therein providing a consistent semantics for both compiled and
dynamically dispatched method invocations.

The action of world age on our example is shown in figure~\ref{fig:eval-methods}.
The function \c{f} is defined exactly as it is in Clojure and yet its results
are different; instead of referring to the newly-added method \c{g} (that returns
\c{42}), it uses the version defined when \c{f} was first called returning 2. 
Thus, the first result of calling \c{f} is 84. Only after execution returns
to the top level does the new definition of \c{g} become visible to \c{f}, at
which point the result matches Clojure's.

Besides returning to the top level (either explicitly or through use of
\c{eval}, which executes its argument at top level) programmers can use the
\c{invokelatest} function to use whatever the newest definitions are. These
tools provide escape hatches for cases in which programmers do want to access
newer definitions, such as when calling user-generated code, for example.

\cbstart
A number of patterns commonly appear in Julia packages to work with the
restrictions imposed by world age when combined with \eval. I cover a few of them here.

\paragraph{Boilerplating.}

The most common use of \eval is to automatically generate code for
boilerplate functions. These generated functions are typically created at the top-level so that
they can be used by the rest of the program. Consider the
\c{DualNumbers.jl} package, which provides a common dual number representation
for automatic differentiation. A dual number, which is a pair of the normal value and an ``epsilon'', which represents the derivative of the value, should support the same
operations as any number does and mostly defers to the standard operations.
For example, the \c{real} function, which gets the real component of a number
when applied to a dual number should recurse into both the actual and epsilon
value. \Eval can generate all of the needed implementations at package
load time (\c{@eval} is a macro that passes its argument to \eval as an AST).
\begin{lstlisting}
  for op in (:real, :imag, :conj, :float, :complex)
     @eval Base.$op(z::Dual) = Dual($op(value(z)), $op(epsilon(z)))
  end
\end{lstlisting}
A common sub-pattern is to generate proxies for interfaces defined by an
external system. For this purpose, the \c{CxxWrap.jl} library uses \eval at the
top level to generate (with the aid of a helper method that generates the ASTs)
proxies for arbitrary C++ libraries.
\begin{lstlisting}
  eval(build_function_expression(func, funcidx, julia_mod))
\end{lstlisting}

\paragraph{Defensive callbacks.}

The most widely used pattern for \c{invokelatest} deals with function values of
unknown age.  For example, when invoking a callback provided by a client, a
library may protect itself against the case where the provided function was
defined after the library was loaded. There are two forms of this pattern.
The simplest uses \c{invokelatest} for all callbacks, such as the library
\c{Symata.jl}:
\begin{lstlisting}
  for hook in preexecute_hooks
    invokelatest(hook)
  end
\end{lstlisting}
Every hook in \c{preexecute_hooks} is protected against world-age errors (at
the cost of slower function calls).  To avoid this slowdown, the second
common pattern catches world-age exceptions and falls back to \c{invokelatest} such as in
from the \c{Genie.jl} web server:
\begin{lstlisting}
  fr::String = try
    f()::String
  catch
    Base.invokelatest(f)::String
  end
\end{lstlisting}
This may cause surprises, however. If a sufficiently old method exists, the
call may succeed but invoke the wrong method.\footnote{In Julia,
  higher-order functions are passed by name as generic functions,
  so a callback will be subject to multiple dispatch.}
  This pattern may also catch unwanted exceptions and
execute \c f twice, including its side-effects.

\paragraph{Domain-specific generation}

As a language targeting scientific computing, Julia has a large number of
packages that do various symbolic domain reasoning.  Examples include
symbolic math libraries, such as \c{Symata} and \c{GAP}, which have the
functionality to generate executable code for symbolic expressions.
\c{Symata} provides the following method to convert an internal expression
(a \c{Mxpr}) into a callable function. Here, \c{Symata} uses a translation
function \c{mxpr_to_expr} to convert the \c{Symata} \c{mxpr} into a Julia
\c{Expr}, then wraps it in a function definition (written using explicit AST
forms), before passing it to \eval.

\begin{lstlisting}
  function Compile(a::Mxpr{:List}, body)
    aux = MtoECompile()
    jexpr = Expr(:function,
                 Expr(:tuple, [mxpr_to_expr(x, aux) for x in margs(a)]...),
                 mxpr_to_expr(body, aux))
    Core.eval(Main, jexpr)
  end
\end{lstlisting}

\paragraph{Bottleneck}

Generated code is commonly used in Julia as a way to mediate between a high-level 
DSL and a numerical library. Compilation from the DSL to executable
code can dramatically improve efficiency while still retaining a high-level
representation. However, functions generated thusly cannot be called from the
code that generated them, since they are too new. Furthermore, this code is expected
to be high-performance, so using \c{invokelatest} for every call is not acceptable. The
bottleneck pattern overcomes these issues. The idea is to split the program
into two parts: one that generates code, and another that runs it. The two
parts are bridged with a single \c{invokelatest} call (the ``bottleneck''),
allowing the second part to call the generated code efficiently. The pattern
is used in the \c{DiffEqBase} library, part of the DifferentialEquations family
of libraries that provides numerical differential equation solvers.

\begin{lstlisting}
  if hasfield(typeof(_prob),:f) && hasfield(typeof(_prob.f),:f) &&
       typeof(_prob.f.f) <: EvalFunc
    Base.invokelatest(__solve,_prob,args...; kwargs...)
  else
    __solve(_prob,args...;kwargs...)
  end
\end{lstlisting}
Here, if \c{_prob} has a field \c f, which has another field \c f, and the
type of said inner-inner \c f is an \c{EvalFunc} (an internally-defined
wrapper around any function that was generated with \eval), then it will
invoke the \c{__solve} function using \c{invokelatest}, thus allowing \c{__solve}
to call said method. Otherwise, it will do the invocation normally.

\paragraph{Superfluous eval}

This is a rare anti-pattern, probably indicating
a misunderstanding of world age by some Julia programmers. For example,
\code{Alpine.jl} package has the following call to \eval:
\begin{lstlisting}
  if isa(m.disc_var_pick, Function)
    eval(m.disc_var_pick)(m)
\end{lstlisting}
Here, \code{eval(m.disc_var_pick)} does nothing useful but imposes a performance
overhead. Because \c{m.disc_var_pick} is already a function value, calling
\eval on it is similar to using \c{eval(42)} instead of \c{42} directly;
this neither bypasses the world age nor even interprets an AST.

\paragraph{Name-based dispatch}
Another anti-pattern uses \eval to convert function names to functions.
For example, \c{ClassImbalance.jl} package chooses a function to call,
using its uninterpreted name:
\begin{lstlisting}
  func = (labeltype == :majority) ? :argmax : :argmin
  indx = eval(func)(counts)
\end{lstlisting}
It would be more efficient to operate with function values directly,
i.e. \c{func = ... : argmin} and then call it with \c{func(counts)}.
Similarly, when a symbol being looked up is generated dynamically, as it is
in the following example from \c{TextAnalysis.jl}, the use of \eval could be
avoided.
\begin{lstlisting}
  newscheme = uppercase(newscheme)
  if !in(newscheme, available_schemes) ...
  newscheme = eval(Symbol(newscheme))()
\end{lstlisting}
This pattern could be replaced with a call \lstinline{getfield(TextAnalysis, Symbol(newscheme))},
where \lstinline{getfield} is a special built-in function that finds a value in the
environment by its name. Using \lstinline{getfield} would be more efficient than \eval.
\cbend

The goal of world age was to nail down what methods are visible to any given 
part of the program for giving a consistent semantics to compilation. However,
its utility is not limited to compilation: I can also use it to solve the key
problem for a gradual type system for a language like Julia with open multiple
disptach. Concequently, I will return to world age in more detail later when
I discuss the design of the type system.

\subsection{Discussion}

The design of Julia makes a number of compromises, and I discuss some of the implications here.

\noindent\begin{minipage}{\textwidth}
\begin{wrapfigure}{r}{6.4cm}
\begin{lstlisting}[linewidth=6.3cm,aboveskip=0cm]
 abstract type AbsPt end
 struct Pt <: AbsPt
     x::Int
     y::Int
 end

 abstract type AbsColPt <: AbsPt end
 struct ColPt <: AbsColPt
     x::Int
     y::Int
     c::String
 end

 copy(p::Pt, dx, dy) = 
    Pt(p.x+dx, p.y+dy)
 copy(p::ColPt, dx, dy) =
    ColPt(p.x+dx, p.y+dy, p.c)

 move(p::AbsPt, dx, dy) = 
    copy(p, dx, dy)
\end{lstlisting}
\end{wrapfigure}

\paragraph{Object oriented programming.}

Julia's design does not support the class-based object oriented
programming style familiar from Java. Julia lacks the encapsulation
that is the default in languages going back all the way to Smalltalk:
all fields of a struct are public and can be accessed
freely. Moreover, there is no way to extend a point class \c{Pt} with
a color field as one would in Java; in Julia the user must plan ahead
for extension and provide a class \c{AbsPt}. Each ``class'' in that
programming style is a pair of an abstract and a concrete class. One
can define methods that work on abstract classes such as the \c{move}
method which takes any point and new coordinates. The \c{copy} methods
are specific to each concrete ``class'' as they must create instances.
The unfortunate side effect of the fact that abstract classes have
neither fields nor methods is that there is no documentation to remind
the programmer that a \c{copy} method is needed for \c{ColPt}. This
has to be discovered by inspection of the code.
\end{minipage}

\paragraph{Functional programming.}

Julia supports several functional programming idioms---higher order
functions, immutable-by-default values---but has no arrow types. Instead,
the language ascribes incomparable nominal types to functions.  Thus, many
traditional typed idioms are impractical, and it is impossible to dispatch
on function types.  However, nominal types do allow dispatch on methods
passed as arguments, enabling a different set of patterns. For example, the
implementation of \c{reduce} delegates to a special-purpose function
\c{reduce_empty} which, given a function and list type, determines the value
corresponding to the empty list. If reducing with \c{+}, the natural empty
reduction value is 0, for the correct 0. Capturing this, \c{reduce_empty}
has the following definition: \code{reduce_empty(::typeof(+), T)} \c{=}
\c{zero(T)}. In this case, \c{reduce_empty} dispatches the nominal \c{+}
function type, then returns the zero element for \c{T}.

\paragraph{Gradual typing.} The goal of gradual type systems is to allow
dynamically typed programs to be extended with type annotations after the
fact~\cite{SiekTaha06,tf-dls06}. Julia's type system superficially appears to
fit the bill; programs can start untyped, and, step by step, end up fully
decorated with type annotations. But there is a fundamental difference. In a
gradually typed language, a call to a function \c{f(t::T)}, such as \c{f(x)},
will be statically checked to ensure that the variable \c{x}'s declared type
matches the argument's type \c{T}. In Julia, on the other hand, a call to
\c{f(x)} will not be checked statically; if \c{x} does not have type \c{T}, 
then Julia throws a runtime error.

Another difference is that, in Julia, a variable, parameter, or field
annotated with type \c{T} will \emph{always} hold a value of type \c{T}.
Gradual type systems only guarantee that values will act like type \c{T},
wrapping untyped values with contracts to ensure they they are
indistinguishable~\cite{tf-popl08}. If a gradually-typed program manipulates a
value erroneously, that error will be flagged and blame will be assigned to
the part of the program that failed to respect the declared types. Similarly,
Julia departs from optional type systems, like Hack \cite{hack13} or
Typescript \cite{typescript13}. These optional type systems provide no
guarantee whatsoever about what values a variable of type \c{T} actually
holds.

Julia is closest in spirit to Thorn~\cite{oopsla09}. The two languages
share a nominal subtype system with tag checks on field assignment and method
calls. In both systems, a variable of type \c{T} will only ever have values of
type \c{T}. However, Julia differs substantially from Thorn, as it lacks a
static type system and adds multiple dispatch.

\section{Implementing Julia}\label{imp}
\noindent\begin{minipage}{\textwidth}
\begin{wrapfigure}{r}{.4\textwidth}
\includegraphics[width=.3\textwidth]{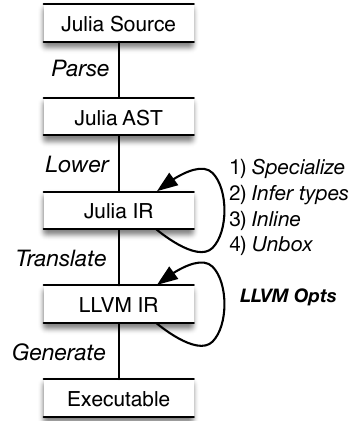}
\caption{Julia JIT compiler}\label{jit}
\end{wrapfigure}

Julia is engineered to generate efficient native code at run-time.  The
Julia compiler is an optimizing just-in-time compiler structured in three
phases: source code is first parsed into abstract syntax trees; those trees
are then lowered into an intermediate representation that is used for Julia
level optimizations; once those optimizations are complete, the code is
translated into LLVM IR and machine code is generated by LLVM~\cite{LLVM}.
Fig.~\ref{jit} is a high level overview of the compiler pipeline.
\end{minipage}

With the exception of the standard library which is pre-compiled, all Julia
code executed by a running program is compiled on demand. The compiler is
relatively simple: it is a method-based JIT without compilation tiers; once
methods are compiled they are not changed as Julia does not support
deoptimization with on-stack replacement.

Memory is managed by a stop-the-world, non-moving, mark-and-sweep garbage
collector. The mark phase can be executed in parallel. The collector has a
single old generation for objects that survive a number of cycles. It uses a
shadow stack to record pointers in order to be precise.

\noindent\begin{minipage}{\textwidth}
\begin{wrapfigure}{r}{.32\textwidth}
\begin{minipage}{.3\textwidth}\small
\caption{Source files}
\begin{tabular}{@{}lrr@{}}\hline
Language   &files &           code\\\hline
Julia       &                   296&  115,252\\
C            &                   79&  44,930\\
C++           &                  21&  18,491\\
Scheme         &                 17&   8,270\\
C/C++ Header    &                44&    6,205\\
make              &               7&     684\\
Bourne Shell       &              2&      85\\
Assembly            &             4&     74\\\hline
                     &          470& 193,991\\\hline
\end{tabular}\label{cloc}\end{minipage}
\end{wrapfigure}

Since v0.5, Julia natively supports multi-threading but the feature is still
labeled as ``experimental''. Parallel loops use the \c{Threads.@threads}
macro which annotates \c{for} loops that are to run in a multi-threaded
region. Other part of the multi-threaded API are still in flux. An
alternative to Julia native threading is the ParallelAccelerator system
of~\cite{ecoop17} which generates OpenMP code on the fly for parallel
kernels. The system crucially depends on type stability---code that is not
type stable will execute single threaded.
Fig.~\ref{cloc} gives an overview of the implementation of Julia v0.6.2. The
standard library, Core, Base and a few other modules, accounts for most of
the use of Julia in Julia's implementation. The middle-end is written in C
and Julia;
C++ is used for LLVM IR code generation.  Finally, Scheme and Lisp are used
for the front end. External dependencies such as LLVM, which is used as back
end, do not participate to this figure.
\end{minipage}

\subsection{Method specialization}

Julia's compilation strategy is built on runtime type information.  Every
time a method is called with a new tuple of argument types, it is
specialized to these types. Optimizing methods at invocation time, rather
than ahead of time, provides the JIT with key pieces of information: the
memory layout of all arguments is known, allowing for unboxing and direct
field access. Specialization, in turn, allows for devirtualization.
Devirtualization replaces method dispatch with direct calls to a specialized
method. This reduces dispatch overhead and enables inlining. As the
compilation process is rather slow, results are cached, thus methods are
only compiled the first time they are called with a new type.  This process
converges as long as functions are only called with a limited number of
types.
If a function gets called with many different argument types, then
invocations will repeatedly incur the cost of specialization. Julia cannot
avoid this pathology, as programs that generate a large number of call
signatures are easy to write. To alleviate this problem, Julia allows tuple
types to contain a \c{Vararg} component, which is treated as having type
\c{Any}.  Likewise, each function value has its own type, but Julia only
specializes on function types if the argument is called in the method body.
Other heuristics are used for type \c{Type}.  Julia has one recourse against
type unstable code, programmers can use the \c{@nospecialize} annotation to
prevent specialization on a specific argument.

\subsection{Type inference}

Type information enables many of Julia's key optimizations. The compiler
performs a data-flow analysis to discover types after specializing.  Julia
uses a set constraint-based analysis with constraints arising from return
values, method dereferences, and argument types. Type requirements need to
be satisfied at function call sites and field assignments. The system
propagates constraints forward to satisfy requirements, inferring the types
for intermediate values along the way.

\noindent\begin{minipage}{\textwidth}
\begin{wrapfigure}[10]{r}{8cm}
\begin{minipage}{3.4cm}
\begin{lstlisting}[linewidth=3.4cm]
function f(a,b)
    c = a+b
    d = c/2.0
    return d
end
\end{lstlisting}
\end{minipage}
\begin{minipage}{4.5cm}
\begin{lstlisting}[linewidth=5.5cm]
function f(a::Int,b::Int)
    c = a+b::Int
    d = c/2.0::Float64
    return d
end => Float64
\end{lstlisting}
\end{minipage}
\caption{A simple example of type inference} \label{simpinf}
\end{wrapfigure}

Given the concrete types of all function arguments, intraprocedural type
inference propagates types forward into the method body. An example is shown
in Fig.~\ref{simpinf}.  When \c{f} is called with a pair of integers, type
inference finds that \c{a+b} returns an integer; therefore \c{c} is likewise
an integer. From this, it follows that \c{d} is a float and so is the return
type of the method. Note that this explanation relies on knowing the return
type of \c{+}. Since addition could be overloaded, it is necessary to be
able to infer the return types of arbitrary methods. Return types may vary
depending on argument type, and previous inference results may not cover the
current case. Therefore, when a new function is called, analysis of the
caller must be suspended and continue on the callee to figure out the return
type of the call.
\end{minipage}

\noindent\begin{minipage}{\textwidth}
\begin{wrapfigure}{r}{9cm}
\begin{minipage}{4cm}
\begin{lstlisting}[linewidth=4cm]
function a()
    return b(3)+1
end
function b(num)
    return num+2
end
\end{lstlisting}
\end{minipage}
\begin{minipage}{5cm}
\begin{lstlisting}[linewidth=5cm]
function a()
    return b(3)+1::Int
end => Int
function b(num::Int)
    return num+2::Int
end => Int
\end{lstlisting}
\end{minipage}
\caption{Simple interprocedural type inference} \label{simpininf}
\end{wrapfigure}

Interprocedural analysis is simple for non-recursive methods as seen in
Fig.~\ref{simpininf}: analysis proceeds with the called method and the
return type is computed. For recursive methods cycle elimination is
performed. Once a cycle is identified, it is executed until it reaches
convergence. The cycle is then contracted into a single monolithic function
from the perspective of analysis. More challenging are methods whose
argument or return types can grow indefinitely depending on its
arguments. To avoid this, Julia limits the size of the inferred types to an
arbitrary bound. In this manner, the set of possible types is finite and
therefore termination of the analysis is guaranteed.
\end{minipage}

\subsection{Method inlining}

Inlining replaces a function call by the body of the called function. In
Julia, it can be realized in a very efficient way because of its synergy
with specialization and type inference. Indeed, if the body of a method is
type stable, then the internal calls can be inlined. Conversely, inlining can
help type inference because it gives additional context. For instance, inlined
code can avoid branches that can be eliminated as dead code, which allows in
turn to propagate more precise type information. Yet, the memory cost
incurred by inlining can be sometimes prohibitive; moreover it requires
additional compilation time. As a consequence, inlining is bounded by a
number of pragmatic heuristics.

\subsection{Object unboxing}

Since Julia is dynamic, a variable may hold values of many types. As a
consequence, in the general case, values are allocated on the heap with a
tag that specifies their type. Unboxing allows to manipulate values
directly. This optimization is helped by a combination of design
choices. First, since concrete types are final, a concrete type specifies
both the size of a value and its layout. This would not be the case in Java
or TypeScript due to subtyping. In addition, Julia does not have a null
value; if it did, there would be need for an extra tag for primitive
values. As a consequence, values such as integers and floats can always be
stored unboxed.  Repeated boxing and unboxing can be expensive, and
unboxing can also be impossible to realize although the type information is
present, in particular for recursive data structures. As with inlining,
heuristics are thus used to determine when to perform this optimization.

\begin{figure}
\centering
\includegraphics[width=.8\textwidth]{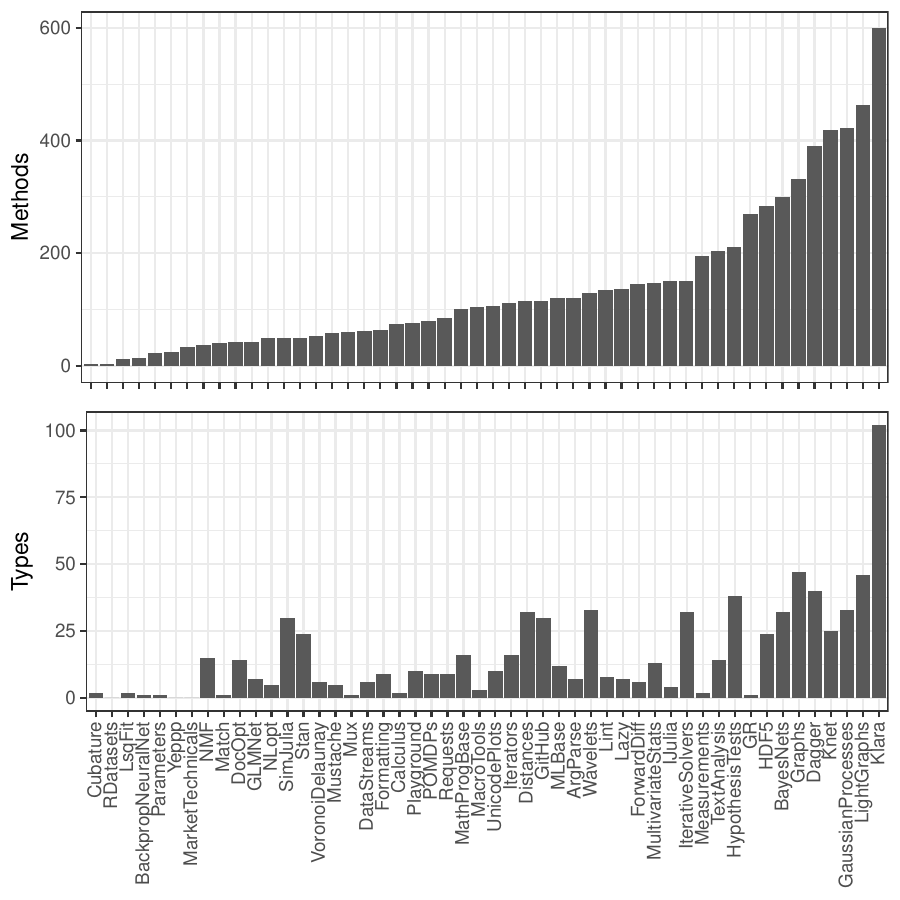}
\caption{Number of methods and types by package}\label{numfuncs}
\end{figure}

\section{Julia in Practice}\label{pra}

In order to understand how programmers use the language, I analyzed a corpus
of 50 packages hosted on GitHub. I chose packages---libraries, in Julia
parlance---over runnable end-user programs out of necessity: no central
repository exists of Julia programs. Packages were
included based on GitHub stars. Selected packages also had to pass their own
test suites.  Additionally, I analyzed Julia's standard library.

\subsection{Typeful programming}

Julia is a language where types are optional. Yet, knowing them is
profitable since it enables major optimizations. 
Users are thus encouraged
to program in a typeful style where code is, as much as possible, type
stable. To what extent is this rule followed?

\begin{figure}
  \centering
  \hfill
  \subfloat[Methods by percentage of typed arguments]{
    \label{perctyped}
      \includegraphics[scale=0.5]{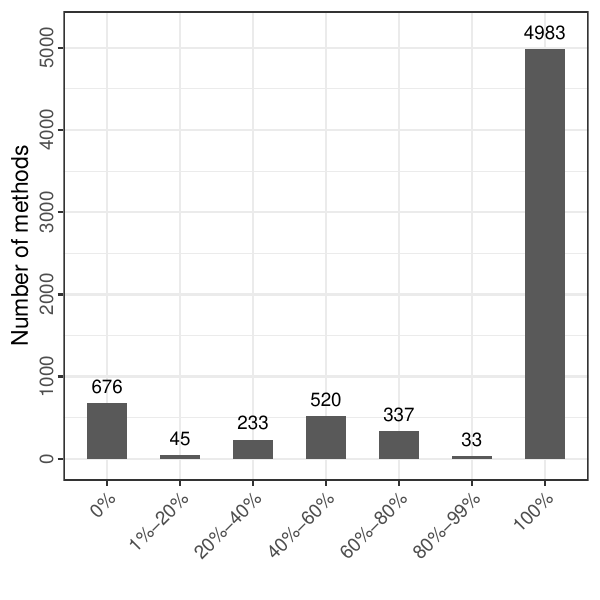}
      \vspace{-3mm}
  }
  \hfill
  \subfloat[Targets per call site per package]{
    \label{targets_per_call_site_per_package}
      \includegraphics[scale=0.5]{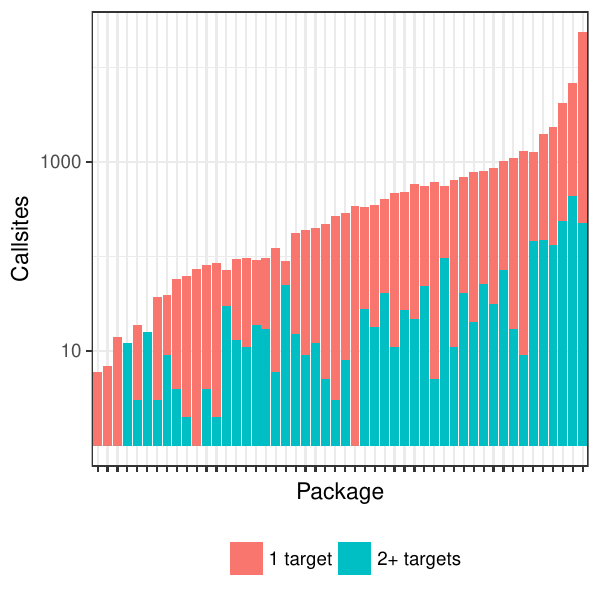}
      \vspace{-3mm}
  }
  \hfill
  \caption{High-level multimethod usage metrics}
\end{figure}

\subsubsection{Type annotations}

Fig.~\ref{numfuncs} gives the number of methods and types defined in each
package after it was loaded into Julia, to ensure that generated methods were
counted. I performed structural analysis of parsed ASTs, allowing us to
measure only methods and types written by human developers. In total, the
corpus includes 792 type definitions and 7,018 methods. The median number of
types and methods per package is 9 and 104, respectively. Klara, a library
for Markov chain Monte Carlo inference, is the largest package by both
number of types and methods with 102 and 599, respectively. Three packages,
MarketTechnicals, RDatasets, and Yeppp, define zero types; while Cubature
defines just 3 methods, the fewest in the corpus.  Clearly, Julia users
define many types and functions. However, the level of dynamism remains a
question.
Fig.~\ref{perctyped} shows the distribution of type annotations on arguments
of method definitions. 0\% means all arguments are untyped (\c{Any}),
while 100\% means that all arguments are non-\c{Any}.  An impressive 4,983
(or 62\%) of methods are fully type-annotated.  
Despite having the opportunity to write untyped methods, developers
define mostly typed methods.

\subsubsection{Type stability}

\noindent\begin{minipage}{\textwidth}
\begin{wrapfigure}{r}{0.5\textwidth}
\end{wrapfigure}

Type stability is key to devirtualizing and inlining methods.  I measure
type instability at run-time by dynamic analysis of the test suites of our
corpus. Each called method was recorded along with the tuple of types of its
arguments and the call site. I filtered calls to anonymous and
compiler-generated functions to focus on functions defined by humans.
Fig.~\ref{targets_per_call_site_per_package} compares, for each package, the
number of call sites where all the calls target only one specialized
method to those that call two and more. Calls are recorded regardless
of whether they were devirtualized. The y-axis is shown in log scale. On
average, 92\% of call sites target a single specialized method.  Code is
thus in general type stable, which agrees with the assumption that
programmers attempt to write type stable code.
\end{minipage}

\subsection{Multiple dispatch}

Multiple dispatch is the most prominent features of Julia's design.  Its
synergy with specialization is crucial to understand the performance of the
language and its ability to devirtualize and inline efficiently.  How is
multiple dispatch used from a programmer's perspective?  Moreover, a promise
of multiple dispatch is that it can be used to extend existing behavior with
new implementations. How much do Julia libraries extend existing
functionality, and what functionality do they extend?

\subsubsection{Overloading}

Fig.~\ref{fnhist} examines how multiple dispatch is used to extend existing
functionality. I use the term \emph{external overloading} to mean that a
package adds a method to a function defined in a library.  Packages are
binned based on the percentage of functions that they overload versus
define. Packages with only external overloading are at 100\%, while packages
that do not use external overloading would be in the 0\% bin.  Many packages
are defined without extensive use of external overloading.  For 28 out of 50
packages, fewer than 30\% of the functions they define are
overloads. However, the distribution of overloading has a long tail, with a
few libraries relying on overloads heavily.  The Measurements package has
the highest proportion of overloads, with 147 overloads out of a total of
161 methods (91\%). This is justified by the purpose of Measurements: it
propagates errors throughout other operations, which is done by extending
existing functions.

\begin{figure}
  \centering
  \subfloat[Packages by \% of overloaded functions]{
    \label{fnhist}
      \includegraphics[width=0.75\textwidth]{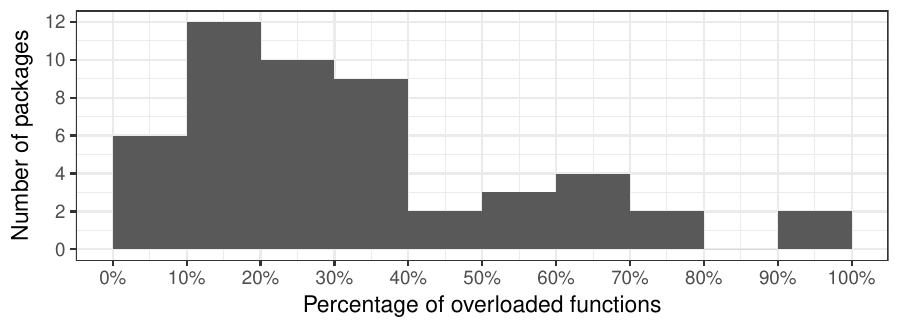}
      \vspace{-3mm}
  }

  \noindent\subfloat[Function overloads by category]{
    \label{fnover}
      \includegraphics[width=0.75\textwidth]{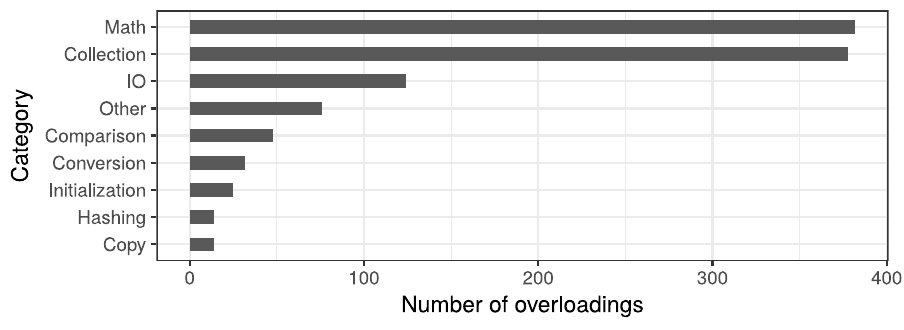}
      \vspace{-3mm}
  }
  \caption{Function overloadings}
\end{figure}

To address the question of what are overloaded, I manually categorized the
top 20th quantile of overloaded functions (128 out of 641) into 9 groups.
Fig.~\ref{fnover} depicts how many times functions from each group is
overloaded.  Multiple dispatch is used heavily to overload mathematical
operators, like addition or trigonometric functions. Libraries overload
existing operators to work with their own types, providing natural
interfaces and interoperability with existing code. Examples include
Calculus, which overloads arithmetic to allow symbolic expressions; and
ForwardDiff, which can compute numerical derivatives of existing code using
dual numbers that act just like normal values.  Collection functions also
are widely overloaded.  Many libraries have collection-like objects, and by
overloading these methods they can use their collections where Julia expects
any abstract collection.  However, Julia's interfaces are only defined by
documentation, as a result of its dynamic design. The \c{AbstractArray}
interface can be extended by any struct, and it is only suggested in the
documentation that implementations should overload the appropriate methods.
Use cases for math and collection extension are easy to come by, so their
prevalence is expected. However, the lack of overloads in other
categories illustrates some surprising points. For example, the large number
of IO, math, and collection overloads (which implement variations on
\c{tostring}) suggest a preponderance of new types.  However, few overloads
to compare, convert, or copy  are provided.

\begin{figure}[H]
  \vspace{1.5em}
\begin{tabular}{ll}
\begin{minipage}{.5\textwidth}
\includegraphics[width=0.9\columnwidth]{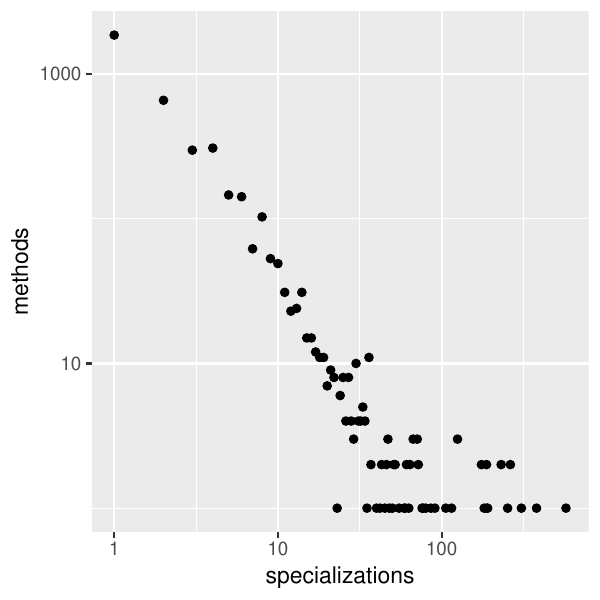}
\caption{Number of specializations per method}
\label{fig:SPM}
\end{minipage}&\begin{minipage}{.5\columnwidth}
\includegraphics[width=0.9\textwidth]{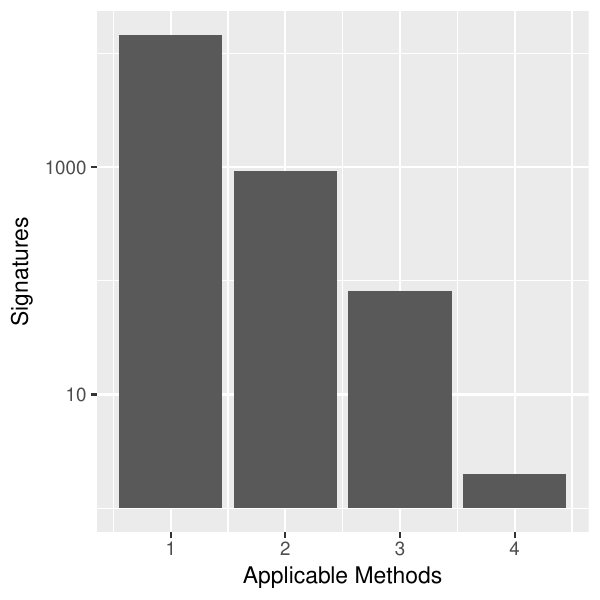}
\caption{Applicable methods per call signature}
\label{fig:AMCS}
\end{minipage}\end{tabular}
\end{figure}

\subsection{Specializations}

Figure~\ref{fig:SPM} gives the number of specializations per method recorded
dynamically on our corpus. The data uses strict eliminations, so that the
results from different packages can be summed without duplicate functions.
The distribution has a heavy tail, which shows that programmers actually
write methods that can be very polymorphic. Note that polymorphism is not in
contradiction with type stability, since a method called with different
tuples of argument types across different call sites can be type stable for
each of its call sites.
Conversely, 46\% of the methods have only been specialized once after
running the tests. Many methods are thus used monomorphically: this hints that
a number of methods may have a type specification that prevent polymorphism,
which means that programmers tend to think of the concrete types they want
their methods applied to, rather than only an abstract type specification.

Figure~\ref{fig:AMCS} corroborates this hypothesis. It represents the number
of applicable methods per call signature. A method is applicable if the
tuple of types corresponding to the requirements for its arguments is a
supertype of that of the actual call. This data is collected on dynamic
traces for functions with at least two methods. 93\% of the signatures can
only dispatch to one method, which strongly suggests that methods tend to be
written for disjoint type signatures. As a consequence it shows that the
specificity rules, used to determine which method to call, boil down to
subtyping in the vast majority of cases.

\subsection{Impact on performance}

\begin{figure}[!t]
\includegraphics[width=1\textwidth]{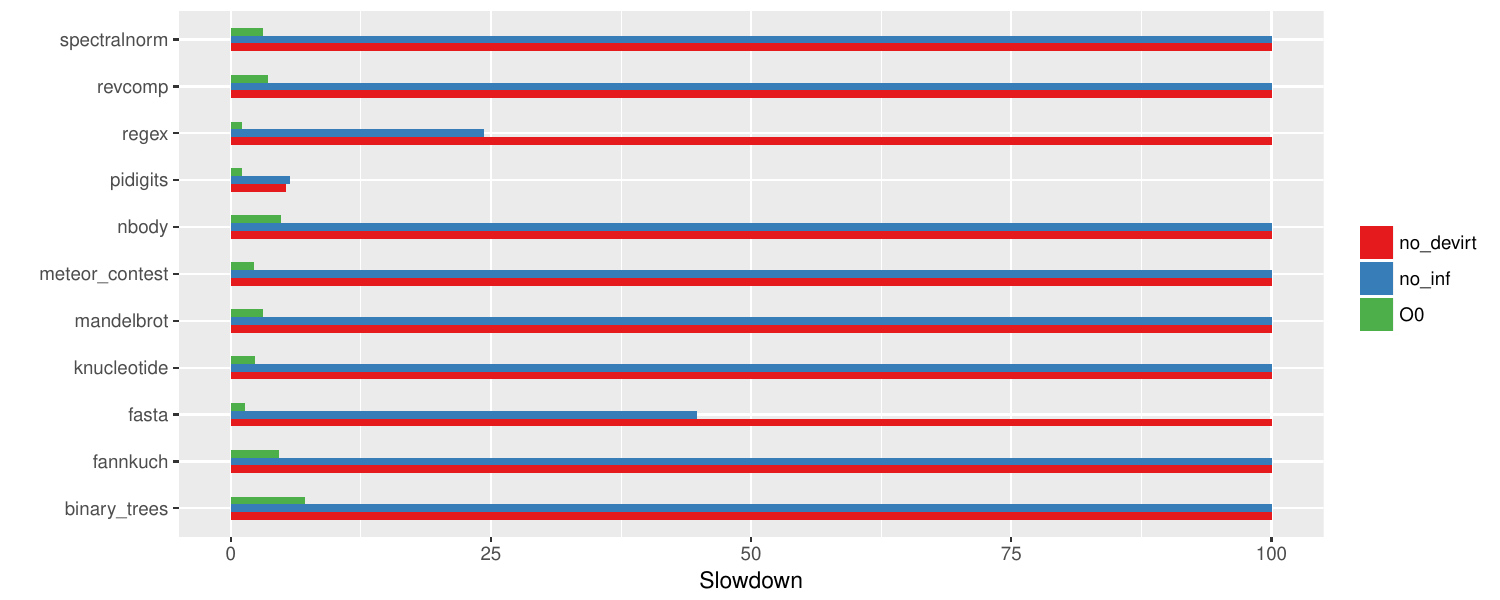}
\caption{Optimization and performance}\label{perfer}
\end{figure}

Fig.~\ref{perfer} illustrates the impact on performance of LLVM
optimizations, type inference and devirtualization. By default Julia uses
LLVM at optimization level \code{O2}. Switching off all LLVM optimizations
generates code between 1.1x and 7.1x slower.  Turning off
type inference means that method are specialized correctly but all
internal operations will be performed on values of type \code{Any}.
Functions that have only a single method may still be devirtualized and
dispatched to. The graph is capped at 100x slowdown. The actual
slowdowns range between 5.6x and 2151x. Lastly, turning off devirtualization
implies that no inlining will be performed and all function calls are
dispatched dynamically. The slowdowns range between 5.3x and 1905x.

Obviously, Julia was designed to be optimized with type information. These
results suggest that performance of fully dynamic code is rather bad.  It is
likely that if users were to write more dynamic code, some of the techniques
that have proved successful for other dynamic languages could be ported to
Julia. But clearly, the current implementation crucially relies on code
being type stable and on devirtualization and inlining. The impact of the
LLVM optimizations is small in comparison.

\chapter{Subtyping in Julia}\label{ch:juliasub} 

Julia's key relation over types is subtyping. Every method invocation is
resolved using subtyping both to determine which methods could apply as well
as to figure out which is the most specific. Consequently, subtyping is
critical both for the semantics of Julia itself and for programmers reasoning
about Julia. 

Unfortunately, while I have previously formalized the \emph{algorithm} that
Julia uses~\cite{oopsla18b} the system has proven theoretically intractable.
As will be shown later in this chapter Julia's subtyping relation is very
complex and proving meaningful properties about it has proven elusive. Subtyping
in Julia is supposed to be based on nominal subtyping---that if the set of
values that type $A$ describes is a subset of the set of values of type $B$
then $A$ is a subtype of $B$ and vice versa. Proving even this simple fact
for a practical subset of Julia is challenging.

As I were wrestling with this complexity a question came up: if subtyping is
this complex then is it even decidable? Julia has bounded quantified types
and decidable subtyping in such a setting would be the exception rather than
the rule~\cite{Grigore:2017:JGT:3009837.3009871,hu2019decidability}. As it turns
out, subtyping in Julia is also provably undecidable; the remainder of this section
will be devoted to proving that fact and discussing how I then accommodate 
this subtyping relation.

\section{Related Work}

Subtyping is key for a language with multiple dispatch.
Subtyping is used to decide which methods might be called at a given site
or whether a given invocation is safe or not. Static typing for a lanugage with
multiple dispatch must then rely extensively on the subtyping relation. Moreover,
the decidability of subtyping then determines whether the type system as a whole
is decidable.

Parametric polymorphism is the usual pain point; it is easy to create type
languages for which subtyping is very hard or impossible to decide with
parametric polymorphism. Languages with multiple dispatch differ on whether
parametric polymorphism is supported or not. Most previous efforts focused on
non-polymorphic types, such as Cecil~\cite{Chambers94}, Typed
Clojure~\cite{Tobin16}, and MultiJava~\cite{clifton2000multijava}. Subtyping
is used to check that classes implement all of the required methods of their
supertypes. The subtype relations themselves are defined over covariant tuples
and discrete unions. Approaches that combine multiple dispatch with parametric
polymorphism are more involved. 

Mini-Cecil~\cite{Litvinov98,Litvinov03} is one example of a language with both
universal polymorphism and parametric multiple dispatch. In Mini-Cecil,
universal types have only top-level quantifiers.
Fortress~\cite{Allen11}, in addition, supports arrow types, and
internally uses both universal and existential types with top-level
quantifiers.
Mini-Cecil and Fortress both use a constraint generation strategy to resolve
subtyping; they support union and intersection types but do not provide
distributivity ``in order to simplify constraint
solving''~\cite{Litvinov03}.  For Mini-Cecil typechecking is
decidable. Fortress argued decidability based on~\cite{Castagna94}, though
no proof is provided.

In type systems with bounded existential types, as well as type systems with
nominal subtyping and variance, decidability of subtyping has been a major
concern~\cite{Kennedy07,10.1007/978-3-642-10672-9_10}. Pierce demonstrated
that even a small subset of System $F_{\leq}$ is undecidable~\cite{Pierce92},
demonstrating how easy it is to accidentally be undecidable even with very
simple polymorphic type systems.

The decidability of subtyping in practical languages has been extensively
studied. For example, Java was shown to be undecidable by
Grigore~\cite{Grigore:2017:JGT:3009837.3009871} and Scala's current type
system is undeicdable~\cite{hu2019decidability}. It is relatively easy to
accidentally introduce an undecidability into subtyping.

Usually undecidability is not a practical problem for a language. If
undecidability arises from features that are uncommonly used, as seen in
Java~\cite{tate2011taming}, then few programmers are likely to run into
programs that fail to compile. Moreover, when undecidability is readily
accessible it then became part of the practice of software engineering for the
language in question, as seen in C++~\cite{bagnara2022coding}. Exploitation of
undecidability is also possible: theoretical results about Java have produced
encodings of increasingly complex language grammars~\cite{gil2019fling} into
the type system. Undecidabilities are managable in practice so long as either
the execution model is easy to understand or the critical features or patterns
are not typically used.

\section{Formalization of Subtyping}
\renewcommand{\t}{\ensuremath{\tau}\xspace}
\newcommand{\tytext}[1]{\ensuremath{\text{\texttt{#1}}}}
\newcommand{\tupletyp}[1]{\tytext{Tuple}\{#1\}}
\newcommand{\uniontyp}[1]{\tytext{Union}\{#1\}}
\newcommand{\cstrt}[2]{\tytext{#1}\{#2\}}
\newcommand{\wheret}[2]{\ensuremath{#1\;\tytext{where}\;#2}}

Julia's type language is deceptively simple. The language has a number
of basic type forms including:
\begin{itemize}
  \item tag types such as \jlinl{Int} or \jlinl{Rational} which are simple inhabitants of an explicit subtyping lattice and can be parameterized such as \jlinl{Vector\{Int\}},
  \item tuple types \jlinl{Tuple\{Int, Int\}},
  \item untagged union types \jlinl{Union\{Int, String\}},
  \item and bounded existential types \jlinl{Vector\{T\} where Nothing <: T <: Int}
\end{itemize}

The first three structures are relatively straightforward; the addition of
bounded existential types makes the subtyping relationship much more
complicated.

I will not be describing Julia's subtyping relation in depth in this work;
instead I will be relying on the formalization of Belyakova et.
al.~\cite{yuliasubtyping}. This formalism captures the overwhelming majority
of subtyping for Julia's type language (primarily excluding variadic length
tuples). In order to describe the problem with subtyping in Julia I will be using
an excerpt of that formalism, shown in figure~\ref{fig:jlsubex}. 

I will only describe the part of the formalism required to understand the
proof of undecidability. Judgments are of the form $E \vdash \t <: \t' \vdash
E'$, which should be read as ``$\t$ is a subtype of $\t'$ against the
environment $E$ producing $E'$.''

\begin{figure}
\begin{mathpar}
\inferrule[Refl]
    {  }
    { E \vdash \t <: \t \vdash E }

\inferrule[Union]
    {  }
    { E \vdash \uniontyp{} <: \t \vdash E }

\inferrule[Tuple]
    { E \vdash a_1 <: a'_1 \vdash E_1 \\ \ldots \\ E_{n-1} \vdash a_n <: a'_n \vdash E_n}
    { E \vdash \tupletyp{a_1, \ldots, a_n} <: \tupletyp{a'_1, \ldots, a'_n} \vdash E_n }

\inferrule[App\_Inv]{ n \leq m \\ E_0 = \mathit{add}(E, \text{barrier}) \\ 
    \forall\, 0 < i \leq n, \; E_{i-1} \vdash a_i <: a'_i \vdash E'_i \wedge E'_i \vdash a'_i <: a_i \vdash E_i \\
    }
    { E \vdash \cstrt{name}{a_1, \ldots, a_m} <: \cstrt{name}{a'_1, \ldots, a'_n} \vdash \mathit{del}(\text{barrier}, E_i)}

\inferrule[L\_Intro]
    { \mathit{add}(E, ^L T_{\t_1}^{\t_2}) \vdash \t <: \t' \vdash E'}
    { E \vdash (\wheret{\t}{\t_1 <: T <: \t_2}) <: \t' \vdash \mathit{del}(T, E') }

\inferrule[R\_Intro]
    { \mathit{add}(E, ^R T_{\t_1}^{\t_2}) \vdash \t <: \t' \vdash E' \\
      \mathit{consistent}(T, E') }
    { E \vdash \t <: (\wheret{\t'}{\t_1 <: T <: \t_2}) \vdash \mathit{del}(T, E') }

\inferrule[L\_Right]{
  \mathit{search}(T, E) = {^L T}_{l}^{u} \\
  E \vdash l <: \t \vdash E' }
  { E \vdash \t <: T \vdash E' }

\inferrule[L\_Left]{
  \mathit{search}(T, E) = {^L T}_{l}^{u} \\
  E \vdash u <: \t \vdash E' }
  { E \vdash T <: \t \vdash E' }

\inferrule[R\_Right]{
  \mathit{search}(T, E) = {^R T}_{l}^{u} \\
  (\mathit{is\_var}(\t) \wedge \mathit{search}(\t, E) = {^L S}_{l'}^{u'}) \implies \lnot \mathit{outside}(T, S, E) \\
  E \vdash \t <: u \vdash E' }
  { E \vdash \t <: T \vdash \mathit{upd}({^R T}_{\uniontyp{l, \t}}^{u}, E') }

\inferrule[R\_Left]{
  \mathit{search}(T, E) = {^R T}_{l}^{u} \\
  E \vdash l <: \t \vdash E' }
  { E \vdash T <: \t \vdash \mathit{upd}({^R T}_{l}^{t}, E') }

\inferrule[R\_L]{
  \mathit{search}(T_1, E) = {^R T_1}_{l_1}^{u_1} \\
  \mathit{search}(T_2, E) = {^L T_2}_{l_2}^{u_2} \\
  \mathit{outside}(T_1, T_2, E) \implies E \vdash u_2 <: l_2 \vdash E' \\
  E \vdash u_1 <: l_2 \vdash E''
}
    { E \vdash T_1 <: T_2 \vdash \mathit{upd}(^R {T_1}_{\uniontyp{T_1, l_1}}^{u_1}, E') }
\end{mathpar}
\caption{Julia subtyping (extract)}
\label{fig:jlsubex}
\end{figure}

Subtyping happens against an environment $E$ that carries the in-scope
variables and defines their variance; a variable in $E$ looks like $^S T_l^u$,
where $S$ is the side (left or right) it first appeared on, $l$ is the lower
bound, and $u$ is the upper bound. $E$ may also contain barriers (used to
indicate where I switch from a covariant to an invariant context). Left-side
variables have forall semantics, while right-side variables have existential
semantics.

The proof of undecidability relies on seven subtyping rules that work as follows:
\begin{itemize}
  \item \textsc{Refl}: Reflexivity of subtyping (every \t is a subtype of itself) is axiomatic within the formalism and
  makes no demands of or modifications to the environment $E$.
  \item \textsc{Union}: I do not reproduce the full generality of left-hand union subtyping
  as it is not required for the proof of undecidability. Instead, I simply capture that the 
  empty union is the bottom type (e.g. is a subtype of all other types \t). This follows from
  the original rule \textsc{Union\_Left}.
  \item \cbstart \textsc{Tuple}: Julia tuples have standard covariant subtyping; each
  element of the left-hand tuple $a_i$ is checked against the matching element in the right-hand tuple $a'_i$.
  The first element is checked against the original context $E$, begetting the context $E_1$, and so on, until all
  pairs of elements have ben checked. The final context $E_n$ is then returned.
  \item \textsc{App\_Inv}: In contrast, constructor applications (such as \jlinl{Vector\{T\}})
  are invariant on their arguments where tuples were covariant. The system implements this by appending
  a barrier element to the context $E$, then checking that for each pair $a_i$ and $a'_i$ that
  first $E_{i-1} \vdash a_i <: a'_i \vdash E'_i$ and then that $E'_i \vdash a'_i <: a_i \vdash E_i$.
  Checking both directions of subtyping ensures that the final environment $E_i$ requires that $a'_i$
  and $a_i$ are equal.
  \item \textsc{L\_Intro}: Left-hand addition of variables is straightforward. The system
  checks that the body of the introduction form \t is a subtype of the right-hand side \t'
  against the environment extended with $T$, $\mathit{add}(E, ^L T_{\t_1}^{\t_2})$. Once the
  subtyping relation has been established into the environment $E'$ this new variable is then
  removed before returning using $\mathit{del}$.
  \item \textsc{R\_Intro}: Right-hand addition works analogously, albeit with the introduction
  of the \textit{consistent} metafunction. \textit{consistent} simply ensures that the lower bound
  of $T$ (which may have changed in the interim) is a subtype of the upper bound in $E'$.
  \item \textsc{L\_Right}: The \textsc{R/L\_Right/Left} rules handle cases where a ``bare'' type 
  variable appears in the subtyping relationship. The L/R refers to where the variable was originally 
  introduced (on the left or right hand side, respectively), while the Right/Left refers to where the variable
  appeared.

  In the case of \textsc{L\_Right}, this was a variable that first appeared on the left-hand-side
  that is now appearing on the right hand side. Left-hand variables have forall semantics, so I need
  to ensure that all possible instantiations will be consistent with the bound implied by this subtyping
  relationship. As a result, to show that $E \vdash \t <: T \vdash E'$ I need to show that the
  lower bound of $T$ in $E$, $l$, is a subtype of \t using $E \vdash l <: \t \vdash E'$.
  \item \textsc{L\_Left}: The left-left case is then analogous to the left-right case discussed above.
  Instead of checking that the lower bound is consistent with the left-hand type expression, I instead check
  \item \textsc{R\_Right}: Right-hand type variables are existentially, rather than universally, quantified.
  As a result, I modify right-hand variables to ensure that they are consistent with the relation that I re trying
  to prove.

  Trivially, I need to ensure that the existing upper bound is consistent with this new lower bound, which is checked
  with $E \vdash \t <: u \vdash E'$. I also need to ensure that if the left-hand-side is a left-introduced variable
  (which is checked using $\mathit{is\_var}(\t)$ and $\mathit{search}(\t, E) = {^L S}_{l'}^{u'})$) that it shares the same variance
  as I do. This prevents right-hand variables that are introduced inside of a invariant context from being instantiated with lower bound begotten from covariantly-quantified variables outside of that context.

  For example, consider the statement \c{(Vector\{T\} where T<:Real) <: (Vector\{U where U<:Real\})}. Intuitively, this is a false proposition: the right-hand vector is one of heterogenous real numbers, whereas the left-hand vector is of a homogenous single (albeit unspecfied) type of real number. This is equivalent to ${^L T}^\texttt{Real}\ \mathit{barrier}\ {^R U}^\texttt{Real} \vdash \texttt{Vector}\{T\} <: \texttt{Vector}\{U\}$. If I allowed \textsc{R\_Right} to apply in this case then the judgment holds if ${^L T}^\texttt{Real}\ \mathit{barrier}\ {^R U}^\texttt{Real} \vdash \texttt{Real} <: \texttt{Real}$ (by applying \textsc{L\_Left}) which is trivially true. Thus, I cannot allow \textsc{R\_Right} to hold in invariant contexts since I cannot instantiate an invariant variable to be equal
  to a covariant one. Instead, the rule \textsc{R\_L} needs to be used.
  \item \textsc{R\_Left}:  In the right-left case, here, I have a right-hand variable $T$ on the left-hand side of a subtyping judgment. I need to ensure that the bounds on $T$ remain consistent (by checking that $E \vdash l <: \t \vdash E'$).
  If all of these are satisifed, then our resulting environment is $\mathit{upd}({^R T}_{l}^{\t}, E')$,
  or $E'$ updated with $T$ upper bounded by the new $\t$.
  \item \textsc{R\_L}: If I are checking subtyping between two variables
  that are both on the ``wrong'' side (e.g. a right-introduced variable on the left and vice versa)
  I need to perform a more complex procedure. $\mathit{outside}(T_1, T_2, E)$ is true if $T_1$ precedes
  $T_2$ in $E$ and is separated by a barrier; that is, if $T_2$ is invariant and $T_1$ is an earlier covariant
  definition. If this is the case I need to constrain $T_1$ and $T_2$ to be equal by checking that
  the upper bound of $T_2$ is a subtype of the lower bound of $T_1$. In any case, I also need to ensure
  that $u_1$ is a subtype of $l_2$.

  A brief example is in the same test I used for \textsc{R\_Right} where I are testing if a quantified vector is a subtype of a
  vector of quantified values. In this case, I need to check that the upper bound of $U$, \texttt{Real}, is a subtype of
  the (implicit) lower bound of $T$, $\uniontyp{}$. This is false and causes the result of subtyping in that example to be correct.
\end{itemize}
\cbend

\newcommand{\fsub}{$F_{\leq}$\xspace} 
\newcommand{\fsubd}{$F_{\leq}^D$\xspace} 
\newcommand{\fsubp}{$F_{\leq}^P$\xspace} 

\section{Proof of Undecidability}
Our proof of undecidability proceeds by reduction of subtyping in one of the
intermediate deterministic fragments of System \fsub as described by
Pierce~\cite{pierce1992bounded}  System \fsubp, to Julia subtyping. I do this
by translating System \fsubp judgements to Julia subtyping judgements and
vice versa. At a high level, the translation works by flipping upper-bounded
universal quantification in \fsubp to lower-bounded existential quantification
in Julia in the opposite order. I will begin by describing System \fsubp and
then our reduction from subtyping in System \fsubp to subtyping in Julia.

\subsection{System \fsubp}

System \fsubp is a restricted version of Pierce's System \fsub without
arrow types and with types that carry explicit information about whether they 
appear in negative or positive (left or right) position. I provide the grammar
and subtyping rules for System \fsubp in figure~\ref{fig:fsubp}. Following Pierce, 
$\Gamma^-(\alpha) = \t$ holds if $\alpha \leq \t \in \Gamma$.

\newcommand{\Alt}{~\vert~}
\newcommand{\Altf}{\faded{\Alt}}
\begin{figure}

  \footnotesize
  \[
  \begin{array}{ccl@{\qquad}l}
      \\ \t^+ & ::= & & \text{\emph{Positive types}}
      \\ &\Alt& \mathit{Top}       & \text{Top type}
      \\ &\Alt& \lnot \t^-  & \text{Negative negation}
      \\ &\Alt& \forall \alpha\leq\t^-.\t^+  & \text{Positive quantification}
  \end{array}
  \begin{array}{ccl@{\qquad}l}
      \\ \t^- & ::= & & \text{\emph{Negative types}}
      \\ &\Alt& \alpha       & \text{Type variable}
      \\ &\Alt& \lnot \t^+  & \text{Positive negation}
      \\ &\Alt& \forall \alpha\leq\mathit{Top}.\t^-  & \text{Negative quantification}
  \end{array}
  \]

~\\
\begin{mathpar}
\inferrule[PTop]{ }{\Gamma^- \vdash \t^- \leq \mathit{Top}}

\inferrule[PVar]{ \Gamma^- \vdash \Gamma^-(\alpha) \leq \t^+ }{\Gamma^- \vdash \alpha \leq \t^+}

\inferrule[PAll]{ \Gamma^-, \alpha \leq \phi^- \vdash \sigma^- \leq \tau^+}{\Gamma^- \vdash \forall \alpha\leq\mathit{Top}.\sigma^- \leq \forall \alpha \leq \phi^-.\tau^+}

\inferrule[PNeg]{ \Gamma^- \vdash \t^- \leq \sigma^+ }{ \Gamma^- \vdash \lnot \sigma^+ \leq \lnot \t^- }
\end{mathpar}
\caption{Subtyping for System \fsubp}
\label{fig:fsubp}
\end{figure}

Pierce showed that subtyping in System \fsubp is turing complete by reducing
reduction in a rowing machine to  subytping in a further reduced version
called System \fsubd. I focus here on System \fsubp as each rule is simpler
compared to System \fsubd and begets a more consice translation to and from
Julia subtyping terms.

\subsection{Reduction from System \fsubp to Julia subtyping}

\newcommand{\jltrans}[1]{\ensuremath{\llbracket#1\rrbracket}} To show turing
completeness of subtyping in Julia I show that that there exists a
translation, denoted $\jltrans{\t}$, such that for any System \fsubp
environment $\Gamma^-$ and types $\t, \sigma$ there exists a resultant Julia
environment $E$ where  $\Gamma^- \vdash \t \leq \sigma \iff \jltrans{\Gamma^-}
\vdash \jltrans{\sigma} <: \jltrans{\t} \vdash E$.

The translation needs to convert environments $\Gamma^-$ into Julia subtyping
environments $E$ as well as both positive and negative types $\t^+$ and
$\t^-$. Our translation rules are depicted in~\ref{fig:fsubtojl}. I use
the nominal type constructors $\cstrt{Neg}$ and $\cstrt{All}$ simply to
transition into distinguishable invariant contexts; their definitions
are simply \jlinl{struct Neg{T} end} and \jlinl{struct All{T} end}.

\begin{figure}

  \footnotesize
  \[
  \begin{array}{ccl@{\qquad}l}
      \\ \jltrans{\Gamma^-, \alpha \leq \phi^-} & ::= & \jltrans{\Gamma^-} {^L \alpha_l^u}  {^R \beta_\alpha^\alpha} & \text{\emph{Environments}}
      \\ \jltrans{\cdot} & ::= & \cdot & 
      \\ 
      \\ \jltrans{\mathit{Top}} & ::= & \uniontyp{} & \text{\emph{Positive types}}
      \\ \jltrans{\lnot \t^-} & ::= & \cstrt{Neg}{\jltrans{\t^-}} & 
      \\ \jltrans{\forall \alpha\leq\t^-.\t^+} & ::= & \wheret{\tupletyp{\cstrt{All}{\alpha}, \jltrans{\t^+}}}{\jltrans{\t^-} <: \alpha <: \texttt{Any}} &
      \\
      \\ \jltrans{\alpha} & ::= & \alpha & \text{\emph{Negative types}}
      \\ \jltrans{\lnot \t^+} & ::= & \wheret{\cstrt{Neg}{\alpha}}{\jltrans{\t^+} <: \alpha <: \texttt{Any}} & 
      \\ \jltrans{\forall \alpha\leq\mathit{Top}.\t^-} & ::= & \wheret{\tupletyp{\cstrt{All}{\alpha}, \jltrans{\t^-}}}{\uniontyp{} <: \alpha <: \texttt{Any}} &
  \end{array}
  \]
\caption{Translation from System \fsubp to Julia}
\label{fig:fsubtojl}
\end{figure}

I assume that variable names in System \fsubp and Julia are equivalent
for the purposes of clarity and to avoid confusion; I will be referring
to variables using the System \fsubp terminology $\alpha$ and $\beta$. Note that
I implicitly map System \fsubp variables to a pair of Julia variables; I treat
the System \fsubp environment $\Gamma^-$ as mapping both of these Julia variables
to their originating System \fsubp bound.

The translation of the environment is simple: for a given System \fsubp variable
$\alpha$ it creates the matching Julia type variables $\alpha$ and $\beta$ and
bounds the right-hand variable $\beta$ to be equal to $\alpha$. 

Positive and negative type translation is simple: I take a universally quantified
term with an upper bound and flip it into an existentially quantified term with a 
lower bound. Additionally, I introduce a tuple that contains an invariant reference
cell containing the newly-introduced variable as its first element and the translated
version of the type being quantified over as the second element. Translation for 
negated types follows from the original System \fsub definition of negation wherein 
$\lnot \tau \equiv \forall \alpha \leq \tau.\alpha$; I simply translate the negation
as if it were explicitly written out with this definition.

\paragraph{Forward}
I begin by showing that $\Gamma^- \vdash \t^- \leq \sigma^+ \implies \jltrans{\Gamma^-}
\vdash \jltrans{\sigma^+} <: \jltrans{\t^-} \vdash E$. To do so, I proceed by induction
on the System \fsub rule used to derive the premise.
\begin{itemize}
  \item \textsc{PTop}; want to show that $\jltrans{\Gamma^-} \vdash \jltrans{\mathit{Top}} <: \jltrans{\t^-} \vdash E'$ for some $E'$. This follows trivially from rule \textsc{Top} as $\jltrans{\mathit{Top}}$ is $\uniontyp{}$, the bottom type.
  \item \textsc{PVar}; want to show that $\jltrans{\Gamma^-} \vdash  \jltrans{\t^+} <: \jltrans{\alpha} \vdash E'$.
  I know that $\Gamma^- \vdash \Gamma(\sigma) \leq \t^+$ so by the IH there is some $E'$ such that $\jltrans{\Gamma^-} \vdash  \jltrans{\t^+} <: \jltrans{\Gamma^-(\alpha)} \vdash E'$.
  Let $\sigma = \Gamma^-(\alpha)$; note that this implies $\alpha \leq \sigma \in \Gamma^-$, so $\jltrans{\Gamma^-(\alpha)} = \jltrans{\sigma}$.
  Therefore, by \textsc{L\_Right}, $\jltrans{\Gamma^-} \vdash \jltrans{\t^+} <: \alpha \vdash E'$ since $\jltrans{\Gamma^-} \vdash \jltrans{\t^+} <: \jltrans{\sigma} \vdash E'$
  because by definition of environment translation $^L \alpha_{\jltrans{\sigma}}^\texttt{Any} \in \jltrans{\Gamma^-}$.
  Finally, $\jltrans{\Gamma^-} \vdash  \jltrans{\t^+} <: \jltrans{\alpha} \vdash E'$ as $\jltrans{\alpha} = \alpha$.
  \item \textsc{PAll}: 
  \newcommand{\extctx}{{^L\alpha_{\jltrans{\phi^-}}^{\texttt{Any}}}, {^R\beta_{\uniontyp{}}^{\texttt{Any}}}}
  \newcommand{\extctxp}{{^L\alpha_{\jltrans{\phi^-}}^{\texttt{Any}}}, {^R\beta_{\uniontyp{\alpha}}^{\texttt{Any}}}}
  \newcommand{\extctxpp}{{^L\alpha_{\jltrans{\phi^-}}^{\texttt{Any}}}, {^R\beta_{\uniontyp{\alpha}}^{\alpha}}}
  I want to show $\jltrans{\Gamma^-} \vdash \jltrans{\forall \alpha \leq \phi^-.\tau^+} <: \jltrans{\forall \beta\leq\mathit{Top}.\sigma^-} \vdash \jltrans{\Gamma^-}$.
  Equivalently, applying the translation, I want to show that $\jltrans{\Gamma^-} \vdash \wheret{\tupletyp{\cstrt{All}{\alpha}, \jltrans{\tau^+}}}{\jltrans{\phi^-} <: \alpha <: \texttt{Any}} <: \wheret{\tupletyp{\cstrt{All}{\beta}, \jltrans{\sigma^-}}}{\uniontyp{} <: \beta <: \texttt{Any}} \vdash \jltrans{\Gamma^-}$. To show this, I must apply L\_Intro and R\_Intro, making our goal $\jltrans{\Gamma^-}, \extctx \vdash \tupletyp{\cstrt{All}{\alpha}, \jltrans{\tau^+}} <: \tupletyp{\cstrt{All}{\beta}, \jltrans{\sigma^-}} \vdash \jltrans{\Gamma^-}, \extctxpp$. I are left with two subgoals after applying Tuple.
  \begin{itemize}
    \item $\jltrans{\Gamma^-}, \extctx \vdash \cstrt{All}{\alpha} <: \cstrt{All}{\beta} \vdash \jltrans{\Gamma^-}, \extctxpp$; which I resolve by applying App\_Inv, again begetting two subgoals.
    \begin{itemize}
      \item $\jltrans{\Gamma^-}, \extctx \vdash \alpha <: \beta \vdash \jltrans{\Gamma^-}, \extctxp$, the forward direction of subtyping. I get $\jltrans{\Gamma^-}, \extctx \vdash \alpha <: \texttt{Any} \vdash \jltrans{\Gamma^-}, \extctxp$ by applying R\_Right, then
      $\jltrans{\Gamma^-}, \extctx \vdash \texttt{Any} <: \texttt{Any} \vdash \jltrans{\Gamma^-}$ by L\_Left.
      \item The converse case, $\jltrans{\Gamma^-}, \extctxp \vdash \beta <: \alpha \vdash \jltrans{\Gamma^-}, \extctxpp$; I first apply R\_Left to get $\jltrans{\Gamma^-}, \extctxp \vdash \uniontyp{} <: \alpha \vdash \jltrans{\Gamma^-}, \extctxp$, which holds trivially.
    \end{itemize}
    \item By the IH, $\jltrans{\Gamma^-, \alpha\leq\phi^-}, \extctxpp \vdash \jltrans{\tau^+} <: \jltrans{\sigma^-} \vdash \jltrans{\Gamma^-}, \extctxpp$.
  \end{itemize}
  \item \textsc{PNeg}: 
  I want to show $\jltrans{\Gamma^-} \vdash  \jltrans{\lnot \t^-} <: \jltrans{\lnot \sigma^+} \vdash \jltrans{\Gamma^-}$, or, equivalently under translation, that $\jltrans{\Gamma^-} \vdash \cstrt{Neg}{\jltrans{\t^-}} <: \wheret{\cstrt{Neg}{\alpha}}{\jltrans{\sigma^+} <: \beta <: \texttt{Any}} \vdash \jltrans{\Gamma^-}$.
  \newcommand{\negctx}{{^R \beta_{\jltrans{\sigma^+}}^{\texttt{Any}}}}
  \newcommand{\negctxp}{{^R \beta_{\uniontyp{\jltrans{\sigma^+}, \jltrans{\t^-}}}^{\texttt{Any}}}}
  I show this by applying rule \textsc{R\_Intro} begetting $\jltrans{\Gamma^-}, \negctx \vdash \cstrt{Neg}{\jltrans{\t^-}} <: \cstrt{Neg}{\beta} \vdash E$. Note that while $E$ contains $\beta$ it will be removed by \textsc{R\_Intro} in the output environment.
  Rule \textsc{App\_Inv} then applies, giving us two goals, one for each direction.
  \begin{itemize}
    \item In the forward direction I want to show that $\jltrans{\Gamma^-}, \negctx \vdash \jltrans{\t^-} <: \beta \vdash E'$ for some $E'$. Only rule \textsc{R\_Right} can apply, the application of which requires
    us to show $\jltrans{\Gamma^-}, \negctx \vdash \jltrans{\t^-} <: \texttt{Any} \vdash \jltrans{\Gamma^-}, \negctx$ which holds trivially by \textsc{Top}. Therefore $E'$ must be $\jltrans{\Gamma^-}, \negctxp$.
    \item In the backwards direction I want to show that $\jltrans{\Gamma^-}, \negctxp \vdash \beta <: \jltrans{\t^-} \vdash E$. I must do so by applying rule \textsc{R\_Left}, which then
    requires us to show that $\jltrans{\Gamma^-}, \negctxp \vdash \uniontyp{\jltrans{\sigma^+}, \jltrans{\t^-}} <: \jltrans{\t^-} \vdash E$. Applying rule \textsc{Union\_Left} gives us two ensuing cases:
    \begin{itemize}
      \item First, I need to show $\jltrans{\Gamma^-}, \negctxp \vdash \jltrans{\sigma^+} <: \jltrans{\t^-} \vdash E''$. I apply the weakening lemma to simplify this to
      $\jltrans{\Gamma^-} \vdash \jltrans{\sigma^+} <: \jltrans{\t^-} \vdash E'''$, which holds by the IH.
      \item Next, I need to show that $E''' \vdash \jltrans{\t^-} <: \jltrans{\t^-} \vdash E'''$. I trivially apply reflectivity to show the rule and conclude that $E = E'''$.
    \end{itemize}
  \end{itemize}
\end{itemize}

\paragraph{Reverse}
The reverse direction consists of showing that $\jltrans{\Gamma^-} \vdash \jltrans{\sigma^+} <: \jltrans{\t^-} \vdash E \implies \Gamma^- \vdash \t^- \leq \sigma^+$. I proceed by rule induction on the Julia subtyping
rule used to derive $\jltrans{\Gamma^-} \vdash \jltrans{\sigma^+} <: \jltrans{\t^-} \vdash E$ while performing case analysis of the structure of the translated type. I begin by case analyzing on the left-hand side $\t^-$ to see if it is a variable or not, then
on the right-hand-side $\sigma^+$.
 eing translated which then uniquely identifies the rule being applied.
\begin{itemize}
  \item $\t^- = \alpha$. I know that $\jltrans{\Gamma^-} \vdash \jltrans{\sigma^+} <: \jltrans{\alpha} \vdash \jltrans{\Gamma^-}$ or equivalently $\jltrans{\Gamma^-} \vdash \jltrans{\sigma^+} <: \alpha \vdash \jltrans{\Gamma^-}$. 
  If $\alpha$ is a left-hand variable, then I know that $^L \alpha_{\jltrans{\sigma^-}}^{\texttt{Any}}$ is in the environment as \textsc{L\_Right} must have been used to resolve it and by definition of environment translation $\alpha\leq\sigma^- \in \Gamma$. 
  I then know that $\jltrans{\Gamma^-} \vdash \jltrans{\sigma^+} <: \jltrans{\sigma^-} \vdash \jltrans{\Gamma^-}$ and by the IH that $\Gamma^- \vdash \sigma^- \leq \sigma^+$ or that $\Gamma^- \vdash \Gamma^-(\alpha) \leq \sigma^+$. If $\alpha$ is a right-hand
  variable then I know that $^L \beta_{\jltrans{\sigma^-}}^\texttt{Any} {^R \alpha_{\beta}^\beta} \in \jltrans{\Gamma^-}$ by definition of environment translation and that $\beta\leq\sigma^- \in \Gamma^-$. I must have used rule \textsc{R\_Right} to conclude
  this, so I then know that $\jltrans{\Gamma^-} \vdash \jltrans{\sigma^+} <: \beta \vdash E$ which then implies $\jltrans{\Gamma^-} \vdash \jltrans{\sigma^+} <: \jltrans{\sigma^-} \vdash E$ from which $\Gamma^- \vdash \sigma^- \leq \sigma^+$ follows.
  Then, $\Gamma^- \vdash \Gamma^-(\beta) \leq \sigma^+$ and thus $\Gamma^- \vdash \Gamma^-(\alpha) \leq \sigma^+$. 
  \item $\sigma^+ = \mathit{Top}$; therefore $\Gamma^- \vdash \t^- \leq \sigma^+$ holds iff $\Gamma^- \vdash \t^- \leq \mathit{Top}$ which holds by \textsc{PTop}.
  \item $\sigma^+ = \lnot \sigma^-$. First, I are given that $\jltrans{\Gamma^-} \vdash \cstrt{Neg}{\jltrans{\sigma^-}} <: \jltrans{\t^-} \vdash \jltrans{\Gamma^-}$. It follows that $\jltrans{\t^-}$ must be of the
  form $\cstrt{Neg}{\jltrans{\t^+}}$ as this judgment must have been concluded by first applying \textsc{L\_Intro} followed by \textsc{App\_Inv}; this can only occur when the constructor names match and therefore there must exist some $\t^+$
  such that $\t^- = \lnot \t^+$. Thus, I want to show that $\Gamma^- \vdash \lnot \t^+ \leq \lnot \sigma^-$ given $\jltrans{\Gamma^-} \vdash \cstrt{Neg}{\jltrans{\sigma^-}} <: \wheret{\cstrt{Neg}{\beta}}{\jltrans{\t^+} <: \beta <: \texttt{Any}} \vdash \jltrans{\Gamma^-}$.
  To conclude the latter \textsc{R\_Intro} must have been used, so I have $\jltrans{\Gamma^-}, {^R \beta_{\jltrans{\t^+}}{\texttt{Any}}} \vdash \cstrt{Neg}{\jltrans{\sigma^-}} <: \cstrt{Neg}{\beta} \vdash E$. Rule \textsc{App\_Inv} must
  have then been applied to find that there exists some environment $E'$ such that $\jltrans{\Gamma^-}, {^R \beta_{\jltrans{\t^+}}{\texttt{Any}}} \vdash \jltrans{\sigma^-} <: \beta \vdash E'$ (1) and
  $E' \vdash \beta <: \jltrans{\sigma^-} \vdash E$ (2). Rule \textsc{R\_Right} must have been used to resolve (1) so therefore $E' = \jltrans{\Gamma^-}, {^R \beta_{\uniontyp{\jltrans{\t^+}, \jltrans{\sigma^-}}}{\texttt{Any}}}$ 
  by applying \textsc{Any} to show the bound. Thus, to show $\jltrans{\Gamma^-}, {^R \beta_{\uniontyp{\jltrans{\t^+}, \jltrans{\sigma^-}}}{\texttt{Any}}} \vdash \beta <: \jltrans{\sigma^-} \vdash E$, we
  must have derived $\jltrans{\Gamma^-}, {^R \beta_{\uniontyp{\jltrans{\t^+}, \jltrans{\sigma^-}}}{\texttt{Any}}} \vdash \uniontyp{\jltrans{\t^+}, \jltrans{\sigma^-}} <: \jltrans{\sigma^-} \vdash E$ to apply \textsc{R\_Left}.
  Rule \textsc{Union\_Left} must have been applied with $\jltrans{\Gamma^-}, {^R \beta_{\uniontyp{\jltrans{\t^+}, \jltrans{\sigma^-}}}{\texttt{Any}}} \vdash \jltrans{\t^+} <: \jltrans{\sigma^-} \vdash E''$ as a precondition.
  Since $\beta$ is fresh, $\jltrans{\Gamma^-} \vdash \jltrans{\t^+} <: \jltrans{\sigma^-} \vdash \jltrans{\Gamma^-}$ then follows, and therefore by application of the IH $\Gamma^- \vdash \sigma^- \leq \t^+$. Thus,
  I can apply \textsc{PVar} to conclude that $\Gamma^- \vdash \lnot \t^+ \leq \sigma^-$ or $\Gamma^- \vdash \t^- \leq \sigma^+$.
  \item $\sigma^+ = \forall \alpha\leq\sigma^-.\t^+$. By similar reasoning to the prior case I can derive that $\t^{-} = \forall \beta.\phi$. Thus, we
  want to show that $\Gamma^- \vdash \lnot (\forall \beta\leq\mathit{Top}.\phi) \leq \forall \alpha\leq\sigma^-.\t^+$ given that 
  $\jltrans{\Gamma^-} \vdash \jltrans{\forall \alpha\leq\sigma^-.\t^+} <: \jltrans{\forall \beta\leq\mathit{Top}.\phi} \vdash E$. I 
  expand the translations to get $\jltrans{\Gamma^-} \vdash \wheret{\tupletyp{\cstrt{All}{\alpha}, \jltrans{\t^+}}}{\jltrans{\sigma^-} <: \alpha <: \texttt{Any}} <: \wheret{\tupletyp{\cstrt{All}{\beta}, \jltrans{\phi}}}{\uniontyp{} <: \beta <: \texttt{Any}} \vdash E$.
  I had to have applied \textsc{L\_Intro} and \textsc{R\_Intro} to get $G \vdash \tupletyp{\cstrt{All}{\alpha}, \jltrans{\t^+}} <: \tupletyp{\cstrt{All}{\beta}, \jltrans{\phi}} \vdash E$ where $G=\jltrans{\Gamma^-}, {^L \alpha_{\jltrans{\sigma^-}}^{\texttt{Any}}}, {^R \beta_{\jltrans{\phi}}^{\texttt{Any}}}$.
  To have concluded this, I need to have had some intermediate context $E'$ such that I can use \textsc{tuple} to conclude both $G \vdash \cstrt{All}{\alpha} <: \cstrt{All}{\beta} \vdash E'$ (1) and $E' \vdash \jltrans{\t^+} <: \jltrans{\phi} \vdash E$ (2). Furthermore,
  to have concluded (1), I had to have applied \textsc{App\_Inv} with intermediate context $E''$ such that $G \vdash \alpha <: \beta \vdash E''$ (1.1)  and $E'' \vdash \beta <: \alpha \vdash E'$. The same sequence as applies in the forward case must be used to conclude
  these judgements, so $E' = \jltrans{\Gamma^-} {^L \alpha_{\jltrans{\sigma^-}^{\texttt{Any}}}} {^R \beta_\alpha^\alpha}$. Note that $\jltrans{\Gamma^-, \alpha\leq\sigma^-}=E'$. Therefore, from (2), it follows that $\jltrans{\Gamma^-, \alpha\leq\sigma^-} \vdash \jltrans{\t^+} <: \jltrans{\phi} \vdash E$.
  Thus, by the IH, I conclude that $\Gamma^-, \alpha\leq\sigma^- \vdash \phi \leq \t^+$ and then by application of \textsc{PAll} that $\Gamma^- \vdash \forall \beta.\phi \leq \forall \alpha\leq \sigma^-.\t^+$.
\end{itemize}

Therefore, System \fsubp subtyping holds if and only if the translated version in Julia holds. By Pierce's result, then, I can conclude that subtyping in Julia is undecidable. 

\section{Undecidability in Practice}

The undecidability result is not a purely theoretical one; the described
translation is able to match the undesirable behavior of System \fsubp in
Julia. For example, Ghelli's looping gadget can be trivially translated into
Julia as shown in figure~\ref{fig:ghelli}.
\begin{figure}[H]
\begin{lstlisting}
function Neg(T)
  return (Ref{X} where X>:T)
end
function Kappa(T)
  return (Tuple{Ref{Y},Neg(Y)} where Y>:T)
end
const Theta = (Tuple{Ref{Z},Neg(Kappa(Z))} where Z)
Kappa(Theta) <: Theta # stackoverflow error
\end{lstlisting}
\caption{Ghelli's looping gadget, Julia version.}
\label{fig:ghelli}
\end{figure}

Similarly, I can mechanically translate any System \fsubp judgement
(including typing context $\Gamma$) into Julia and check it using Julia's
typechecker with predictable results, seen in figure~\ref{fig:fsubptrans}.
The translation takes a Julia representation of a System \fsubp environment
(a dictionary mapping names to System \fsubp upper bound types) and a pair
of System \fsubp types to compare and translates it into an equivalent Julia
subtype check. Translation itself follows the construct used in the proof exactly,
though getting it started takes more effort in the form of a small gadget.

The gadget ensures  that the System \fsubp variables are the same on both
sides of the Julia subtyping judgment. I accomplish this by having both
encoded Julia types \jlinl{A} and \jlinl{B} begin on the left side of the
judgment in the tuple \jlinl{Tuple{Ref{A}, Ref{V} where Vl>:B}} around which
the in-scope variables are added. The right hand side is then constructed as
\jlinl{Tuple{Ref{T}, Ref{V} where Vr>:T} where T}. Julia then first concludes
\jlinl{Ref{A} <: Ref{T}} therein forcing \jlinl{A} and \jlinl{T} to be equal.
Julia then checks if \jlinl{(Ref{Vl} where Vl>:B) <: (Ref{Vr} where Vr>:T)}
or, substituting for \jlinl{T}, if \jlinl{(Ref{Vl} where Vl>:B) <: (Ref{Vr}
where Vr>:A)}. Running in the forward direction, Julia then ensures that
\jlinl{Vl <: Vr} using \textsc{R\_Right}, then checks the opposite direction
\jlinl{Vr <: Vl}. Applying \textsc{R\_Left} then \textsc{L\_Right} then gives
us \jlinl{A <: B} as the ultimate judgment that needs to be concluded. 

\begin{figure}
\begin{lstlisting}
abstract type FSubType end
struct FSubVar <: FSubType
  name::Symbol
end
struct FSubTop <: FSubType end
struct FSubUni <: FSubType
  binding::Symbol
  ub::FSubType
  body::FSubType
end


enc(v::FSubVar, eenv::Dict{Symbol,TypeVar}) = eenv[v.nam
enc(v::FSubTop, eenv::Dict{Symbol, TypeVar}) = Union{}
function enc(v::FSubUni, eenv::Dict{Symbol, TypeVar}) 
  nvn = TypeVar(gensym(v.binding), enc(v.ub, eenv), Any)
  eenv[v.binding] = nvn
  return UnionAll(nvn, Tuple{Ref{nvn}, enc(v.body, eenv)})
end
enc(v::FSubType) = enc(v, Dict{Symbol, TypeVar}())

enc(env::Dict{Symbol,FSubType}) = Dict(k => begin nvn = TypeVar(gensym(k)); nvn.lb = enc(v); nvn end for (k,v) in env)

esub(a::FSubType,b::FSubType) = enc(b) <: enc(a)
function esub(a::FSubType, b::FSubType, env::Dict{Symbol,FSubType})
  tenv = enc(env)
  A = enc(b, tenv)
  B = enc(a, tenv)
  lhs = foldl((t,v) -> UnionAll(v,t), values(tenv); init=Tuple{Ref{A}, Ref{V} where V>:B})
  rhs = (Tuple{Ref{T}, Ref{V} where V>:T} where T)
  return lhs <: rhs
end
\end{lstlisting}
\caption{Implementation of translation from System \fsubp subtyping judgments to Julia subtyping.}
\label{fig:fsubptrans}
\end{figure}

While somewhat contrived, these examples demonstrate that it is possible to
get undesirable behavior out of Julia's type system and that deep theoretical
analysis is probably intractable with the current conception. 

The practical impact of subtyping's undecidability is minor. The key feature
being used, lower bounds on type variables, is extremely rare in practice.
The only identifiable uses of it occur in one place in the entire standard library,
nowhere in broader user code, and rarely as a result of type inference. As a result,
the main impacts are on the theory of the system.

From the perspective of a type system the result is frustrating: it means that
static dispatch resolution is essentially incomplete if it is to terminate.
After all, if methods are chosen with subtyping and subtyping is undecidable,
then I cannot universally decide which methods will be invoked in a finite
amount of time. My objective is instead to try to find a subset of the type
language that is \emph{sufficient} to work practically while skirting around
the unpleasant generalities.

Julia programmers are already reasoning extensively about types and Julia
programs, obviously, already make extensive use of types. It follows, then,
that a reasonable implementation of subtyping already exists as part of the
Julia compiler. Moreover, Julia's own implementation of subtyping is the one
against which programmers judge whether their types are too complex or not; if
a relation sends the implementation into an infinite loop most programmers
will consider that a bug. Therefore, I will use the notion of subtyping
as implemented in Julia.

Adopting this concept leads to a problem. Julia implements subtyping as a very
large, relatively opaque, code base that is not tractable to theoretical
examination. The simplified formalization after all has proven largely
too difficult to reason about. Modeling the implementation is practically
infeasible and any such effort would be unlikely to produce pleasant formal 
results.

As a result, instead of a concrete formal model of subtyping I will instead 
be relying on a ``black box'' concept of what subtyping is. In place of a
concrete modelled algorithm for subtyping my type system instead relies on
subtyping adhering to a set of properties; the type system will then work with
any satisfying instantiation of this relation. This design accommodates both
future extensibility (as the type language and underlying subtyping relation are 
not linked to any specific definition) and the inherent need to choose trade-offs
when implementing subtyping in Julia.

For the purposes of typing \emph{itelf}, luckily, relatively few properties
are needed of subtyping. As will be described later, one property is required:
if $\jsub{\t_1}{\t_2}$ and $\jsub{\t_2}{\t_3}$ then $\jsub{\t_1}{\t_3}$---that
subtyping is transitive.

\DeclareDocumentCommand\TR{om}{\EM{\IfNoValueTF{#1}{\PackageWarning{}{Undefined Type System}}{#1}\llbracket #2 \rrbracket}}

\DeclareDocumentCommand\TRG{omm}{\EM{\IfNoValueTF{#1}{\PackageWarning{}{Undefined Type System}}{#1}\llbracket #2 \rrbracket_{#3}}}
\DeclareDocumentCommand\TAG{ommm}{\EM{\IfNoValueTF{#1}{\PackageWarning{}{Undefined Type System}}{#1}\llparenthesis #2 \rrparenthesis_{#3}^{#4}}}

\newcommand{\OTS}{{\mathcal{O}}}
\newcommand{\CTS}{{\mathcal{C}}}
\newcommand{\BTS}{{\mathcal{B}}}
\newcommand{\TTS}{{\mathcal{T}}}
\newcommand{\SOMS}{{\mathcal{S}}}
\newcommand{\sspce}{;~}
\newcommand{\idbody}[1]{\SubCast{#1}\x\sspce \SubCast{#1}\x}

\newcommand{\bscast}[2]{\EM{\BehCast{#1}{{#2}}}}
\newcommand{\kty}[1]{\EM{\xt{kty}(#1)}}

\newcommand{\WHERE}{~\EM{\xt{\bf where}}~}
\newcommand{\OR}{\EM{~\xt{\bf or}}~}
\newcommand{\IF}{\EM{~\xt{\bf if}}~}

\newcommand{\HS}{\hspace{.2cm}}
\newcommand{\LS}{\hspace{1cm}}
\newcommand{\namet}[2]{\EM{#1\!\!\,:\,\!#2}}
\newcommand{\TypeCk}[3]{\EM{#1\vdash #2:#3}}

\newcommand{\EM}[1]{\ensuremath{#1}\xspace}
\newcommand{\xt}[1]{\ensuremath{\mathsf{#1}}}
\newcommand{\bt}[1]{\xt{\bf #1}}

\newcommand{\EMxt}[1]{\EM{\xt{#1}}}

\newcommand{\x}   {\EMxt x}
\newcommand{\xp}   {\EMxt{x'}}
\newcommand{\n}   {\EMxt n}

\providecommand{\m}{m}
\renewcommand{\m}   {\EMxt m}
\providecommand{\mp}{mp}
\renewcommand{\mp}   {\EMxt{m'}}
\newcommand{\s}   {\EM{\sigma}}
\DeclareDocumentCommand\a{o}{\IfNoValueTF{#1}{\EMxt {a}}{\EM{\xt {a}_{#1}}}}
\DeclareDocumentCommand\ap{o}{\IfNoValueTF{#1}{\EMxt {a'}}{\EM{\xt {a'}_{#1}}}}
\DeclareDocumentCommand\app{o}{\IfNoValueTF{#1}{\EMxt {a''}}{\EM{\xt {a''}_{#1}}}}
\DeclareDocumentCommand\t{o}{\IfNoValueTF{#1}{\EMxt t}{\EM{\xt t_{#1}}}}
\DeclareDocumentCommand{\tp}{o}{\IfNoValueTF{#1}{\EM{ \xt t' }}{\EM{\xt t_{#1}'}}}
\DeclareDocumentCommand{\tpp}{o}{\IfNoValueTF{#1}{\EM{ \xt t'' }}{\EM{\xt t_{#1}''}}}
\DeclareDocumentCommand{\tppp}{o}{\IfNoValueTF{#1}{\EM{ \xt t''' }}{\EM{\xt t_{#1}'''}}}
\DeclareDocumentCommand{\e}{o}{\IfNoValueTF{#1}{\EM{ \xt e }}{\EM{\xt e_{#1}}}}
\DeclareDocumentCommand{\ep}{o}{\IfNoValueTF{#1}{\EM{ \xt e' }}{\EM{\xt e_{#1}'}}}
\DeclareDocumentCommand{\epp}{o}{\IfNoValueTF{#1}{\EM{ \xt e'' }}{\EM{\xt e''_{#1}}}}
\DeclareDocumentCommand{\eppp}{o}{\IfNoValueTF{#1}{\EM{ \xt e''' }}{\EM{\xt e'''_{#1}}}}
\DeclareDocumentCommand{\fd}{o}{\IfNoValueTF{#1}{\EM{ \xt{fd} }}{\EM{\xt{fd}_{#1}}}}
\DeclareDocumentCommand{\fdp}{o}{\IfNoValueTF{#1}{\EM{ \xt{fd}' }}{\EM{\xt{fd}_{#1}'}}}
\DeclareDocumentCommand{\fdpp}{o}{\IfNoValueTF{#1}{\EM{ \xt{fd}'' }}{\EM{\xt{fd}_{#1}''}}}
\DeclareDocumentCommand{\fdppp}{o}{\IfNoValueTF{#1}{\EM{ \xt{fd}''' }}{\EM{\xt{fd}_{#1}'''}}}
\DeclareDocumentCommand{\md}{o}{\IfNoValueTF{#1}{\EM{ \xt{md} }}{\EM{\xt{md}_{#1}}}}
\DeclareDocumentCommand{\f}{o}{\IfNoValueTF{#1}{\EM{ \xt f }}{\EM{\xt f_{#1}}}}
\DeclareDocumentCommand{\mdp}{o}{\IfNoValueTF{#1}{\EM{ \xt{md}' }}{\EM{\xt{md}_{#1}'}}}
\DeclareDocumentCommand{\mdpp}{o}{\IfNoValueTF{#1}{\EM{ \xt{md}'' }}{\EM{\xt{md}_{#1}''}}}
\DeclareDocumentCommand{\mdppp}{o}{\IfNoValueTF{#1}{\EM{ \xt{md}''' }}{\EM{\xt{md}_{#1}'''}}}
\DeclareDocumentCommand{\C}{o}{\IfNoValueTF{#1}{\EM{ \xt{C} }}{\EM{\xt{C}_{#1}}}}
\DeclareDocumentCommand{\mt}{o}{\IfNoValueTF{#1}{\EM{ \xt{mt} }}{\EM{\xt{mt}_{#1}}}}
\DeclareDocumentCommand{\mtp}{o}{\IfNoValueTF{#1}{\EM{ \xt{mt}' }}{\EM{\xt{mt}_{#1}'}}}
\DeclareDocumentCommand{\mtpp}{o}{\IfNoValueTF{#1}{\EM{ \xt{mt}'' }}{\EM{\xt{mt}_{#1}''}}}
\DeclareDocumentCommand{\mtppp}{o}{\IfNoValueTF{#1}{\EM{ \xt{mt}''' }}{\EM{\xt{mt}_{#1}'''}}}
\DeclareDocumentCommand{\D}{o}{\IfNoValueTF{#1}{\EM{ \xt{D} }}{\EM{\xt{D}_{#1}}}}
\DeclareDocumentCommand{\Dp}{o}{\IfNoValueTF{#1}{\EM{ \xt{D'} }}{\EM{\xt{D'}_{#1}}}}
\DeclareDocumentCommand{\Dpp}{o}{\IfNoValueTF{#1}{\EM{ \xt{D''} }}{\EM{\xt{D''}_{#1}}}}

\newcommand{\K}   {\EMxt K}
\renewcommand{\k} {\EMxt k}
\newcommand{\Kk}   {\K~\k}
\newcommand{\Kp}  {{\EMxt{K'}}}
\newcommand{\Kpp}  {{\EMxt{K''}}}
\newcommand{\Kppp}  {{\EMxt{K'''}}}
\renewcommand{\sp}{{{\EM{\s'}}}}
\newcommand{\spp}{{{\EM{\s''}}}}
\newcommand{\A}   {\EMxt {A}}
\newcommand{\I}   {\EMxt {I}}
\newcommand{\E}   {\EMxt {E}}
\newcommand{\Cp}  {\EMxt{C'}}
\newcommand{\Cpp}  {\EMxt{C''}}
\newcommand{\Cppp}  {\EMxt{C'''}}
\newcommand{\cmd}  {\EMxt{M}}
\newcommand{\cmdp}  {\EMxt{M'}}
\newcommand{\M}{\EMxt{M}}
\newcommand{\Env}   {\EM{\Gamma}}
\newcommand{\Envp}   {\EM{\Gamma'}}
\newcommand{\EE}   {\EM{\textsf{{E}}}}
\newcommand{\any} {\EM{\star}}
\newcommand{\this}{\EMxt{this}}
\newcommand{\that}{\EMxt{that}}
\newcommand{\none}{\EM{\cdot}}
\newcommand{\CW}    {\EMxt{?C}}
\newcommand{\CWp}   {\EMxt{?C'}}
\newcommand{\CWpp}  {\EMxt{?C''}}
\newcommand{\DW}     {\EMxt{?D}}
\newcommand{\DWp}    {\EMxt{?D'}}
\newcommand{\DWpp}   {\EMxt{?D''}}

\newcommand{\FRead}[1]   {\EM{\this.#1}}
\newcommand{\FWrite}[2]  {\EM{\this.#1} = #2}
\newcommand{\FReadR}[2]   {\EM{#1.#2}}
\newcommand{\FWriteR}[3]  {\EM{#1.#2} = #3}
\newcommand{\Call}[3]  {\EM{#1.#2(#3)}}
\newcommand{\KCall}[5] {\EM{{#1}.{#2}_{{#4} \shortrightarrow {#5}}(#3)}}
\newcommand{\DynCall}[3]  {\EM{#1@#2_{\any\shortrightarrow\any}(#3)}}
\newcommand{\sDynCall}[3]  {\EM{#1@#2(#3)}}

\newcommand{\New}[2]   {\EM{\new\;#1(#2)}}
\newcommand{\SubCast}[2]{\EM{\langle{#1}\rangle\,{#2}}}
\newcommand{\BehStart}{\EM{\blacktriangleleft}}
\newcommand{\BehEnd}{\EM{\blacktriangleright}}
\newcommand{\BehCast}[2]{\EM{\BehStart\! #1\! \BehEnd #2}}

\newcommand{\MonStart}{\EM{\triangleleft}}
\newcommand{\MonEnd}{\hspace{-2px}\EM{\triangleright}}
\newcommand{\MonCast}[2]{\EM{\MonStart\, #1\, \MonEnd #2}}
\newcommand{\typegen}[2]{\xt{ifcgen}(#1,#2)}

\newcommand{\new}      {\EM{\bt{new}}}
\newcommand{\HT}[2]    {\EM{{#1}\!:{#2}}}
\newcommand{\Mdef}[5]  {\EM{ \HT{ #1( \HT{#2}{#3})}{#4}\;\{{#5}\}}}
\newcommand{\obj}[2]   { \EM{ #1\{#2\}}}
\newcommand{\ThorSub}[4]{\EM{#1~#2 \vdash #3 \Sub_t #4}}
\newcommand{\behcast}[7]{\EM{#5\,#6\,#7 = \xt{bcast}(#1, #2, #3, #4)}}
\newcommand{\behcastE}[7]{\EM{\xt{bcast}(#1, #2, #3, #4) = #5\,#6\,#7}}
\newcommand{\behcastS}[4]{\EM{\xt{bcast}(#1, #2, #3, #4)}}

\newcommand{\B}        {\EM{~|~}}

\newcommand{\Red}{\EM{\rightarrow}}
\newcommand{\Reduce}[6]   {\EM{{#1}~{#2}~{#3} \Red {#4}~{#5}~{#6}}}
\newcommand{\ReduceA}[6]  {\EM{#1~#2~#3 } & \EM{\Red #4~#5~#6}}
\newcommand{\class}       {\EM{\bf{class}}}
\newcommand{\Class}[3]    {\EM{\bt{class}\;#1\,\{\,#2~#3\,\}}}

\newcommand{\Ftype}[2]    {\EM{ \HT{#1}{#2} }}
\newcommand{\Fdef}[2]    {\EM{ \HT{#1}{#2} }}
\newcommand{\Mtype}[3]    {\EM{ \HT{#1(#2)}{#3}}}
\newcommand{\opdef}[2]    {\framebox[1.1\width]{#1} ~ #2\\}
\newcommand{\Map}[2]     {\EM{ #1[#2] }}
\newcommand{\Bind}[2]     {\EM{#1 \mapsto #2}}

\newcommand{\Sub}{\EM{<:}}
\newcommand{\OK}{\EM{~\checkmark}}
\newcommand{\names}[1]{\EM{\xt{names}(#1)}}
\newcommand{\cload}[1]{\EM{\xt{nodups}(#1)}}

\newcommand{\ConsSub}{\EM{\lesssim}}

\newcommand{\EvalRulePassthrough}[1]{#1}
\newcommand{\EvalRuleLabel}[1]{\EvalRulePassthrough{\tiny\scshape\textcolor{gray}{#1}}}
\newcommand{\CondRule}[3]{ \EvalRuleLabel{#1} & #3 & #2 \\}
\newcommand{\SuchRule}[3]{ #3 &~{\emph{s.t.}} #2 \\}
\newcommand{\EnvType}[5]{ \EM{#1\,#2\,#3\vdash #4 : #5}}
\newcommand{\EnvTypex}[5]{ \EM{#1\,#2\,#3\vdash_{\!s} #4 : #5}}

\newcommand{\RuleRef}[1]{\hyperlink{infer:#1}{\TirNameStyle{#1}}}
\newcommand{\IRule}[4][]{\inferrule*[lab={\tiny \hypertarget{infer:#2}{#2}},#1]{#3}{#4}}
\newcommand{\Rule}[4][]{\inferrule*{#3}{#4}}
\newcommand{\HasType}[3]{ \EM{#1} (\EM{#2}) = \EM{#3}}
\newcommand{\wrap}[4]{\EM{\xt{W}(#1,#2,#3,#4)}}
\newcommand{\wrapE}[1]{\EM{\xt{W}(#1)}}
\newcommand{\wrapAny}[3]{\EM{\xt{W}\!\any(#1,#2,#3)}}

\newcommand{\classoff}[2]{\EM{\xt{mtypes}(#1,~#2)}}

\newcommand{\mtype}[3]{\EM{\xt{mtype}(#1,#2,#3)}}

\newcommand{\Convertible}[3]{\EM{#1 \vdash #2 \Mapsto #3}}
\newcommand{\ConvertE}[4]{\EM{#1 \vdash_{\!s} #3 \Mapsto #4}}
\newcommand{\In}{\EM{\in}}
\newcommand{\T}{\EM{\xt T}}
\newcommand{\AND}{\EM{\wedge}}
\newcommand{\App}[2]{\EM{#1(#2)}}

\newcommand{\SSub}[4]{\EM{#1~#2\vdash_{\!s} #3\Sub #4}}
\newcommand{\StrSub}[4]{\EM{#1~#2\vdash #3\Sub #4}}
\newcommand{\StrNotSub}[4]{\EM{#1~#2\vdash #3 \not\Sub #4}}
\newcommand{\ThrSub}[4]{\EM{#1~#2\vdash_{\!s} #3~\src\Sub~#4}}
\newcommand{\ConSub}[4]{\EM{#1~#2 \vdash #3 \lesssim #4}}
\newcommand{\OKW}{\EM{~\checkmark_{s}}}
\newcommand{\OKX}[1]{\EM{~\checkmark_{#1}}}
\newcommand{\EnvTypeW}[4]{ \EM{#1\,#2\vdash_{\!s} #3 : #4}}
\newcommand{\EnvTypeS}[4]{ \EM{#1\,#2\vdash_{\!s} #3 : #4}}
\newcommand{\EnvTypeT}[4]{ \EM{#1\,#2\vdash_{\!s} #3 : #4}}
\newcommand{\EnvTypeE}[5]{ \EM{#1\,#2\vdash_{\!s} #4 : #5}}
\newcommand{\trulename}[1]{#1}

\newcommand{\sign}[1]{\xt{signature}(#1)}

\newcommand{\WFtype}[2]{\EM{#1\vdash#2 \OK}}
\newcommand{\WF}[4]{\EM{#1\,#2\,#3\vdash#4 \OK}}
\newcommand{\WFp}[2]{#1~#2\OK}
\newcommand{\WFq}[1]{#1\OK}

\newcommand{\WFpx}[2]{#1~#2\OK_x}
\newcommand{\WFx}[4]{\EM{#1\,#2\,#3\vdash_{\!s}~#4 \OK}}
\newcommand{\WFX}[5]{\EM{#1\,#2\,#3\vdash_{\!#5}~#4 \OK}}
\newcommand{\WFtypex}[2]{\EM{#1\vdash_{\!s}~ #2 \OK}}
\newcommand{\WFtypeX}[3]{\EM{#1\vdash_{\!s}~ #3 \OK}}

\newcommand{\WFtypeW}[2]{\EM{#1\vdash_{\!s}~ #2 \OK}}
\newcommand{\WFW}[3]{\EM{#1\,#2\vdash_{\!s}~#3 \OK}}
\newcommand{\WFpW}[2]{#1~#2\OKW}
\newcommand{\WFpX}[3]{\EM{#1~#2\OKX{#3}}}

\newcommand{\V}{\EM{\checkmark}}

\newcommand{\dyn}[1]{\xt{dyn}(#1)}

\newcommand{\fresh}[1]{\EM{#1~\xt{fresh}}}
\newcommand{\Kt}[1]{\EM{\text{ktype}(#1)}}
\newcommand{\All}[1]{\EM{\forall ~\xt #1 ~.~}}
\newcommand{\SP}{\hspace{.5cm}}
\newcommand{\SPP}{\SP\SP}

\newcommand{\IGNOREUNLESSNEEDED}[1]{}
\newcommand{\figref}[1]{Fig.~\ref{#1}\xspace}
\newcommand{\ruleref}[1]{Rule~{\small #1}\xspace}

\newcommand{\kafka}{\ensuremath{\mathsf KafKa}\xspace}
\newcommand{\src}[1]{\colorbox[gray]{0.89}{$#1$}}
\newcommand{\dt}[1]{\,\xt{?}#1}
\newcommand{\consistent}[3]{\EM{#1 \vdash #2 ~\sim~ #3}}
\newcommand{\rtranst}[6]{#1 \Rightarrow #2 ~ #3 / #4 \vdash #5 \looparrowright_{beh} #6}
\newcommand{\rtranstz}[4]{#1 \Rightarrow #2 \vdash #3 \looparrowright_{mon} #4}

\providecommand{\figref}[1]{Fig.~\ref{#1}}
\newcommand{\secref}[1]{Sec.~\ref{#1}}
\newcommand{\chapref}[1]{Chapter~\ref{#1}}
\newcommand{\appref}[1]{App.~\ref{#1}}
\newcommand{\thmref}[1]{Theorem~\ref{#1}}
\newcommand{\lemref}[1]{Lemma~\ref{#1}}
\newcommand{\juliette}{\textsc{Juliette}\xspace}
\newcommand{\juliav}{Julia 1.4.1\xspace}
\newcommand{\jlmultnum}{357\xspace}
\newcommand{\jladdnum}{166\xspace}
\renewcommand{\eval}{\texttt{\small eval}\xspace}
\renewcommand{\Eval}{\texttt{\small Eval}\xspace}
\newcommand{\evals}{\texttt{\small eval}s\xspace}
\newcommand{\evaled}{\texttt{\small eval}ed\xspace}
\newcommand{\invokelatest}{\texttt{\small invokelatest}\xspace}
\newcommand{\Invokelatest}{\texttt{\small Invokelatest}\xspace}
\newcommand{\faded}[1]{\color{gray}#1}
\newcommand{\add}[1]{%
  \cbcolor{green}
  \begin{changebar}
    #1
  \end{changebar}
  \cbcolor{gray}}
\newcommand{\subst}[2]{\ensuremath{[#1\!\mapsto\!#2]}}
\newcommand{\dom}{\ensuremath{\mathop{dom}}}

\renewcommand{\obar}[1]{\overline{#1}}
\newcommand{\ol}[2]{\left.\overline{#1}\right.^{#2}}
\newcommand{\eqtt}{\texttt{=}}
\newcommand{\HEY}[2]{{\footnotesize\textsc{#1-#2}}\xspace}
\newcommand{\TRR}[1]{\HEY T{#1}}
\newcommand{\TER}[1]{\HEY{TE}{#1}}
\newcommand{\WAE}[1]{\HEY E{#1}}
\newcommand{\WAEF}[1]{\HEY{F}{#1}}
\newcommand{\WAT}[1]{\HEY{T}{#1}}
\newcommand{\WAQE}[1]{\HEY{QE}{#1}}
\newcommand{\WAOE}[1]{\HEY{OE}{#1}}
\newcommand{\WAOD}[1]{\HEY{OD}{#1}}
\newcommand{\WAOT}[1]{\HEY{OT}{#1}}
\newcommand{\WACAN}[1]{\HEY{CN}{#1}}
\newcommand{\langsep}{\;|\;}
\newcommand{\err}{{\texttt{error}}\xspace}
\newcommand{\primv}{\ensuremath{\nu}\xspace}
\newcommand{\primvs}{\ensuremath{\obar{\primv}}\xspace}
\newcommand{\Primop}{\ensuremath{\Delta}\xspace}
\newcommand{\primop}[1]{\ensuremath{\delta_{#1}}\xspace}
\newcommand{\primopd}{\primop{l}}
\newcommand{\primcall}[2]{\ensuremath{\primop{#1}(#2)}\xspace}
\newcommand{\primcalld}[1]{\primcall{l}{#1}}
\newcommand{\PrimopSE}{\ensuremath{\Omega}\xspace} 
\newcommand{\PrimopRT}{\ensuremath{\Psi}\xspace} 
\renewcommand{\v}{\ensuremath{\texttt{v}}\xspace}
\newcommand{\vs}{\obar{\v}\xspace}
\newcommand{\skp}{\texttt{unit}\xspace}
\newcommand{\mname}{\ensuremath{\texttt{m}}\xspace}
\newcommand{\mval}[1]{{\texttt{#1}}\xspace}
\providecommand{\m}{m}
\renewcommand{\m}{\mval{m}}
\renewcommand{\e}{\ensuremath{\texttt{e}}\xspace}
\newcommand{\es}{\ensuremath{\obar{\e}}\xspace}
\renewcommand{\x}{\ensuremath{x}\xspace}
\newcommand{\xs}{\ensuremath{\obar{\x}}\xspace}
\newcommand{\seq}[2]{\ensuremath{#1\,;\,#2}}
\newcommand{\mcall}[2]{\ensuremath{\mathop{#1}(#2)}}
\newcommand{\mcallt}[4]{\ensuremath{\mathop{#1}(#2) ::_{#3} #4}}
\newcommand{\mcalld}{\mcall{\m}{\vs}}
\renewcommand{\md}{\ensuremath{\texttt{md}}\xspace}
\newcommand{\mds}{\ensuremath{\obar{\md}}\xspace}
\newcommand{\evalty}[2]{\ensuremath{{\color{violet}\llbracket\,{\color{black}#1}\,\rrbracket}_#2}}
\newcommand{\evalchk}[2]{\ensuremath{\llbracket\,#1\,\rrbracket_{#2}}}
\newcommand{\evalinf}[1]{\ensuremath{\llbracket\,#1\,\rrbracket}}
\newcommand{\evalg}[1]{\ensuremath{{\color{violet}\llparenthesis\,{\color{black}#1}\,\rrparenthesis}}}
\newcommand{\evalt}[2]{\ensuremath{\evalg{#2}_{\color{violet}#1}}}
\newcommand{\mdef}[3]{\ensuremath{\triangleleft\, \mathop{#1}(#2)\eqtt\,#3 \,\triangleright}}
\newcommand{\mdefty}[4]{\ensuremath{\triangleleft\, \mathop{#1}(#2)::#3\eqtt\,#4 \,\triangleright}}
\newcommand{\mdefd}{\ensuremath{\mdef{\m}{\obar{\jty{\x}{\t}}}{\e}}}
\newcommand{\mdeftyd}{\ensuremath{\mdefty{\m}{\obar{\jty{\x}{\t}}}{\mu}{\e}}}
\newcommand{\rslt}{\ensuremath{\texttt{r}}\xspace}
\newcommand{\p}{\ensuremath{\texttt{p}}\xspace}
\newcommand{\iexp}{\ensuremath{\texttt{i}}\xspace}
\renewcommand{\t}{\ensuremath{\tau}\xspace}
\newcommand{\ts}{\ensuremath{\obar{\t}}\xspace}
\providecommand{\g}{g}
\renewcommand{\g}{\ensuremath{\sigma}\xspace}
\newcommand{\gs}{\ensuremath{\obar{\g}}\xspace}
\newcommand{\mty}{\ensuremath{\mathbb{f}_\m}\xspace}
\newcommand{\Unit}{\ensuremath{\mathbb{1}}\xspace}
\newcommand{\typeof}{\ensuremath{\mathop{\mathbf{typeof}}}\xspace}
\newcommand{\dispatch}{\ensuremath{\mathop{\mathbf{dispatch}}}\xspace}
\newcommand{\notjsub}[2]{\ensuremath{#1 \not<: #2}}
\newcommand{\jty}[2]{\ensuremath{#1 :: #2}}
\newcommand{\MT}{\ensuremath{\mathrm{M}}\xspace}
\newcommand{\MTg}{\ensuremath{\MT_g}\xspace}
\newcommand{\MText}[2]{\ensuremath{#1 \bullet #2}}
\newcommand{\getmd}{\ensuremath{\mathop{\mathbf{getmd}}}\xspace}
\newcommand{\latest}{\ensuremath{\mathop{\mathbf{latest}}}}
\newcommand{\applcbl}{\ensuremath{\mathop{\mathrm{applicable}}}}
\newcommand{\containseq}{\ensuremath{\mathop{\mathrm{contains}}}}
\newcommand{\Cx}{\ensuremath{\texttt{C}}\xspace}
\newcommand{\Xx}{\ensuremath{\texttt{X}}\xspace}
\newcommand{\Tx}{\ensuremath{\texttt{T}}\xspace}
\newcommand{\OEx}{\ensuremath{\texttt{E}}\xspace}
\newcommand{\hole}{\ensuremath{\square}\xspace}
\newcommand{\plugx}[2]{\ensuremath{#1\left[#2\right]}}
\newcommand{\plugCx}[2]{\ensuremath{{\color{violet} #1\left[{\color{black} #2}\right]}}}
\newcommand{\rep}{\ensuremath{\mathtt{rep}}\xspace}
\newcommand{\rV}{\mathop{\mathcal{V}}}
\newcommand{\rX}{\mathop{\mathcal{X}}}
\newcommand{\rT}{\mathop{\mathcal{T}}}
\newcommand{\rC}{\mathop{\mathcal{C}}}
\newcommand{\rVd}{\ensuremath{\rV(\v)}\xspace}
\newcommand{\rXX}[1]{\ensuremath{\rX(#1,\ \mcalld)}\xspace}
\newcommand{\rXd}{\rXX{\Xx}}
\newcommand{\rCC}[1]{\ensuremath{\rC(#1,\ \rdx)}\xspace}
\newcommand{\rCd}{\rCC{\Cx}}
\newcommand{\CanSym}{\ensuremath{\mathop{\boldsymbol{Can}}}\xspace}
\newcommand{\Can}[2]{\ensuremath{\CanSym\left(#1\ \mathbf{as}\ #2\right)}\xspace}
\newcommand{\evalstep}{\ensuremath{\rightarrow}\xspace}
\newcommand{\evalstepc}{\ensuremath{\evalstep^*}\xspace}
\newcommand{\evalstepm}[1]{\ensuremath{\evalstep_{#1}}\xspace}
\newcommand{\stwa}[2]{\ensuremath{\langle #1, #2 \rangle}}
\newcommand{\evalwast}[2]{\ensuremath{#1\ \evalstep\ #2}}
\newcommand{\evalwa}[4]{\ensuremath{\evalwast{\stwa{#1}{#2}}{\stwa{#3}{#4}}}}
\newcommand{\evalwad}[2]{\evalwa{\MT}{#1}{\MT}{#2}}
\newcommand{\evalwastc}[2]{\ensuremath{#1\ \evalstepc\ #2}}
\newcommand{\evalwac}[4]{\ensuremath{\evalwastc{\stwa{#1}{#2}}{\stwa{#3}{#4}}}}

\newcommand{\evalwastm}[3]{\ensuremath{#1\ \evalstepm{#2}\ #3}}
\newcommand{\evalwam}[5]{\ensuremath{\evalwastm{\stwa{#1}{#2}}{#3}{\stwa{#4}{#5}}}}
\newcommand{\evalerrwa}[3]{\ensuremath{#1\ \vdash\ #2\,\evalstep\,#3}}
\newcommand{\evalerrwad}[2]{\evalerrwa{\MT}{#1}{#2}}
\newcommand{\evalfullwast}[2]{\ensuremath{#1\ \overset{c}{\evalstep} \ #2}}
\newcommand{\evalfullwa}[4]{\ensuremath{\evalfullwast{\stwa{#1}{#2}}{\stwa{#3}{#4}}}}
\newcommand{\Gm}{\ensuremath{\Gamma}\xspace}
\newcommand{\typedwa}[4]{\ensuremath{#2\ \vdash_{#1}\ #3\ :\ #4}}
\newcommand{\typedwad}[2]{\typedwa{}{\Gm}{#1}{#2}}
\newcommand{\gm}{\ensuremath{\gamma}\xspace}
\newcommand{\substok}[2]{\ensuremath{#1 \vdash #2}\xspace}
\newcommand{\substokd}{\substok{\Gm}{\gm}}
\newcommand{\nv}{\ensuremath{\nu}\xspace}
\newcommand{\nvs}{\ensuremath{\obar{\nv}}\xspace}
\newcommand{\Ex}{\ensuremath{\texttt{E}}\xspace}
\newcommand{\mdeq}[3]{\ensuremath{#2(#1) \leadsto #3}\xspace}
\newcommand{\mdeqd}{\mdeq{\obar{\sigma}}{\m}{\m'}}
\newcommand{\SpecEnv}{\Phi}
\newcommand{\mspec}[4]{\ensuremath{\vdash^{#1}_{{\color{violet}#2}\leadsto{\color{violet}#3}} #4}\xspace}
\newcommand{\mspecd}{\mspec{\SpecEnv}{\MT}{\MT'}{\mdeqd}}
\newcommand{\expropt}[6]{\ensuremath{#2\ \vdash^{#1}_{{\color{violet}#3}\leadsto{\color{violet}#5}}\ #4 \leadsto #6}}
\newcommand{\exproptd}[2]{\expropt{\SpecEnv}{\Gm}{\MT}{#1}{\MT'}{#2}}
\newcommand{\expropte}[2]{\expropt{\SpecEnv}{}{\MT}{#1}{\MT'}{#2}}
\newcommand{\mdopt}[5]{\ensuremath{\vdash^{#1}_{{\color{violet}#2}\leadsto{\color{violet}#3}} #4 \leadsto #5}\xspace}
\newcommand{\mdoptd}[2]{\mdopt{\SpecEnv}{\MT}{\MT'}{#1}{#2}}
\newcommand{\tableopt}[4]{\ensuremath{\vdash^{#1} #3 \leadsto #4}}
\newcommand{\tableoptd}{\tableopt{\SpecEnv}{\e}{\MT}{\MT'}}
\newcommand{\tableoptref}[4]{\ensuremath{\vdash^{#1}_{#2} #3 \leadsto #4}}
\newcommand{\tableoptrefd}{\tableoptref{\SpecEnv}{\e}{\MT}{\MT'}}
\newcommand{\tableoptexpr}[3]{\ensuremath{\vdash_{{\color{violet}#1}\leadsto{\color{violet}#2}} #3}}
\newcommand{\tableoptexprd}[1]{\tableoptexpr{\MT}{\MT'}{#1}}
\newcommand{\rdx}{\ensuremath{\texttt{rdx}}\xspace}
\newcommand{\vnm}{\ensuremath{\texttt{v}^{\neq \m}}\xspace}

\newcommand{\infers}[4]{#1 |\hspace{-0.48em}\sim_{#2} #3 : #4}

\providecommand{\Alt}{~\vert~}
\providecommand{\Altf}{\faded{\Alt}}

\chapter{Gradual Typing for Julia}\label{ch:juliatypes} 

I have discussed the complexities in typing Julia at length; what remains is
to actually type Julia; now, what remains is to describe how to type Julia itself
and sketch how some abstractions in Julia might be captured.

I extend Julia to create Typed Julia. Typed Julia builds on normal Julia by
adding two constructs:
\begin{itemize}
	\item Typed methods that are statically guaranteed to be correct with respect to their type arguments.
	\item Protocols that require that a suitable implementation of a method exists for some set of types.
\end{itemize}
These typed features coexist with untyped Julia code, allowing programs
to be incrementally specified and typed.

The core of Typed Julia is the ability to write typed methods. Typed methods
are syntactically identical to untyped methods, differing only in that they
are annotated with the \jlinl{@typed} annotation. Only methods that are
annotated as \jlinl{@typed} will be statically checked.

As one kind of abstraction, Typed Julia also supports protocols. A protocol,
declared using the annotation \jlinl{@protocol}, simply asserts the existence
of methods that can handle every instantiation of some signature. The type
system then verifies this assertion using a mechanism similar to checking
completeness of pattern matching.

The guarantee of Typed Julia is that ``typed code will not go wrong,'' with
the exception of method ambiguities. This safety property is quite robust: it
can be broken neither by untyped code nor by use of \eval (up to a point);
only when new methods have been dynamically added and execution has returned
to the top level can the safety property be broken. Therefore, Typed Julia
provides an unusually strong soundness guarantee for a gradually typed
language.
\cbstart
Consider a small example of (untyped) Julia code:
\cbend
\begin{lstlisting}
abstract type AbstractArray end
struct List <: AbstractArray
	array::Array{Int, 1}
end
struct Range <: AbstractArray
	start::Int
	end::Int
end
length(v::List) = length(v.array)
size(u::Range) = u.end - u.start
\end{lstlisting}
Here I define two kinds of array-like-thing: a list (backed by an
array) and a range (backed by a start and an end index). The \jlinl{List}
supports a \jlinl{length} method that returns the length of the underlying
array, while \jlinl{Range} has \jlinl{size} that determines the number of
elements in the range. 

Suppose, now, that I wanted to implement a method \jlinl{array_like} that
allocated a basic array that is the same size as the input, whatever that
input array is. One, flawed, implementation might be to say
\begin{lstlisting}
array_like(a::AbstractArray) = Array(length(a))
\end{lstlisting}
This implementation has a simple bug: it will fail if it is ever given
a \jlinl{Range}. However, if I only ever tested \jlinl{array_like} with
\jlinl{List}s I would only find out about this after someone tried it with
a \jlinl{Range} and complained.

Bugs where a program assumes that a special-case function is universally
applicable are extremely common in Julia. One real example is the use of
\jlinl{for i in 1:length(A)} to loop through every index in an
array~\footnote{\href{https://towardsdatascience.com/1-length-a-considered-harmful-or-how-to-make-julia-code-generic-safe-ac7b39cfc2f0}{1:length(A) considered harmful — or how to make Julia code “generic safe”}}.
Not all \jlinl{AbstractArray}s are integer indexable in the first place and
not all of those arrays start at 1. However, most libraries and users test
against basic \jlinl{Array} which are both integer indexable and start at 1 so
do not notice the problem.

This problem of \jlinl{length} is one of both definition and usage. On the
definition side there was nothing to tell us that users might expect every
implementation to have a \jlinl{length}. On the usage side, nothing was
obviously wrong when I called a function that only exists on \emph{some}
abstract arrays. To fix this problem I need to check both definition and
usage sites.

Let us start by examining how Typed Julia handles usage sites through gradual
typing. In Typed Julia, only methods with the \jlinl{@typed} annotation need
to pass the type checker. I need only to check the methods that I am currently
implementing. To start with, then, let's add the \jlinl{@typed}
annotation to our new \jlinl{array_like} method:
\begin{lstlisting}
@typed array_like(a::AbstractArray) = Array(length(a))
\end{lstlisting}
Now I run the Typed Julia checker on our program giving us the error
\begin{lstlisting}
No implementations of "length" found for type (::AbstractArray). Suitable implementations:
		length(::List)
\end{lstlisting}
The type checker has identified that not all \jlinl{AbstractArray}s
implement \jlinl{length}; while one exists for \jlinl{Range}, no such
method exists that can handle every \jlinl{AbstractArray}. I can use
typed Julia's other feature, protocols, to ensure that such a method
exists. To statically require the existence of the \jlinl{length} method, I add the declaration
\begin{lstlisting}
@protocol length(a::AbstractArray)::Int
\end{lstlisting}
which requires that there is a \jlinl{size} function for every
\jlinl{AbstractArray} that returns an integer. When I add this to
our example, the protocol definition now statically produces the error
\begin{lstlisting}
Protocol length not satisfied; missing implementation(s) for:
	(::Range)::Int
\end{lstlisting}
identifying that there does not exist a method \jlinl{length(::Range)::Int}.
If I fix this by adding
\begin{lstlisting}
length(r::Range) = size(r)
\end{lstlisting}
to our definitions, then the protocol definition error vanishes.

The typechecker was able to determine that there was an implementation of
\jlinl{size} available for every \jlinl{AbstractArray} and could infer that
their return types adhered to the protocol specification. Now that we
have a protocol for \jlinl{size} our typed implementation of \jlinl{array_like}
passes successfully.

I am then left with the following Typed Julia program, showing the two
key additions.
\begin{lstlisting}
abstract type AbstractArray end
struct List <: AbstractArray
	array::Array{Int, 1}
end
struct Range <: AbstractArray
	start::Int
	stop::Int
end
length(v::List) = length(v.array)
size(u::Range) = u.stop - u.start
length(r::Range) = size(r)
@protocol length(a::AbstractArray)::Int

@typed array_like(a::AbstractArray) = Array(length(a))
\end{lstlisting}
This example has demonstrated both how Julia code can be gradually type checked
(with a typed function calling untyped methods) and how I can canonize abstractions
within the protocol system. Now, let us examine how protocols work in more detail and
some of the design decisions that went into them.

\paragraph{Protocols}

Considering \jlinl{AbstractArray} again, the \jlinl{size} protocol is
both real and not alone. Julia uses \jlinl{AbstractArray} for all
``array-like'' things and has a suitably set of protocols that all
implementations must support. Every scalar \jlinl{AbstractArray}
implementation (that I will refer to as \jlinl{A}) is required by the
documentation to implement two methods, as shown in table~\ref{absarray}.

\begin{table}[h]
\centering
\begin{tabular}{|l|l|}\hline
\xt{size(A)} & Returns a tuple containing the dimensions of A\\ \hline
\xt{getindex(A, ::Int)} & Linear scalar indexing \\ \hline
\end{tabular}
\caption{\lstinline|AbstractArray| interface}
\label{absarray}
\end{table}

No single implementation of all of these methods exists for all
\jlinl{AbstractArray}s; each subtype of \jlinl{AbstractArray} is expected to
implement its own version. The need for this is very straightforward: we cannot
possibly write an implementation of \jlinl{getindex} that works the same way
for both an sequential array implementation and a linked list, for example.

Protocols of this sort in Julia are very common in libraries. For example,
MathOptInterface, an abstraction layer over numerous numerical optimization
solvers, exposes more than a hundred different protocols. One example is
\jlinl{optimize!(dest::AbstractOptimizer, src::ModelLike)}, which uses the
optimizer \jlinl{dest} to optimize the model \jlinl{src}. There must be an
implementation of \jlinl{optimize!} for \emph{every} \jlinl{AbstractOptimizer}
and \jlinl{ModelLike}. Implementations of \jlinl{optimize!} tend to specialize on
\jlinl{dest} while being generic over \jlinl{src} while using protocols
exposed by \jlinl{ModelLike} to interact with the problem being optimized.

Many protocols are implemented outside of the original package. As evidenced
in \figref{fnhist} packages frequently implement functions from other methods
for their own types. \figref{fnover} shows that math and collection types are
the most commonly implemented. Protocol users are also common:
\figref{fig:AMCS} suggests shows that many call sites dispatch to more than
one implementation, demonstrating that reliance on a protocol-describable abstraction
is common in Julia programs.

In spite of this adoption and wide import, protocols in Julia are ill-defined.
Most exist solely in the form of English-language documentation or even just
implicitly in the code itself; no machine-legible form is available. While the
types and the implementations should exist in the programs themselves, there
is no trace of the abstract notion of the protocol. 

Introducing mechanically-checked protocols begets a key question: how do we
declare and discover them? Two methods present themselves: if I wanted to
support as much existing code as possible, I could try to find protocols from
the bottom up, identifying protocols from usages. In contrast, I could also
require protocols to be  explicitly declared. As shown before, I chose to use
explicitly-declared protocols at the expense of being able to easily typed
existing code but the alternative of bottom-up protocols deserves closer
attention.

Bottom-up checking was the original approach I took to identifying
protocols~\cite{nool17}. If I assume that all protocol \emph{implementations}
in a program are correct (and tractable to the incomplete static analyzer)
then this is an ideal solution: no new type annotations are needed and no new
specifications are required; protocols effectively arise implicitly from their
use in the code.

\begin{jllisting}
abstract type Number end
struct Int <: Number end

f(::Int) = 2
g(y::Number) = f(y)
\end{jllisting}

Consider the above example. I have two types (\jlinl{Number} and \jlinl{Int})
and two methods; \jlinl{f} can handle only \jlinl{Int}s, while \jlinl{g} is
supposed to handle any kind of \jlinl{Number}. Since there is an
implementation of \jlinl{f} for every kind of number, it follows that this
implementation is so far type-safe.

\begin{jllisting}
abstract type Number end
struct Int <: Number end
struct Float32 <: Number end

f(::Int) = 2
g(y::Number) = f(y)
\end{jllisting}

If I then add a new implementation of \jlinl{Number}, \jlinl{Float32}, then
the invocation now becomes not-type-safe. Superficially, this feels fine:
there is no implementation of \jlinl{f} to use from \jlinl{g}, so I error
there. However, consider the evolutionary process I took to get here: we
first expected there to be a \jlinl{f} for every kind of \jlinl{Number} and
there was, so all was good. The modification I made to break the program was
to add a new type. I did not touch either \jlinl{f} or \jlinl{g}, so breaking
\jlinl{f} for a new definition potentially far away from it feels unfair.  In
effect, I treat the definition as canonical and uses cases as faulty, even
when the desired semantics is the other way around.

Two further issues arise from the bottom-up approach:
\begin{itemize}
	\item Use-site checking does not establish if an unused implementation is correct.
	Suppose for a second that I never called \jlinl{f} with an abstract \jlinl{Number}
	in a typed position; I would never know that I had violated the \jlinl{f} protocol
	when I added \jlinl{Float32} and would only realize that post facto. Moreover, there is
	no explicit declaration that the protocol exists or must be adhered to; it would be
	easy for a programmer new to the project to miss, misuse, or fail to implement some 
	protocol.
	\item The design lends itself to ``spooky action at a distance.'' As I saw earlier,
	when I broke the \jlinl{f} protocol with \jlinl{Float32} the error occurred when we
	called \jlinl{f} from within \jlinl{g}, implying that the problem lies there. Errors
	should point at the actual cause of the failure, rather than merely where the program
	might go wrong as a result. I should ideally point at the new definition \jlinl{Float32} as the cause
	of the failure.
\end{itemize}

This problem becomes particularly stark with an eye towards the reality that most
existing Julia protocols are violated \emph{somewhere}. Even \jlinl{getindex} has some
nonconformant implementations: \jlinl{LogicalIndex}, used to represent an array of
indices masked by a boolean value, has no implementation of \jlinl{getindex}.
As a result, many ``obvious'' function calls suddenly become use-site errors in
spite of the true fault lying with the implementation.

\paragraph{Eval}

Using \eval Julia programmers can insert new definitions of methods
and types at any point in the program. Let us consider a small example.

\begin{jllisting}
abstract type A end
struct B <: A end
f(::A) = 2
g(a::A) = f(a)

function main()::Int
	g(B())
end
main()
\end{jllisting}

It should be possible to type-check this program; the return type
of both \jlinl{f} and \jlinl{g} should be \jlinl{Int}, trivially.

What should happen, though, if I rewrote the program as
\begin{jllisting}
abstract type A end
struct B <: A end
f(::A) = 2
g(a::A) = f(a)

function main()::Int
	eval(:(f(::B) = "hello world"))
	g(B())
end
main()
\end{jllisting}
There is nothing \emph{semantically} wrong with this program---but what should 
the return type of \jlinl{g(B())} now be? The type checker cannot reasonably
analyze all \jlinl{eval}ed expressions, so I would most likely determine that
\jlinl{main} is type-safe. However, if the invocation of \jlinl{g} ends up
calling our new most specific method \jlinl{f(::B)} then the returned value will be a
\jlinl{String}, not an \jlinl{Int}. 

A type checker written against this ``naive'' dynamics would not be
particularly useful; while it could guarantee that a method \emph{exists}
(since while Julia does allow definitions to be deleted it is so uncommon as
to be easily prohibited dynamically) the type checker needs to know about
\emph{every} potential method that could be called in order to determine what
their return types might be. A type checker with an open-world assumption can
only say that the return type of any method call is \jlinl{Any}, which is not
especially useful. Therefore, a ``closed-world'' assumption is practically
mandatory.

Luckily for us, Julia has also ran into this issue internally. I talked
earlier about how Julia tries to one-shot compile every method call into a
statically typed version. This process would be broken just as badly by
dynamically added methods as the type checker would be. Consequently, Julia
``fixes'' the set of visible methods to those that were defined the last time the
code reached the top level (or was invoked with the \jlinl{invokelatest}
function). Thus, with the right selection of methods I can make a
closed-world assumption that pans out in reality.

Julia's semantics are illustrated with a small modification of the above example. If we
add a second call to \jlinl{main} like the following:
\begin{jllisting}
abstract type A end
struct B <: A end
f(::A) = 2
g(a::A) = f(a)

function main()::Int
	eval(:(f(::B) = "hello world"))
	g(B())
end
main()
main() # breaks the return type
\end{jllisting}

then execution returns to the top level and allows \jlinl{main} to see the
newly-added definition. The second invocation of \jlinl{main} now returns the
string \jlinl{"hello world"} and violates the return type annotation. Dynamic
guards on return types are therefore needed once the set of visible methods
has evolved past what I originally checked against, but no earlier.

Type checking for Julia therefore is something of a ``hybrid'' between
a traditional gradual type system (which needs to insert casts to ensure
type safety) and that for a statically typed language (which does not). As
long as the program's ``world'' matches the one that was originally checked
against no casts are needed. However, if the world ``moves on'' from that state
through the insertion of new definitions via \jlinl{eval} then the return type
becomes unsound and needs dynamic checking.

If I contextualize the type system for Julia within the taxonomic framework
that I described with the KafKa language~\cite{chung2018kafka} then it is a
hybrid of a fully-static type checker, the concrete, and the transient
semantics. If the program is running in the original world age that was used
for type checking, then no checks whatsoever are needed. In effect, Typed
Julia programs in the original world age have a simple soundness guarantee, 
as they would in a fully-typed language. If the program has moved on from that
world age then Typed Julia must check return types as the transient semantics does, 
but does so using the concrete notion of type membership where tags are checked
rather than the superficial identity of the value.

\paragraph{Ambiguities} The final challenge that I need to consider is the
question of ambiguities. As a simple example, consider two versions of the
\jlinl{+} function:

\begin{lstlisting}
+(::Int, ::Number) = ...
+(::Number, ::Int) = ...
\end{lstlisting}
Now, suppose that I call \jlinl{+(1, 2)}. \jlinl{1} is an \jlinl{Int} and
\jlinl{2} is a \jlinl{Number}, so I could apply the first, but \jlinl{1} is also
a \jlinl{Number} while \jlinl{2} is a \jlinl{Int} so the second could apply as well.
Additionally, the two definitions are not related by subtyping, so neither can be called
more specific than the other. Thus, no singular most specific implementation exists,
and Julia errors at \jlinl{+(1, 2)}. Should I statically call out either these definitions
of \jlinl{+} as ones that could be ambiguous or should I treat the invocation as an error
if an ambiguity is possible?

If I look at the prior work many typed languages with multiple dispatch
identified ambiguities. Eliminating the potential for ambiguities was a major
goal and challenge for both Fortress and
Cecil~\cite{Allen11,DBLP:journals/pacmpl/ParkHSR19,Litvinov98}. Julia's
experience suggests, however, that ambiguity checking may actually be
undesirable, particularly at definition sites.

Julia, at one point, included a definition-time ambiguity detection heuristic. 
However, unlike Fortress and Cecil, Julia decided to \emph{remove} their
definition-time ambiguity
heuristic~\footnote{https://github.com/JuliaLang/julia/pull/16125}. From the
perspective of static typing, this decision seems unusual; why suffer runtime
errors when they could have been caught statically?

One of Julia's selling points has been the ``unreasonable effectiveness'' (in their words)
of multiple dispatch~\footnote{\url{https://www.juliaopt.org/meetings/santiago2019/slides/stefan_karpinski.pdf}}.
When the Julia developers talk about this ``effectiveness,'' they are referring
to how multiple dispatch solves the expression problem in a compositional way
wherein different libraries can provide the same resuable abstraction for
their own structures and thereby be composed.

Let's revisit the \jlinl{+} function again. Suppose I implemented a
library that had polynomial types. I would like to be able to add a
polynomial to a polynomial. Additionally, I want to be able to add simple
numbers to a polynomial to get another polynomial with a larger constant.
Moreover, because our polynomial acts like a \jlinl{Number} I subtype
\jlinl{Number}. An (extremely simplified) implementation of this in Julia
could look as follows.
\begin{lstlisting}
struct Polynomial <: Number
	values::Tuple
end

Base.+(x::Polynomial, y::Polynomial) = Polynomial(x.values .+ y.values)
Base.+(x::Polynomial, y::Number) = Polynomial((x.values[1:end-1]..., x.values[end] + y))
\end{lstlisting}
I can then use a \jlinl{Polynomial} as \jlinl{Polynomial((1,2,3)) + 3} which
gives us \jlinl{Polynomial(1,2,6)}. So far, so good.

The problem is that the ``unreasonable effectiveness'' of Julia can burn us.
Practical Julia programs frequently compose different libraries with one another,
therein creating ambiguities. Suppose that there's then an automatic differentiation
library that defines a dual number and an addition method.
\begin{lstlisting}
struct Dual <: Number
	value
	epsilon
end
Base.+(x::Number, y::Dual) = Dual(y.value+x, y.epsilon)
\end{lstlisting}
If I then imported both \jlinl{Polynomial} and \jlinl{Dual} into a third
project and then try to invoke \jlinl{Polynomial(1,2,3) + Dual(0, 0)}, I get
an ambiguity. This
function call is trivially ambiguous: I cannot decide whether to call the
implementation for \jlinl{Polynomial} or \jlinl{Dual}. Therefore a sound
static ambiguity checker should reject these definitions. So far, so good.

The problem with this answer is entirely practical. It is common for libraries
to add special implementations of shared functionality for their own specific
use cases and for those implementations to be potentially ambiguous with
methods from other
packages~\footnote{\url{https://github.com/JuliaStats/DataArrays.jl/issues/51},
\url{https://github.com/JuliaStats/DataArrays.jl/issues/77}}. Resolving these
ambiguities requires adding a suitable more-specific method for \emph{every 
combination of types}. Therefore, disambiguation methods quadratic in the
number of libraries need to be added, requiring both:
\begin{itemize}
	\item a hilarious number of additional implementations,
	\item and perfect awareness of all other extensions of the same function in all other libraries.
\end{itemize}
These requirements are impractical and led to the removal of ambiguity
checking as a default in Julia (though it still exists for use in test
cases)~\footnote{\url{https://github.com/JuliaLang/julia/issues/6190}}. Julia's
experience was that most ambiguity detection were false, and did not eventuate
in actual executions.

Arguably, this realization that ambiguity checking at definition time is
impractical is a consequence of Julia realizing multiple dispatch at scale.
Ambiguity checking makes sense in the context of a single library or project
where a single team of developers controls the entire system. In the
aforementioned cases of both Cecil and Fortress the largest programs written
in the language by an overwhelming margin were their respective compilers. In
these programs ambiguities were clear bugs and the responsibility of only a
single team. In contrast, in Julia's much larger ecosystem of
loosely-interacting developers the very compositionality of multiple dispatch
makes the potential for ambiguities more common and makes it harder to resolve
ambiguities organizationally.

The final reason for not performing ambiguity checking in Julia is that actual
ambiguous calls are \emph{rather uncommon in practice}: most of the time only
one library is being used at a time or composition is relatively simple,
therein avoiding the ambiguity. When an ambiguity error is encountered it is
usually a sign of bad library design or can be resolved easily by the user
adding a suitable disambiguation method themselves---therein addressing the
``awareness of all other implementations'' issue previously mentioned.

As a result of this experience I do not include static ambiguity checking
as a goal for static typing in Julia; method calls in statically-typed
Julia may still fail dynamically with an ``ambiguous method call'' error.

I will now present the type system for Julia in two major parts:
\begin{itemize}
	\item First, a theoretical system that types the \juliette~\cite{belyakova2020world}
	calculus and introduces the key operations and metafunctions.
	\item Second, an implementation that allows programmers to utilize the
	proposed system and allows it to be applied to various use cases.
\end{itemize}

\section{Related Work}

Many of the challenges inherent to typing Julia are not unique. Where does
Typed Julia fall in the broader landscape of gradual typing and how does it
relate to the usual properties of a gradual type system? Additionally, the
combination of static typing and dynamic metaprogramming is not new, nor is
typing for multiple dispatch. Let us consider how the prior work in this space
bit off these problems.

\subsection{Gradual Typing}

\renewcommand{\t}{\ensuremath{t}\xspace}
\renewcommand{\tp}{\ensuremath{t'}\xspace}

\renewcommand{\a}{\ensuremath{a}\xspace}
\renewcommand{\e}{\ensuremath{e}\xspace}

\noindent 

Gradual type systems aim to allow the incremental addition of types to untyped
code. Siek et al~\cite{siek2015refined}, for example, claims five criteria for
gradual typing:

\begin{enumerate}
\item Equivalency to normal static typing for fully typed terms;
\item Equivalency to dynamic typing for fully untyped terms;
\item Type soundness;
\item Statically-typed code will not be blamed for type errors;
\item Any set of type annotations can be removed from a partially typed program without changing program behaviour (the gradual guarantee).
\end{enumerate}
Criteria 1 and 2 are uncontroversial. Criteria 3, 4, and 5, however,
are somewhat trickier.

An example of the disagreement comes from Siek's homepage, wherein he states that
\begin{displayquote}
I've been fortunate to see some of my ideas get used in the software industry:
\begin{itemize}
\item Microsoft created a gradually-typed dialect of JavaScript, called TypeScript.
\item Facebook has added gradual typing to PHP. [...]
\end{itemize}
\end{displayquote}
Neither TypeScript nor Hack (Facebook's PHP type system) satisfy criteria 3 or 4.
Both implement so-called \emph{optional} type systems, wherein they erase type annotations
after static checking. As a result, they are \emph{unsound} (untyped code may violate typed
assumptions at any time) and \emph{have no concept of blame}. As a result, I would conclude
that TypeScript and Hack are not gradually typed, contradicting Siek's own description. This
internal disagreement about what it means to be a gradual type system is reflective of a broader
lack of consensus around the term.

Type soundness alone is a point of much research. Despite being foundational to the concept of
static typing, when untyped code is introduced  At least 5 distinct concepts have been
proposed, each providing different theoretical and practical tradeoffs. Options for how to
define soundness range from ``nothing'' (as in the optional approach) to ``type inhabitants must
carry type tags that are a subtype of the statically declared type'' (the concrete approach, used
in languages like C\#).

This then plays into the discussion about blame. The concept of blame is to redirect errors
created by untyped code that manifest in typed code back into the untyped source. For example,
if an untyped function returns an ill-typed value to a typed caller the untyped function should
be blamed. However, blame inherently depends on what soundness guarantee the system provides, as that
dictates where and when errors will occur. For example, blame makes little sense in a concrete setting
(as only statically checked behaviors are allowed to pass type boundaries and typed mutable references are 
checked on write), but is vital for other semantics.

Finally, I have the gradual guarantee. The idea behind the gradual guarantee
is that it captures the migration process, wherein an untyped program is
incrementally typed while remaining observationally identical. It expresses
this process in the opposite direction: with a type system that satisfies the
gradual guarantee I can always remove types while retaining semantic
equivalence. Thus, if I have an untyped program that only needs annotations
to be typed then a  gradual guarantee-compliant type system would let us add
the annotations in any order and anywhere I so chose.

The problem is that in practice few untyped programs are actually typable
without modification. The TypeScript documentation, for example, explicitly
discusses common modifications needed to Javascript programs for them to be
typable; Takikawa et al's~\cite{takikawa2016sound} benchmark suite required
numerous code modifications beyond simple annotation insertion in order to
satisfy the type checker. Few programmers write perfectly typable untyped code
without the aid of a type checker. The programs that the gradual guarantee
applies to are, in reality, those that have \emph{already} been typed---not
those that have yet to be typed. 

Julia makes answering this question easy. Julia's existing type checks work by
comparing the runtime type tag associated with values to the type annotations
applied to methods, and the guarantees about dispatch follow from such. Thus,
the concrete semantics is a very straightforward choice for Julia. In turn,
this choice means that blame is immaterial and that (without considerable
runtime modifications, as in~\cite{Muehlboeck2017}) I can not satisfy the
gradual guarantee. However, as Julia argument type annotations are already
used for method dispatch and cannot be removed without modifying the program's
semantics, no Julia type system could realistically satisfy the gradual
guarantee. Therefore, while the proposed type system for Julia is not gradual
by some definitions of the term, it allows for mixing typed and untyped code
to interoperate while closely adhering to existing programmer expectations
about type annotations.

The concrete semantics alone are not sufficient for Julia, however, due to the
challenges posed by multiple dispatch as mentioned earlier. The intersection
of multiple dispatch and gradual typing is a relatively unexplored domain. The
primary realized example is for the Dylan
language~\cite{mehnert2010extending}. Gradual typing in Dylan is fundamentally
different than what I describe here, though; Mehnert's approach is built
around a nonlocal type inference algorithm that tries to build out a set of
constraints begotten by a realized program and solve them. On one hand, this
dramatically simplifies several problems, such as the need for protocols or
dealing with underspecified argument types. At the same time, however, it
provides many fewer developer-facing benefits of typing (such as improved
documentation and robustness to changes) compared to the local system that I
present here.

On the vein of nonlocal inference another topic that deserves mention is soft
typing. Soft type systems aim to automatically infer types for untyped
programs using type inference~\cite{fagan1991soft}; the dream was that one
could provide a completely untyped program and it would be inferred to a 
fully typed one for performance. Soft typing has similar problems to the
gradual guarantee, however: real untyped programs are rarely written in every
detail so that they would type check if you just worked hard enough. Soft
typing, in particular,  has issues with untypable operations: soft type system
rely on nonlocal unificiation-based inference algorithms that may take some
time to conclude that something has gone wrong. Whether the unification
process provides an error that makes it clear where that ``wrong'' was depends
on the structure of the program. Practically, then, when writing code for soft
typing one must take as much care as they would if they were writing code for
a traditional gradual type system with static type annotations but while
suffering much worse error messages.

\subsection{Eval}

Most programming languages control \emph{where} definitions are visible, as part
of their scoping mechanisms. Controlling \emph{when} function definitions become
visible is less common.
Languages with an interactive development environment had to deal with the
addition of new definitions for functions from the start~\cite{lisp}.
Originally, these languages were interpreted. In that setting, allowing new
functions to become visible immediately was both easy to implement and did not
incur any performance overhead.

Just-in-time compilation changed the performance landscape, allowing dynamic
languages to have competitive performance. However, this meant that to generate
efficient code, compilers had to commit to particular versions of functions.
If any function is redefined, all code that depends on that function must
be recompiled; furthermore, any function currently executing has to be
deoptimized using mechanisms such as on-stack-replacement~\cite{Hoelzle92}.
The drawback of deoptimization is that it makes the compiler more complex
and hinders some optimizations. For example, a special \c{assume} instruction
is introduced as a barrier to optimizations by \cite{popl18}, who formalized
the speculation and deoptimization happening in a model compiler.

Java allows for dynamic loading of new classes and provides sophisticated
controls for where those classes are visible. This is done by the class-loading
framework that is part of the virtual machine~\cite{LB98}. Much research
happened in that context to allow the Java compiler to optimize code in the
presence of dynamic loading. Detlefs~\cite{detlefs99} describe a technique, which they
call preexistence, that can devirtualize a method call when the receiver object
predates the introduction of a new class. Further research looked at performing
dependency analysis to identify which methods are affected by the newly added
definitions, to be then recompiled on demand~\cite{nguyen1996interprocedural}.
Glew~\cite{glew2005type} describes a type-safe means of inlining and
devirtualization: when newly loaded code is reachable from previously optimized
code, these optimizations must be rechecked.

Controlling \emph{when} definitions take effect is important in dynamic software
updating, where running systems are updated with new code~\cite{lee}.
Hicks~\cite{hicks} introduce a calculus for reasoning about representation-consistent
dynamic software updating in C-like languages. One of the key elements for their
result is the presence of an \c{update} instruction that specifies when an
update is allowed to happen. This has similarities to the world-age mechanism
described here.

Substantial amounts of effort have been put into building calculi that support
\eval and similar constructs. For example, Politz~\cite{politz12} described the
ECMAScript 5.1 semantics for \eval, among other features. Glew~\cite{glew2005method}
formalized dynamic class loading in the framework of Featherweight Java, and
Matthews~\cite{matthews2008operational} developed a calculus for \eval in Scheme. These
works formalize the semantics of dynamically modifiable code in their respective
languages, but, unlike Julia, the languages formalized do not have features
explicitly designed to support efficient implementation.

Julia's use of the world-age mechanism, the method tables that I mentioned
earlier, allows Julia to ``lock down'' what methods might be visible at any
point in time. In this manner, Julia dramatically simplified the
implementation of their JIT compiler.

\subsection{Static Typing of Multiple Dispatch}

Static type systems aim to identify and rule out classes of dynamic error. In
a multiple dispatch context, this entails identifying code wherein one of the
two aforementioned errors, no applicable methods and ambiguous method call,
could occur. Practically, static type systems also enable code completions and
facilitate automatic refactoring. 

Static typing for multiple dispatch is an old idea, with an early
comprehensive concept put forward by \cite{agrawal1991multimethods}, which
describes a type system able to eliminate both no method found and ambiguous
method errors. Agrawal focuses on ambiguous method errors, for as in
comparison, it is easier to identify cases were no method exists versus when
multiple ambiguous methods apply. They describe an algorithm designed to
statically identify cases where an ambiguous invocation may occur and how
likely these cases are under different language semantics. Notably, however,
they focus primarily on systems in the vein of CLOS, which add declaration
order as a means of additional disambiguation beyond subtyping; as a result,
they are able to frequently reject ambiguities in cases where Julia would be
ambiguous.

The Cecil language~\cite{chambers1992object} is the statically typed language with the best analogy to Julia. Cecil features the same external (not associated with any one object) methods and a similar polymorphic type language to Julia's. As a result, its type system can serve as a point of reference for the design of a type system for Julia.

Cecil's type system went through several iterations from the earliest versions described in passing in~\cite{chambers1992object}, further expounded upon to a relatively comprehensive cover of the language in~\cite{chambers1995typechecking}, and finally extended to support constraint based polymorphism~\cite{litvinov1998contraintpolymorphism}. The project aimed to be evolved into the Diesel language (which simplified the object model and implemented a module system), but no publications were forthcoming.

The most relevant work for type checking in Julia is~\cite{chambers1995typechecking}, which describes the core of Cecil's type system. Typechecking in Cecil is broken into two components: \emph{implementation} and \emph{client}.

Implementation-side checking in Cecil is the core of the approach. The issue arises from how Cecil addresses one of Julia's key correctness issues: signatures. In Cecil, inter-library behavioral specifications can be written as a list of methods that are required to work for all possible type instantiations. For example, I could specify addition as
\begin{lstlisting}
type num;
type int subtypes num;
type fraction subtypes num;
signature +(num, num): num;
signature +(int, int): int;
signature +(fraction, fraction): fraction;
\end{lstlisting}
Here, I say that it must be possible to add any two numbers, regardless of their types, producing an arbitrary number. Similarly, I also specify that adding two ints must produce an int and the same for fractions. Cecil statically guarantees that if I added, say, \lstinline{type irrational subtypes num} that it must be possible to add an \lstinline{int} and a \lstinline{irrational} to get a subtype of \lstinline{num}.

Ensuring correctness against a set of signature declarations in Cecil requires that implementations satisfy three properties: conformance, completeness, and consistency. Conformance requires that the argument and result types for each method implementing the signature must be compatible with the types specified by the signature; for example, our implementation of \lstinline{+} for \lstinline{irrational} cannot return a \lstinline{string}.  Completeness enforces that potential signature instantiations must implement the signature; I must implement \lstinline{+} for \lstinline{irrational} because \lstinline{irrational} is a \lstinline{num}. Between them, conformance and completeness rule out message not understood errors. Finally, consistency requires that no ambiguities may exist among different implementations of the same signature; I cannot implement \lstinline{+(fraction, num)} and \lstinline{+(num, fraction)} without \lstinline{+(fraction, fraction)}, for otherwise the latter case would be ambiguous.

Cecil statically requires that all implementations satisfy these three
properties. As a result, client-side checking in Cecil is very simple: as long
as the static type system can guarantee that either some concrete
implementation exists or that some signature is employed, then the call can be
considered safe.

Neither Cecil nor Fortress considered the question of ``what happens with
dynamically-generated code?'' Both systems were, effectively, research projects
that had few external users and were designed from the start to support static
typing. Consequently, \eval and similar dynamic metaprogramming was not a major
concern.

\section{A Core Calculus for Julia}\label{sec:wa-formal}

I formalize my type system for Julia using the \juliette calculus. \juliette
was  originally used in our paper formalizing world age in
Julia~\cite{belyakova2020world}. The calculus focuses on capturing how method
invocation in Julia works with an eye towards what it means to add new methods
and when they can be called. I will first describe the basic \juliette
calculus, define the static semantics for typed \juliette, then consider the
dynamic semantics of typed \juliette alongside the two key correctness 
properties.

\juliette uses \emph{method tables} to represent sets of methods available
for dispatch. The \emph{global table} is the method table that records all
definitions and always reflects the ``true age'' of the world; the global
table is part of \juliette program state. \emph{Local tables} are method
tables used to resolve method dispatch during execution and may lag behind the
global table when new functions are introduced. Local tables are then baked
into program syntax to make them explicit during execution. As in Julia,
\juliette separates method tables (which represent code) from data: as
mentioned in \chapref{ch:julialang}, the world-age semantics only applies to code.
As global variables interact with \eval in the standard way, I omit them from
the calculus.

The treatment of methods is similar in both \juliette and Julia up to (lexically) local method definitions.
In both systems, a generic function is defined by the set of methods with the
same name. In Julia, local methods are syntactic sugar for
global methods with fresh names. For simplicity, I do not model this
aspect of Julia:  \juliette methods are always added to the global method table.
All function calls are resolved using the set of methods found in the
current local table. A function value \m denotes the name of a function and
is not itself a method definition. Then, since \juliette omits global variables,
its global environment is entirely captured by the global method table.

Although in Julia \eval incorporates two features---top-level evaluation and
quotation\footnote{Represented with the \c{\$} operator in Julia, as in
\texttt{eval(:(g() = \$x))} in \figref{fig:eval-methods}.}---only
top-level evaluation is relevant to world age, and this is what I model
in \juliette. Instead of an \eval construct, the calculus has
operations for evaluating expressions in different method-table contexts.
In particular, \juliette
offers a \emph{global evaluation construct} \evalg{\e} (pronounced ``banana
brackets'') that accesses the most recent set of methods. This is equivalent
to \eval's behavior, which evaluates in the latest world age.
Since \juliette does not have global variables, \evalg{\e} reads
from the local environment directly instead of using quotation.

Every function call \mcall{\m}{\vs} in \juliette gets resolved in the closest
enclosing local method table \MT by using an \emph{evaluation-in-a-table}
construct \evalt\MT{\mcall\m\vs}.
Any top-level function call first takes a snippet of the current global table
and then evaluates the call in that \emph{frozen} snippet.
That is, \evalg{\mcall\m\vs} steps
to \evalt\MT{\mcall\m\vs} where \MT is the current global table. Thus, once
a snippet of the global table becomes local table,
all function calls that ensue from the body of of \mcall{\m}{\vs}
will be resolved using this table, reflecting the fact that a currently
executing top-level function call does not see updates to the global table.

\juliette is parameterized over values, types, type
annotations, a subtyping relation, and primitive operations. Only minimal
assumptions are needed for these primitives.

\subsection{Syntax} 

\begin{figure}
\[\footnotesize
\begin{array}{ccl@{\qquad}l}
    \\ \e & ::= & & \text{\emph{Expression}}
    \\ &\Alt& \v & \text{value}
    \\ &\Alt& \x & \text{variable}
    \\ &\Alt& \seq{\e_1}{\e_2} & \text{sequencing}
    \\ &\Alt& \primcalld{\es} & \text{primop call}
    \\ &\Alt& \mcall{\e}{\es} & \text{function call}
    \\ &\Alt& \md & \text{method definition}
    \\ &\Alt& \evalg{\e} & \text{global evaluation}
    \\ &\Alt& \evalt{\MT}{\e} & \text{evaluation in a table}
    \\
    \\ \p & ::= & \evalg{\e} & \text{\emph{Program}}
    \\
    \\ \md & ::= & \mdefd & \text{\emph{Method definition}}
\end{array}\hspace{5mm}
\begin{array}{ccl@{\qquad}l}
    \\ \v & := & \ldots & \text{\emph{Value}}
    \\ &\Alt& \skp & \text{unit value}
    \\ &\Alt& \m & \text{generic function}
    \\
    \\ \g & := & \ldots & \text{\emph{Type tag}}
    \\ &\Alt& \Unit & \text{unit type}
    \\ &\Alt& \mty  & \text{type tag of function \m}
    \\
    \\ \t & := & \ldots & \text{\emph{Type annotation}}
    \\ &\Alt& \top & \text{top type}
\end{array}
\]
\caption{Surface syntax}\label{syntax}
\end{figure}

The surface syntax of \juliette is given in \figref{syntax}. It includes
method definitions \md, function calls~\mcall{\e}{\es}, sequencing
\seq{\e_1}{\e_2}, global evaluation \evalg{\e},
evaluation in a table \evalt{\MT}{\e}, variables \x, values~\v,
primitive calls~\primcall{l}{\es}, type tags \g, and type annotations \t.
Values \v include \skp (unit value, called \c{nothing} in Julia) and \m (generic
function value \footnote{\cbstart Distinguished function values are a simplification begotten by the calculus; Julia allows functions to be called on any receiver, not just special method ones as in \juliette. \cbend}). Primitive operators $\primopd$ represent built-in functions
such as \c{Core.Intrinsics.mul_int}. Type tags \g include \Unit (unit type, called
\c{Nothing} in Julia) and $\mty$ (tag of function value \m). Type annotations \t
include $\top \in \t$ ($\top$ is the top type, called \c{Any} in Julia) and $\g
\subseteq \t$ (all type tags serve as valid type annotations).

\subsection{Semantics}\label{subsec:wa-semantics}

The internal syntax of \juliette is given in the top of \figref{semantics}.  It
includes evaluation result \rslt (either value or error), method table \MT,
and two evaluation contexts, \Xx and \Cx, which are used to define
small-step operational semantics of \juliette.  Evaluation contexts \Xx are
responsible for simple sequencing, such as the order of argument evaluation;
these contexts never contain global/table evaluation expressions \evalg{\cdot}
and \evalt{\MT}{\cdot}.
World evaluation contexts \Cx, on the other hand, capture the full grammar
of expressions.  

Program state is a pair \stwa{\MT}{\plugCx{\Cx}{\e}} of a global method
table \MT and an expression \plugCx{\Cx}{\e}.  

\begin{equation*}
\evalwa{\eqnmarkbox[blue]{MTw1}{\MT}}{\plugCx{\eqnmarkbox[red]{Cx1}{\Cx}}{\e}}{\eqnmarkbox[blue]{MTw2}{\MT'}}{\plugCx{\eqnmarkbox[red]{Cx2}{\Cx}}{\e'}}
\annotatetwo[yshift=1em]{above}{MTw1}{MTw2}{Global method table}
\annotatetwo[yshift=-1em]{below, label below}{Cx1}{Cx2}{World evaluation context}
\end{equation*}
\vspace{1em}

I define the semantics of the calculus using two
judgments: a normal small-step evaluation denoted by
\evalwa{\MT}{\plugCx{\Cx}{\e}}{\MT'}{\plugCx{\Cx}{\e'}}, and a step to an
error \evalerrwa{\MT}{\plugCx{\Cx}{\e}}{\err}.
The $\typeof(\v) \in \g$
operator returns the tag of a value.  I require that $\typeof(\skp) =
\Unit$ and $\typeof(\m) = \mty$.  I write $\typeof(\vs)$ as a shorthand for
$\obar{\typeof(\v)}$.  Function $\Primop(l, \vs) \in \rslt$ computes primop
calls, and function $\PrimopRT(l,\gs) \in \g$ indicates the tag of $l$'s
return value when called with arguments of types \gs.  These functions have
to agree, i.e.  $\forall \vs,\gs. (\typeof(\vs)=\gs \land \Primop(l, \vs) =
\v' \implies \typeof(\v') = \PrimopRT(l,\gs))$.
The subtyping relation \jsub{\t_1}{\t_2} is used for multiple dispatch.
I require that subtyping is transitive
so if $\jsub{\t_1}{\t_2}$ and $\jsub{\t_2}{\t_3}$ then $\jsub{\t_1}{\t_3}$;
transitivity is needed by the type system for subsumption.

\begin{figure}
  \footnotesize
  \[
  \begin{array}{ccl@{\qquad}l}
      \\ \rslt & ::= & & \text{\emph{Result}}
      \\ &\Alt& \v   & \text{value}
      \\ &\Alt& \err & \text{error}
      \\
      \\ \MT & ::= & & \text{\emph{Method table}}
      \\ &\Alt& \varnothing       & \text{empty table}
      \\ &\Alt& \MText{\MT}{\md}  & \text{table extension}
      \\

  \end{array}
  \begin{array}{ccl@{\qquad}l}
      \\ \Xx & ::= & & \text{\emph{Simple evaluation context}}
      \\ &\Alt& \hole & \text{hole}
      \\ &\Alt& \seq{\Xx}{\e} & \text{sequence}
      \\ &\Alt& \primcall{l}{\vs\ \Xx\ \es} & \text{primop call (argument)}
      \\ &\Alt& \mcall{\Xx}{\es} & \text{function call (callee)}
      \\ &\Alt& \mcall{\v}{\vs\ \Xx\ \es} & \text{function call (argument)}
      \\
      \\ \Cx & ::= & & \text{\emph{World evaluation context}}
      \\ &\Alt& \Xx & \text{simple context}
      \\ &\Alt& \plugx{\Xx}{\evalg{\Cx}} & \text{global evaluation}
      \\ &\Alt& \plugx{\Xx}{\evalt{\MT}{\Cx}} & \text{evaluation in a table \MT}
  \end{array}
  \]
~\\
\begin{mathpar}
\inferrule[\WAE{Seq}]
  { }
  { \evalwad{\plugCx{\Cx}{\seq{\v}{\e}}}{\plugCx{\Cx}{\e}} }

\inferrule[\WAE{Primop}]
  { \Primop(l, \vs) = \v' }
  { \evalwad{\plugCx{\Cx}{\primcalld{\vs}}}{\plugCx{\Cx}{\v'}} }

\inferrule[\WAE{MD}]
  { \md \equiv \mdefd }
  { \evalwa
      {\MT}{\plugCx{\Cx}{\md}}
      {\MText{\MT}{\md}}{\plugCx{\Cx}{\m}} }

\inferrule[\WAE{CallGlobal}]
  { }
  { \evalwa{\tikzmarknode{MT1}{\MT}}
      {\plugCx\Cx{\evalg{\plugx\Xx{ \mcall\m\vs}}}}
      {\MT}
      {\plugCx\Cx{\evalg{\plugx\Xx{ \evalt{\tikzmarknode{MT2}{\MT}}{\mcall\m\vs}}}}}}
\annotatetwo[yshift=-1em]{below, label below}{MT1}{MT2}{Copy global table to local}
\vspace{1em}

\inferrule[\WAE{CallLocal}]
  { \typeof(\vs) = \gs \\ \getmd(\tikzmarknode{MTl}{\MT'},\m,\gs)=\mdefd }
  { \evalwa{\tikzmarknode{MTw}{\MT}}
      {\plugCx\Cx{\evalt{\MT'}{\plugx\Xx{\mcall\m\vs}}}}
      \MT{\plugCx\Cx{\evalt{\MT'}{\plugx\Xx{\e\subst\xs\vs}}}} }
\annotate[yshift=1em]{above}{MTl}{Use local method table}
\annotate[yshift=-1em]{below}{MTw}{Ignore global table}
\vspace{2em}

\inferrule[\WAE{ValGlobal}]
  { }
  { \evalwad{\plugCx{\Cx}{\evalg{\v}}}{\plugCx{\Cx}{\v}} }

\inferrule[\WAE{ValLocal}]
  { }
  { \evalwad{\plugCx\Cx{\evalt{\MT'}\v}}{\plugCx\Cx\v} }
\end{mathpar}

~\\

\begin{mathpar}\footnotesize
\inferrule[\WAE{VarErr}]
  { }
  { \evalerrwad{\plugCx{\Cx}{\x}}{\err} }

\inferrule[\WAE{PrimopErr}]
  { \Primop(l, \vs) = \err }
  { \evalerrwad{\plugCx{\Cx}{\primcalld{\vs}}}{\err} }

\inferrule[\WAE{CalleeErr}]
{ \v_c \neq \m }
{ \evalerrwad
    {\plugCx{\Cx}{\mcall{\v_c}{\vs}}}
    {\err} }

\inferrule[\WAE{CallErr}]
{ \typeof(\vs) = \gs \\ \getmd(\MT', \m, \gs) = \err }
{ \evalerrwad
    {\plugCx{\Cx}{\evalt{\MT'}{\plugx{\Xx}{\mcall{\m}{\vs}}}}}
    {\err} }
\end{mathpar}
~\\[3mm]
\[\footnotesize
\begin{array}{rcl}
  \getmd(\MT, \m, \gs) & = &
    \min(\applcbl(\latest(\MT), \m, \gs)) \\
  \\
  \latest(\MT) & = & \latest(\emptyset, \MT) \\
  \latest(mds, \varnothing) & = & mds \\
  \latest(mds, \MText{\MT}{\md}) & = & \latest(mds \cup \md, \MT) 
    \text{ if } \neg \containseq(mds, \md) \\
  \latest(mds, \MText{\MT}{\md}) & = & \latest(mds, \MT)
    \qquad\, \text{ if } \containseq(mds, \md) \\
  \\
  \applcbl(mds, \m, \gs) & = & \{\mdefd \in mds\ |\
    \jsub{\gs}{\ts} \} \\
  \\
  \min(mds) & = & \mdefd \in mds \text{ such that }
    \forall \mdef{\m}{\obar{\jty{\_}{\t'}}}{\_}
    \in mds\ .\ \jsub{\ts}{\ts'}\\
  \min(mds) & = & \err \text{ otherwise}
  \\
  \containseq(mds, \md) & = & \exists\,\md' \in mds \text{ such that } \\
    & & \quad (\md \equiv \mdef{\m}{\obar{\jty{\_}{\t}}}{\_}) \ \land\
        (\md' \equiv \mdef{\m}{\obar{\jty{\_}{\t'}}}{\_}) \
        \land \ \jsub{\ts}{\ts'} \land\ \jsub{\ts'}{\ts}
\end{array}
\]
\caption{Internal Syntax and Semantics}\label{semantics}
\end{figure}

\paragraph{Normal Evaluation}
These rules capture successful program executions.

Rule \WAE{Seq} is completely standard: it throws away the evaluated part of
a sequencing expression. Rules \WAE{ValGlobal} and \WAE{ValLocal} pass value
\v to the outer context. This is similar to Julia where \eval returns the
result of evaluating the argument to its caller.
Rule \WAE{MD} is responsible for updating the global table: a method
definition $\md$ will extend the current global table \MT into
$\MText{\MT}{\md}$, and itself evaluate to $\m$, which is a function value.
Note that \WAE{MD} only extends the method table and leaves existing
definitions in place.  If the table contains multiple definitions of a
method with the same signature, it is then the dispatcher's responsibility
to select the right method; this mechanism is described below in more
detail.

The two call forms $\WAE{CallGlobal}$ and $\WAE{CallLocal}$ form the core of
the calculus. The rule $\WAE{CallGlobal}$ describes the case where a method
is called directly from a global evaluation expression. In Julia, this means
either a top-level call, an \invokelatest call, or a call within \eval such
as \c{eval(:(g(...)))}. The ``direct'' part is encoded with the use of a
simple evaluation context \Xx.  In this global-call case, I need to save
the current method table into the evaluation context for a subsequent use by
$\WAE{CallLocal}$. To do this, I annotate
the call \mcall{\m}{\vs} with a copy of the current global method table
$\MT$, producing \evalt{\MT}{\mcall\m\vs}.

To perform a local call---or, equivalently, a call after the invocation has
been wrapped in an annotation specifying the current global
table---$\WAE{CallLocal}$ is used.  This rule resolves the call according to
the tag-based multiple-dispatch semantics in the ``deepest'' method table
$\MT'$ (the use of \Xx makes sure there are no method tables between $\MT'$
and the call). Once an appropriate method has been found, it proceeds as a
normal invocation rule would, replacing the method invocation with the
substituted-for-arguments method body. Note that the body of the method is
still wrapped in the \evalt{\MT'}{} context. This ensures that nested calls
will be resolved in the same table (unless they are more deeply wrapped
in a global evaluation \evalg{}).

An auxiliary meta-function $\getmd(\MT, \m, \gs)$, which is used to resolve
multiple dispatch, is defined in the bottom of \figref{semantics}. This function
returns the most specific method applicable to arguments with type tags \gs, or
errs if such a method does not exist.  If the method table contains
multiple equivalent methods, older ones are ignored. For example, for the
program
\[
  \evalg{\seq{\mdef{\mval{g}}{}{2}}
  {\seq{\mdef{\mval{g}}{}{42}}{\mcall{\mval{g}}{}}}},
\]
function call \mcall{\mval{g}}{} is going to be resolved in the table
$\MText{(\MText{\varnothing}{\mdef{\mval{g}}{}{2}})}{\mdef{\mval{g}}{}{42}}$,
which contains two equivalent methods (I call methods equivalent if they
have the same name and their argument type annotations are equivalent with
respect to subtyping). In this case, the function \getmd will
return method \mdef{\mval{g}}{}{42} because it is the newest method out of
the two.

Note that functions can be mutually recursive because of the dynamic nature
of function call resolution.

\paragraph{Error Evaluation}

These rules capture all possible error states of \juliette.
Rule \WAE{VarErr} covers the case of a free variable, an \c{UndefVarError} in
Julia.  \WAE{PrimopErr} accounts for errors in primitive operations such as
\c{DivideError}.  \WAE{CalleeErr} fires when a non-function value is called.
Finally, \WAE{CallErr} accounts for multiple-dispatch resolution errors,
e.g. when the set of applicable methods is empty (no method found), and when
there is no best method (ambiguous method).

\section{Static Type System}
\label{sec:typesystem}

The type system for \juliette augments the method definition form with a
return type $\mdefty{\m}{\obar{\jty{\x}{\t}}}{\mu}{\e}$; these types are
either expressed explicitly (in typed code) or are inferred statically (in
untyped code, based on the declared argument types). I do not modify the
dynamic semantics of \juliette for this system. Untyped methods have a return
type $\mu = \star$, while typed methods have a return type $\mu = \t$.

\newcommand{\checkproto}{\ensuremath{\mathop{\mathbf{checkproto}}}}
\newcommand{\pdef}[3]{\ensuremath{\triangleleft\, \mathop{#1}(#2)::#3 \,\triangleright}}

\newcommand{\pd}{\ensuremath{\texttt{pd}}\xspace}
\newcommand{\PT}{\ensuremath{\mathrm{P}}\xspace}
\begin{figure}

  \footnotesize
  \[
  \begin{array}{ccl@{\qquad}l}
      \\ \MT & ::= & & \text{\emph{Method table}}
      \\ &\Alt& \varnothing       & \text{empty table}
      \\ &\Alt& \MText{\MT}{\mdefty{\m}{\obar{\jty{\x}{\t}}}{\mu}{\e}}  & \text{extension}
      \\
      \\ \PT & ::= & & \text{\emph{Protocol table}}
      \\ &\Alt& \varnothing       & \text{empty table}
      \\ &\Alt& \MText{\PT}{\pdef\m{\obar{\jty{}{\t}}}{\t_r}}  & \text{extension}
      \\
      \\ \Gamma & ::= & & \text{\emph{Variable typing}}
      \\ &\Alt& \varnothing       & \text{empty table}
      \\ &\Alt& \jty{\x}{\t} \;\bullet\; \Gamma & \text{variable typing}
  \end{array}
  \begin{array}{ccl@{\qquad}l}
    \\ \e & ::= & & \ldots
    \\ &\Alt& \mcallt{\e}{\es}\MT{\t} & \text{checked function call}
    \\
    \\ \Xx & ::= & & \ldots
    \\ &\Alt& \mcallt{\Xx}{\es}\MT\t & \text{checked function call (callee)}
    \\ &\Alt& \mcallt{\v}{\vs\ \Xx\ \es}\MT\t & \text{checked function call (argument)} 
    \\
    \\ \md & ::= & \mdeftyd & \text{\emph{Method definition}}
  \end{array}
  \]
\begin{mathpar}
\inferrule[\TRR{Value}]
  {  \typeof(\v) = \g }
  { \Gamma \vdash_{\MT ; \PT} \v \Rightarrow \v : \g }

\inferrule[\TRR{Var}]
  { \x : \t \in \Gamma }
  { \Gamma \vdash_{\MT ; \PT} \x \Rightarrow \x : \t }

\inferrule[\TRR{Seq}]
  { \Gamma \vdash_{\MT ; \PT} \e_1 \Rightarrow \e'_1: \t_1 \\ \Gamma \vdash_{\MT ; \PT} \e_2 \Rightarrow \e'_2 : \t_2  }
  { \Gamma \vdash_{\MT ; \PT} \seq{\e_1}{\e_2} \Rightarrow \seq{\e'_1}{\e'_2} : \t_2 }

\inferrule[\TRR{Sub}]
		{\Gamma \vdash_{\MT ; \PT} \e \Rightarrow \e' : \t \\ \jsub{\t}{\t'}}
		{\Gamma \vdash_{\MT ; \PT} \e \Rightarrow \e' : \t'}

\inferrule[\TRR{Primop}]
  { \obar{\Gamma \vdash_{\MT ; \PT} \e \Rightarrow \e' : \t} \\ \Psi(l, \ts) = \t' }
  { \Gamma \vdash_{\MT ; \PT} \primcalld{\es} \Rightarrow \primcalld{\es'} : \t' }

\inferrule[\TRR{Call}]
  { \Gamma \vdash_{\MT ; \PT} \e_r \Rightarrow \e'_r : \typeof(\m) \\ 
  \obar{\Gamma \vdash_{\MT ; \PT} \e_a \Rightarrow \e'_a : \t_a} \\ 
  \dispatch(\MT, \PT, \m, \obar{\t_a}) = \t_r }
  { \Gamma \vdash_{\MT ; \PT} \mcall{\e_r}{\obar{\e_a}} \Rightarrow \mcallt{\e'_r}{\obar{\e_a'}}\MT{\tikzmarknode{tres}{\t_r}} : \t_r }
\annotate[yshift=-0.5em]{below}{tres}{Embed return type}
\vspace{0.5em}

\end{mathpar}

\begin{mathpar}
\inferrule[\TRR{Prog}]
  {  
  	\obar\md \vdash \pd_i \; \forall \pd_i \in \obar\pd \\
  	\obar\md;\ \obar\pd \vdash \md_i \Rightarrow \md_i \; \forall \md_i \in \obar\md }
  { \vdash \obar{\md}; \obar{\pd} \; \e \Rightarrow \obar{\md'} \; \e }

\inferrule[\TRR{MD-Typed}]
  { \obar{\jty\x{\t_a}} \vdash_{\MT ; \PT} \e \Rightarrow \e' : \t_r}
  { \MT; \PT \vdash \mdefty{\m}{\obar{\jty{\x}{\t_a}}}{\t_r}{\e} \Rightarrow \mdefty{\m}{\obar{\jty{\x}{\t_a}}}{\t_r}{\e'} }

\inferrule[\TRR{MD-Untyped}]
  { }
  { \MT; \PT \vdash \mdefty{\m}{\obar{\jty{\x}{\t}}}{\star}{\e} \Rightarrow \mdefty{\m}{\obar{\jty{\x}{\t}}}{\star}{\e} }

\inferrule[\TRR{PD}]
 { \checkproto(\m, \obar{\t}, \t_r, \MT) }
 { \MT \vdash \pdef\m{\obar{\jty{}{\t}}}{\t_r} }
\end{mathpar}

\caption{Typed translation for \juliette}\label{typesystem}
\end{figure}

Typing for \juliette works over partially-typed programs of the form $\MT; \PT
\; \e$ consisting of a method table $\MT$, a protocol table $\PT$, and an
executing (untyped) expression $\e$. \TRR{Prog} defines the program
well-formedness relation, ensuring that all protocol definitions $\obar\pd$ in
the program are satisfied followed by checking that all method definitions are
well-formed against the other definitions and the protocol table. Typed
methods are well-formed if their bodies typecheck; untyped methods are always
well-formed. Protocol well-formedness is delegated to the $\checkproto$
metafunction whose definition I will explore later.

The basic equation typing relation types and translates an expression with the relation

\vspace{1em}
\begin{ceqn}
\begin{equation*}
\tikzmarknode{Gam}{\Gamma} \vdash_{\MT ; \PT} \tikzmarknode{ei}{\e} \Rightarrow \tikzmarknode{er}{\e'} : \tikzmarknode{tres}{\t}
\annotate[yshift=1em]{above, left}{Gam}{Variable typing context}
\annotate[yshift=-1em]{below, left}{ei}{Expression to check}
\annotate[yshift=-1em]{below}{er}{Resulting expression}
\annotate[yshift=1em]{above}{tres}{Resulting type}
\end{equation*}
\end{ceqn}
\vspace{1em}

\noindent Expressions are type checked against a typing context $\Gamma$ that contains
the variables in scope (introduced in \juliette as arguments to the function
containing the current expression). The expression being checked $\e$ is then
translated into a resulting expression $\e'$ while producing a result type
$\t$.

The type system addresses each of the five errors as follows:
\begin{itemize}
	\item \WAE{VarError} by ensuring that all variables lexically exist.
	\item \WAE{PrimopError} by checking that the primitive operation is defined for the realized argument types.
	\item \WAE{CalleeErr} by ensuring that only method-typed variables are valid in invocee position.
	\item \WAE{CallErr}, with the exception of ambiguities, by ensuring that a suitable implementation always exists be it by there being an implementation for a supertype or there being a suitable protocol.
	\item \WAE{TypeErr} by ensuring that, so long as the dynamically used method table is the same as the statically checked one, the values returned from methods always are of the statically-known type.
\end{itemize}

\WAE{TypeErr} deserves a deeper examination. Casts do not appear in normal
\juliette code; they are only inserted by the translation to enforce
statically-determined return types. In the prior example of \jlinl{f} and
\jlinl{g}, \jlinl{g} expected \jlinl{f} to always return a \jlinl{Int} but the
method table was extended at runtime with a \jlinl{f} that returned a
\jlinl{String}. Therefore, I \emph{can} suffer cast failures when a new
method is added after the method table the program was typechecked against. 

The need to check the real returned value from a method determines the
structure of \TRR{Call}. \TRR{Call} typechecks a method invocation by ensuring
that it is statically resolvable using the \dispatch metafunction. The rule
then translates checked calls to include both the statically-determined return
type and the method table against which that return type was generated. This
method table and return type will be used later to ensure type safety of returned
values.

\subsection{Static Dispatch Resolution}

The key operator used by the static type system is \dispatch. Method calls are
the key component of Julia's semantics and the \dispatch metafunction is used
in \TRR{Call} to determine the two key properties for a function call:
\begin{itemize}
	\item Will there be a method to invoke?
	\item What will the return type be?
\end{itemize}

To motivate our treatment of dispatch, suppose that I have evaluated a method
call down to a bare call $\m(\obar\v)$. At this point in evaluation the called
method is known and all arguments are now values; I must now figure out what
method could actually be called here. The \juliette calculus handles this
using rule \WAE{CallLocal}, which works primarily through the  $\getmd(\MT,
\m, \obar{\s}) = \mdefty\m{\obar{\jty\x{\t_a}}}{\_}\e$ function where $\MT$ is
our current method table, $\obar\s = \obar{\typeof(\v)}$ is the vector of
value types, and the method definition $\mdefty\m{\obar{\jty\x{\t_a}}}{\_}\e$
is the resulting implementation found for this value vector. If it cannot find a 
singular most specific implementation it instead produces $\err$.

Now, suppose that instead I have an unevaluated call $\e(\obar{\e'})$. If we
suppose that its arguments are typed $\Gamma \vdash_\MT \e' : \t'$ then by
soundness all ensuing argument values $\obar\v$ are instances of $\obar{\t'}$.
Therefore, to prove that \WAE{CallLocal} can eventually apply it suffices to
show that for any value vector $\obar\v$ with tags $\obar\sigma =
\obar{\typeof(\v)}$ such that $\jsub{\obar\sigma}{\obar{\t'}}$ that I can find
some unique most specific method definition in $\MT$. The statically-inferred
return types are then simply the meet of the return types of the identified
methods.

I use the abstract metafunction $\dispatch$ to describe this operation. I say that
$\dispatch(\MT,\PT,\m,\obar\t)=\t_r$ holds if it can guarantee that calling $\m$
in method table $\MT$ and with protocol table $\PT$ with arguments $\obar\t$ will produce a return type that
is a subtype of $\t_r$. Alternatively, $\dispatch(\MT, \PT, \m, \obar\t) = \err$ if
it cannot provide this guarantee. 

\paragraph{Success} The $\dispatch$ metafunction can succeed in two cases:
\begin{itemize}
	\item Satisfying method: there is some method whose arguments $\obar{\t_a}$ satisfy $\jsub{\obar\t}{\obar{\t_a}}$. In this case resolving $\dispatch$ is trivial: that implementation ensures that the call will go through to something, even if there are more specific implementations that might also be called.
	\item Protocol: there is no single method whose arguments are a supertype of the given arguments, but for every concrete type vector $\jsub{\obar\sigma}{\obar{\t'}}$ there is a method with type $\obar{\t_a}$ such that $\jsub{\obar\sigma}{\obar{\t_a}}$.
\end{itemize}
The first case is trivial: I know a suitable method exists, so I can simply
say ``at least that one will be invoked.'' The second case is addressed by
the protocol mechanism and is deferred by \dispatch to those static declarations.
As mentioned previously, I could check for protocol safety as part of \dispatch
but determined that the design consequences were undesirable.

Thus, I propose a solution to $\dispatch$ success in two parts: the
$\dispatch$ metafunction \emph{itself} only checks for satisfying methods,
methods whose declared types are supertypes of the given argument typing. We
then pair this with an \emph{protocol checking} mechanism, wherein
programmers can define protocols that act as an independent source of truth
for both protocol implementations (e.g. every \jlinl{AbstractArray} must have
a \jlinl{size}) as well as for use sites.

\paragraph{Failure} A \juliette (and Julia) function call can fail in two ways:
\begin{itemize}
	\item No implementations exist: There are no methods whose signatures are a supertype of the runtime argument type
	vector. Equivalently, $\exists \obar{\v}$ such that $\obar{\jsub{\typeof(\v)}{\t}}$ but where for any method $\mdefty{\m}{\obar{\jty{\x}{\t_a}}}{\t_r}{\e}$
	there is some $i$ such that $\notjsub{\t_i}{\t_{a, i}}$.
	\item There are multiple \emph{most specific} implementations of $\m$ for this case. 
\end{itemize}

The first kind of failure, no implementation exists, is straightforward and is
the multidispatch equivalent of  ``message not understood;'' effectively, it
is the inverse of the ``method found'' success case. The second kind is more
challenging. Julia dispatches method calls to the \emph{most specific},
\emph{satisfying} implementation, which is represented in \juliette with the
$\min$ and $\applcbl$ metafunctions. As mentioned earlier, I do not statically
protect against ambiguities; method ambiguous errors may occur at any invocation
at runtime.

\paragraph{Definition}
This then brings us to the definition of $\dispatch(\MT, \PT, \m, \obar\t)$:

\begin{mathpar}
\inferrule
	{\pdef\m{\obar{\jty{}{\t_a}}}{\t_r} \in \PT \\ \jsub{\obar\t}{\obar{\t_a}}}
	{\dispatch(\MT, \PT, \m, \obar\t) = \t_r}

\inferrule 
 {\exists{\mdefty{\m}{\obar{\jty{\x}{\t_a}}}{\mu}{\e}} \in \MT \wedge {\jsub{\obar\t}{\obar{\t_a}}} \\
  \eqnmark{sub}{\forall \obar\s : \jsub{\obar\s}{\obar \t},\; \getmd(\MT, \m, \obar{\s}) = \mdefty{\m}{\obar{\jty{\x}{\t'_a}}}{\mu}{\e} \implies \infers{\obar{\jty{\x}{\t'_a}}}{\MT}{\e}{\t'_r}} \wedge \t'_r <: \t_r}
 {\dispatch(\MT, \PT, \m, \obar\t) = \t_r}
\end{mathpar}

Trivially, if the method is one that I statically checked using a protocol
then the dispatch will defer to that protocol. Otherwise, as is handled by the
second rule, I need to check if this specific invocation is safe.

I break up non-protocol dispatch checking into two clauses. The first clause
ensures that a suitable method always exists, while the second ensures that
the return type is fully general for all possible implementations. The
$\dispatch$ function may return either a typed or typed method as most
specific; the return type need only be valid for the final result and does not
have to reflect a static typing of the method body.

The second clause covers the case when I dispatch to a method more specific
than this ``sufficiently general'' one; it states that any method that can be
invoked from the current call site must be \emph{inferrable} to have the
correct return type.

Inference in my system is treated as a black box. The judgment takes the form
\vspace{2em}
\begin{ceqn}
\begin{equation*}
\tag{Type inference}
\infers{\tikzmarknode{env}{\Gamma}}{\tikzmarknode{mt}{\MT}}{\tikzmarknode{e}{\e}}{\tikzmarknode{t}{\t}}
\annotate[yshift=1em]{above, left}{env}{Inference environment}
\annotate[yshift=-1em]{below, left}{mt}{Inference method table}
\annotate[yshift=1em]{above, right}{e}{Expression to infer}
\annotate[yshift=-1em]{below, right}{t}{Inferred type}
\label{eqn:inference}
\end{equation*}
\end{ceqn}
\vspace{2em}

For convenience I define the notation $\evalwam{\MT'}{\e}{\MT}{\MT''}{\e'}$ as
$\forall \Cx\, \evalwa{\MT'}{\plugCx\Cx{\evalt\MT\e}}{\MT''}{\plugCx\Cx{\evalt\MT{\e'}}}$;
interpret this as ``the expression $\e$ evaluates to $\e'$ with local method table $\MT$.''

The primary required property of inference is weaker than that for a
traditional static type system. Instead of a strong soundness guarantee that
rules out errors entirely inference merely states that \emph{if} an expression
steps then it will be well-typed. Thus, soundness of inference is stated as
follows:

\begin{definition}[Soundness of inference]
\label{def:inferencesound}
If $\infers{}{\MT}{\e}{\t}$ then either
\begin{itemize}
	\item $\exists \v:\;\e = \v$ and $\jsub{\typeof(\v)}{\t}$
	\item $\forall \MT'\, \exists \MT'', \evalwam{\MT'}{\e}\MT{\MT''}{\e'}$ and $\infers{}{\MT}{\e'}{\t}$
	\item $\forall \MT',\, \evalerrwa{\MT'}{\evalt\MT\e}\err$
\end{itemize}
\end{definition}

For the purposes of typing I additionally require that inference is
consistent under substitution:
\begin{definition}[Substitution for inference]
If $\infers{\obar{\jty\x{\t'}}}{\MT}{\e}{\t}$ and $\jsub{\typeof(\obar\v)}{\obar{\t'}}$ then $\infers{}{\MT}{\e\subst\vs\xs}{\t}$
\end{definition}

The good news, then, is that $\dispatch$ itself is easy to implement. I just
need to look for a method whose arguments are supertypes of the known
arguments. If such a method exists, I meet its (potentially-inferred) return
type with the return type of all other possible implementations and I are
done. Julia already implements a type inference system with which I can infer
return types for untyped methods making this task straightforward.

The remaining problem is how does the protocol checking metafunction
$\checkproto$ work?

\subsection{Protocols}

Consider the earlier example of \jlinl{Range} and \jlinl{List}, both of  which
were required to have an implementation of \jlinl{size}. Implementing the
check for \jlinl{size} is easy enough - simply explore all possible
instantiation of \jlinl{AbstractArray} - but it is  not so easy in the general case.
For example, I could have a version of \jlinl{size} that takes multiple
arguments like \jlinl{size(::AbstractArray, ::AbstractArray)} (which requires joint exhaustion
of all arguments) or could start using Julia's type language for richer
properties such as \jlinl{f(::T, ::T) where T<:AbstractArray}.

The protocol checker need not be complete; it must only be sound. My formalism
relies on an abstract protocol checker $\checkproto(\m, \obar{\t}, \MT)$ that
ensures that protocol $\m$ exists in the method table $\MT$ with type
arguments $\t$. The implementation of \checkproto must adhere to the following
specification:
\begin{mathpar}
\inferrule
	{ \forall \vs: (\typeof(\vs) = \gs \wedge \gs <: \ts), \\ 
		\getmd(\MT, \m, \gs) = \mdefty{\m}{\obar{\jty{\x}{\t}}}{\mu}{\e} \\ \wedge \\ \infers{\obar{\jty{\x}{\t}}}{\MT}{\e}{\t'_r} \\ \wedge \\ \jsub{\t'_r}{\t_r} }
	{\checkproto(\m, \obar{\t}, \t_r, \MT)}
\end{mathpar}

Within \juliette I cannot add new types, only new methods. As a result,
\checkproto can guarantee existence of suitable methods as the method table
continues to evolve; like \dispatch, it can not guarantee that the return type
is still correct, however, so dynamic checks will still be needed if the
method table is modified.

Ensuring that some implementation exists for any instantiation of a given type
is analogous to completeness checking in pattern matching. For the \juliette
calculus I present a simple algorithm based on
Maranget~\cite{maranget2007warnings} that is able to provide sound and
complete checking for the \juliette calculus, then discuss its key limitations
when applied to generalized Julia.

\paragraph{Maranget-style checking} Maranget's algorithm is described in terms of
patterns; I will first describe the pattern language and the briefly go over
the function of the algorithm. I will then describe the reduction from
\juliette protocol checking to these patterns.

\newcommand{\wildpat}{\texttt{\_}} In Maranget's system values solely consist
of constructors $c(v_1, \ldots, v_n)$; base values are constructors with no
arguments (for example the nil constructor $nil()$). Patterns are then either
wildcards $\wildpat$, constructor applications $c(v_1, \ldots, v_n)$, or
disjunctions $p_1 \; | \; p_2$. 

\newcommand{\inex}[1]{\ensuremath{\mathcal{I}(#1)}}
I use Maranget's inexhaustiveness algorithm \inex{P, n} to implement completeness
checking. In Maranget's setting $P$ is a ``matrix'' of patterns (that is, each pattern
is a row) and $n$ is the number of arguments being matched. The algorithm \inex{P, n} returns
a pattern vector $p$ of size $n$ that is \emph{not} matched by $P$; if no such vector
exists, it returns $\bot$. The implementation of \inex{P, n} used for Julia is
identical to that of Maranget; I refer the reader to that treatment for details.

I then reduce completeness checking in \juliette to an instance of pattern
matching in Maranget's language. I consider each leaf type $\sigma$ to be a
constructor with zero arguments and to inhabit the same type membership
heirarchy as exists in \juliette. For example, \jlinl{Int()} is a constructor
for the type \jlinl{Number}. In this manner, I can trivially define a mapping
$\text{pat}(\sigma) = \sigma()$ from \juliette base types to patterns.
Abstract types are handled by conversion into a disjunction of base types.

Similarly, I can take a set of implementations and abstract an equivalent
pattern matrix. Suppose that I have method implementations $\m(\obar{\t^1}),
\ldots, \m(\obar{\t^n})$. I can
construct a pattern matrix with $P(\m(\obar{\t^1}),\ldots, \m(\obar{\t^n})) =
\left[\begin{array}{c}\obar{\text{pat}(\t_1)} \\ \ldots \\
\obar{\text{pat}(\t_n)}\end{array}\right]$ that matches the converted argument 
if and only if a suitable method exists in the original \juliette method table.

This system is trivially incomplete. In particular, while it can adequately 
handle tag types, tuples, unions, and \emph{simple} (non-bounded, non-diagonal)
parametric types, its ability to check signatures that exploit Julia's bounded
polymorphism is very limited.

As a practical example of where this may arise consider the earlier example
of \jlinl{+} of two \jlinl{Numbers}. I might write
\begin{lstlisting}
@protocol +(::Number, ::Number)::Number
\end{lstlisting}
to specify that there must exist an implementation of \jlinl{+} for 
\emph{any two} numbers. That is, if I have both \jlinl{Int} and \jlinl{Float64}
as subtypes of \jlinl{Number} I must implement
\begin{lstlisting}
+(::Int, ::Int)::Number
+(::Float64, ::Float64)::Number
+(::Float64, ::Int)::Number
+(::Int, ::Float64)::Number
\end{lstlisting}
to satisfy this protocol. Obviously, writing exponentially many definitions
gets wearisome after only a few types and Julia does not require that \jlinl{+}
implementations do so. Instead, Julia uses a mechanism called promotion (which I
will not describe in detail here) to make the two types equal, therein requiring
only two implementations:
\begin{lstlisting}
+(::Int, ::Int)::Int
+(::Float64, ::Float64)::Float64
\end{lstlisting}

I could write this requirement as a protocol in Julia's full type language as
\begin{lstlisting}
@protocol +(::T, ::T)::T where T<:Number
\end{lstlisting}
but then this wanders enthusiastically outside of the type language supported
by the \emph{warnings}-derived algorithm.

The Maranget-derived checker is sufficient to check protocols in \juliette
programs because \juliette's type system does not have many of the complex
features that real Julia has. While the concept of protocols is quite general,
this specific implementation is not; a truly generic type system for
generalized Julia would need a much more sophisticated protocol checking
mechanism than I describe here.

\section{Typed Dynamic Semantics}

Next, I describe the dynamic semantics for typed \juliette. 

\begin{figure}
\begin{mathpar}
\inferrule[\WAE{TypedCallGlobal}]
  {  }
  { \evalwad
      {\plugCx\Cx{\evalg{{\plugx\Xx{\mcallt\m\vs{\MT'}{\t_r}}}}}}
      {\plugCx\Cx{\evalt{\MT}{{\plugx\Xx{\mcallt\m\vs{\MT'}{\t_r}}}}}}}

\inferrule[\WAE{CallInferredLocal}]
  { \typeof(\vs) = \gs \\ \getmd(\MT',\m,\gs)=\mdeftyd }
  { \evalwad
      {\plugCx\Cx{\evalt{\MT'}{{\plugx\Xx{\mcallt\m\vs{\MT'}{\t_r}}}}}}
      {\plugCx\Cx{\evalt{\MT'}{{\plugx\Xx{\evalinf{\e\subst\xs\vs}}}}} }}

\inferrule[\WAE{CallCheckedLocal}]
  { \typeof(\vs) = \gs \\ \getmd(\MT',\m,\gs)=\mdeftyd }
  { \evalwad
      {\plugCx\Cx{\evalt{\MT'}{{\plugx\Xx{\mcallt\m\vs{\MT' \bullet \MT''}{\t_r}}}}}}
      {\plugCx\Cx{\evalt{\MT'}{{\plugx\Xx{\evalchk{\e\subst\xs\vs}{\t_r}}}}} }}

\inferrule[\WAE{ValChecked}]
  { \jsub{\typeof(\v)}{\t} }
  { \evalwad{\plugCx{\Cx}{{\evalchk{\v}\t}}}{\plugCx{\Cx}{\v}} }

\inferrule[\WAE{ValInferred}]
  {  }
  { \evalwad{\plugCx{\Cx}{{\evalinf{\v}}}}{\plugCx{\Cx}{\v}} }
\end{mathpar}

~\\

\begin{mathpar}\footnotesize
\inferrule[\WAE{TypeErr}]
{ \notjsub{\typeof(\v)}{\t} }
{ \evalerrwad
    {\plugCx{\Cx}{{\evalchk{\v}{\t}}}}
    {\err} }
\end{mathpar}
\caption{Dynamic semantics for typed \juliette.}
\end{figure}

The main differentiation between the typed and untyped dynamic semantics for
\juliette arises from the handling of method invocation. Typed \juliette has
three rules for method calls:
\begin{itemize}
	\item \WAE{Call} used for untyped invocation, and is the same as base \juliette.
	\item \WAE{TypedCallGlobal} used to capture the current global evaluation context
	into a new local evaluation context.
	\item \WAE{CallInferredLocal} used when the current method table \MT is the same as
	the one used to statically check correctness and thus for which static inference can be relied upon.
	\item \WAE{CallCheckedLocal} used when the method table \MT has been extended with some
	$\MT_e$ and thus the returned value needs to be dynamically checked.
\end{itemize}
Typed invocations are solely the domain of the type system; the programmer cannot
manually write typed calls.

The two local typed invocations each rely on their own cast form: \evalchk\e\t which
performs a dynamic check that \e has type \t and \evalinf\e that indicates that
\e has a reliable inferred type. \evalchk\e\t is used when the return type of \e
cannot be statically guaranteed and needs to be dynamically checked. In \evalchk\e\t
\e might get stuck at any point or it might be an ill-typed value, in which case the
program steps to an error. \evalinf\e, on the other hand, guarantees that \e inferred
to \t under the relevant method table. As a result, \e might get stuck internally but
if \e is a value then it will always be well-typed.

I use a standard progress and preservation proof methodology to show soundness
in \juliette. Towards this end, I define a typing relation for the target
language, shown in figure~\ref{fig:targettypes}. The system follows the typed
translation rules closely, with the primary exception being the inclusion of
\TER{Cast} that simply asserts that the result of evaluating some expression
$\e$ is type $\t$, therein providing a interpretation of the cast contexts
$\evalchk{\e}\t$.

This target type system is similar to the one in our prior work in Belyakova
et al~\cite{belyakova2020world} but diverges by providing a typing judgment
for function invocations \TER{Call}, for inferred expressions \TER{Inferred}, and for
checked expressions \TER{Checked}. 
Additionally, it prohibits nested local and global evaluation contexts from
existing within typed expressions.

\begin{figure}
\begin{mathpar}
\inferrule[\TER{Value}]
  {  \typeof(\v) = \g }
  { \Gamma \vdash_{\MT; \PT} \v : \g }

\inferrule[\TER{Var}]
  { \x : \t \in \Gamma }
  { \Gamma \vdash_{\MT; \PT} \x : \t }

\inferrule[\TER{Seq}]
  { \Gamma \vdash_{\MT; \PT} \e_1 : \t_1 \\ \Gamma \vdash_{\MT; \PT} \e_2 : \t_2  }
  { \Gamma \vdash_{\MT; \PT} \seq{\e_1}{\e_2} : \t_2 }

\inferrule[\TER{Sub}]
		{\Gamma \vdash_{\MT; \PT} \e : \t \\ \jsub{\t}{\t'}}
		{\Gamma \vdash_{\MT; \PT} \e : \t'}

\inferrule[\TER{Primop}]
  { \obar{\Gamma \vdash_{\MT; \PT} \e : \t} \\ \Psi(l, \ts) = \t' }
  { \Gamma \vdash_{\MT; \PT} \primcalld{\es} : \t' }

\inferrule[\TER{Call}]
  { \Gamma \vdash_{\MT; \PT} \e : \typeof(\m) \\ 
  \obar{\Gamma \vdash_{\MT; \PT} \e' : \t_a} \\ 
  \dispatch(\MT, \PT, \m, \obar{\t_a}) = \t_r }
  { \Gamma \vdash_{\MT; \PT} \mcallt{\e}{\obar{\e'}}{\MT'}{\t_r} : \t_r }

\inferrule[\TER{Inferred}]
  { \infers{}\MT\e\t }
  { \Gamma \vdash_{\MT; \PT} \evalinf\e : \t }

\inferrule[\TER{Checked}]
  {  }
  { \Gamma \vdash_{\MT; \PT} \evalchk\e\t : \t }
\end{mathpar}
\caption{Typed \juliette target language typing}
\label{fig:targettypes}
\end{figure}

From here, I can state soundness in typed \juliette as follows.

\begin{theorem}[Soundness of typed \juliette]
If $\vdash_{\MT; \PT} \e : \t$ one of three cases holds:
\begin{enumerate}
	\item $\e$ is a value $\v$ where $\jsub{\typeof(\v)}{\t}$
	\item For all $\MT', \MT_e$ there is some $\MT''$ and $\e'$ such that $\evalwam{\MT'}{\e}{\MT \bullet \MT_e}{\MT''}{\e'}$ and $\vdash_{\MT; \PT} \e' : \t$
	\item $\e$ is of the form $\plugx\Xx{\mcall{\m}{\obar\v}}$ and there exists an equivalence class $\mathrm{T}$ in $\MT$ of at least
	two methods with the form $\obar{\mdefty{\m}{\obar{\jty{\_}{\t_a}}}{\_}{\_}} \subseteq \applcbl(\latest(\MT), \m, \obar\sigma)$ that are not pairwise related by subtyping
	(for any two argument types $\obar{\t_a}$ and $\obar{\t'_a}$ in $\mathrm{T}$ then $\notjsub{\obar{\t_a}}{\obar{\t'_a}}$ and $\notjsub{\obar{\t'_a}}{\obar{\t_a}}$) and that are
	are more precise than all other implementations (e.g. for any argument type $\obar{\t_a}$ in $\mathrm{T}$ then $\forall \mdefty{\m}{\obar{\jty{\_}{\t}}}{\_}{\_} \in \applcbl(\latest(\MT), \m, \obar\sigma)$ if $\obar\t \not\in \mathrm{T}$ then $\jsub{\obar{\t_a}}{\obar\t}$).
	\item There exists some $\Xx, \e', \t'$ such that $\e = \plugx\Xx{\evalchk{\e'}{\t'}}$.
	\item There exists some $\Xx, \e'$ such that $\e = \plugx\Xx{\evalinf{\e'}}$ where $\forall \v, \e' \neq \v$.
\end{enumerate}
\end{theorem}
The first two cases are obvious: if I have executed the program to a
value then it should be correctly-typed. Similarly, if I have a typed
context within our execution, it should step to a similarly well-typed
form.

The third case of soundness is the carve-out for ambiguous method
invocations. An ambiguous method invocation occurs when there is a congruence
class within the set of applicable methods (e.g. those that could handle the
given arguments) that are not comparable by subtyping (that are equally
precise as one another) but that are all more precise than any other
implementation. This is where the carve-out of the definition of $\dispatch$
shows up: it merely guarantees that a method \emph{exists}, not that there
will always be a unique most specific one. Thus, I cannot guarantee the
absence of such an equivalence class based on a successful $\dispatch$ check.

The fourth case of soundness applies to dynamic typechecks applied to typed
function invocations when new methods $\MT_e$ have been added to the original
method table $\MT$. In this case I cannot guarantee that all callable methods 
have a return type that infers to a subtype of the expected type \t and thus
must dynamically check the return type. Moreover, the body of the unchecked and
uninferred method might go wrong at any time.

Finally, the fifth case of soundness applies to the expressions resulting from
function invocations that occurred against the statically relevant method
table. In this case \dispatch guarantees that whatever method is called infers
to return a subtype of the expected type \t. Thus, the embedded expression \e
might go wrong but if it is a value \v then it will be appropriately typed.

Note, first, that soundness generalizes on any global method table 
$\MT'$ and extended local method table $\MT_e$. Typed methods themselves
will thus not go wrong; ill-typed values returned from newly-{\eval}ed 
methods might cause a cast failure but will not break typed code.

Next, observe that \evalinf\e requires no dynamic checks. Dynamic checks
\evalchk\e\t are only inserted when the local method table has been extended
with some $\MT_e$. As a result, if no new methods are visible in the local method
table the type system applies no additional dynamic overhead. Only once methods
that were not known to the static checker have been added does the type system begin
adding overhead.

\paragraph{Proof}

\begin{figure}
  \[
  \begin{array}{ccl@{\qquad}ll}
      \\ \rdx & ::= & & \text{\emph{Redex Base}} &
      \\ &\Alt& \x & \text{variable} & \text{(error)}
      \\ &\Alt& \seq{\v}{\e} & \text{sequencing} & \text{(normal)}
      \\ &\Alt& \evalchk{\v}{\t} & \text{cast} & \text{(normal/error)}
      \\ &\Alt& \evalinf{\v} & \text{cast} & \text{(normal)}
      \\ &\Alt& \mcall{\vnm}{\vs} & \text{non-function call} & \text{(error)}
      \\ &\Alt& \md & \text{method definition} & \text{(normal)}
      \\ &\Alt& \primcalld{\vs} & \text{primop call} & \text{(normal/error)}
      \\ &\Alt& \evalg{\v} & \text{value in global context} & \text{(normal)}
      \\ &\Alt& \evalt{\MT}{\v} & \text{value in table context} & \text{(normal)}
      \\ &\Alt& \evalg{\plugx{\Xx}{\mcall{\m}{\vs}}} 
            & \text{function call in global context} & \text{(normal)}
      \\ &\Alt& \evalt{\MT}{\plugx{\Xx}{\mcall{\m}{\vs}}} 
            & \text{function call in table context} & \text{(normal/error)}
      \\ &\Alt& \evalg{\plugx{\Xx}{\mcallt{\m}{\vs}\MT\t}} 
            & \text{typed function call in global context} & \text{(normal)}
      \\ &\Alt& \evalt{\MT}{\plugx{\Xx}{\mcallt{\m}{\vs}\MT\t}} 
            & \text{typed function call in table context} & \text{(normal/error)}
  \end{array}
  \]
  \caption{Redex Bases}\label{fig:wa-calculus-redex}
\end{figure}

For the use of the proof I define \emph{redexes} for the calculus, derived from the original
concept that \juliette used. A redex is an expression that is immediately reducible and contains
no reducible subexpressions. The redexes in typed \juliette are shown in \figref{fig:wa-calculus-redex};
compared to the original \juliette calculus typed \juliette adds the cast redex as well as redexes for
statically-typed function calls in global and local method tables.

I use three lemmas for the proof of soundness.

\begin{lemma}{Unique Form of Expressions}
Any expression \e can be uniquely represented in one of the following ways:
\begin{itemize}
	\item $\e = \v$
	\item $\e = \plugx\Xx{\mcall\m\vs}$
	\item $\e = \plugx\Xx{\mcallt\m\vs\MT\t}$
	\item $\e = \plugx\Xx\rdx$
\end{itemize}
\label{lem:uniq-forms}
\end{lemma}
\begin{proof}
By induction on $\e$ using auxiliary definitions and lemmas extended from~\cite{belyakova2020world}; an additional
set of canonical forms representations are added to handle $\e = \plugx\Xx{\mcallt\m\vs\MT\t}$ and casts $\evalchk\e\t$/$\evalinf\e$.
\end{proof}

\begin{lemma}{Context Irrelevance}
$\evalwa{\MT_g}{\plugCx\Cx{\rdx}}{\MT'_g}{\plugCx\Cx\e'} \; \iff \; \evalwa{\MT_g}{\plugCx{\Cx'}\rdx}{\MT'_g}{\plugCx{\Cx'}{\e}}$
\label{lem:ctx-irrelevance}
\end{lemma}
\begin{proof}
By analyzing normal-evaluation steps I can see that only \rdx matters for the reduction;
therefore, by inspecting the reduction step for either $\plugCx\Cx\rdx$ or $\plugCx{\Cx'}\rdx$
I can then construction a corresponding step for $\Cx'$ or $\Cx$, respectively.
\end{proof}

\begin{lemma}{Simple-Context Irrelevance}
$\evalwam{\MT_g}{\e}\MT{\MT'_g}{\e'} \; \implies \; \evalwam{\MT_g}{\plugx\Xx\e}\MT{\MT'_g}{\plugx\Xx{\e'}} $
\end{lemma}
\begin{proof}
By lemma~\ref{lem:uniq-forms} \e is either \v, $\plugx{\Xx_\e}{\m(\vs)}$, $\plugx{\Xx_\e}{\mcallt\m\vs{\MT_i}\t}$, or $\plugCx{\Cx_\e}{\rdx}$. If \e is \v
then the assumption cannot hold (since $\plugCx\Cx{\evalt\MT\v}$ steps to $\plugCx\Cx\v$). Therefore, I need only
consider the $\plugx{\Xx_\e}{\m(\vs)}$, $\plugx{\Xx_\e}{\mcallt\m\vs{\MT_i}\t}$, or $\plugCx{\Cx_\e}{\rdx}$ cases.
\begin{itemize}
	\item When $\e$ is $\plugx{\Xx_\e}{\mcall\m\vs}$ then $\evalt\MT\e = \evalt\MT{\plugx{\Xx_\e}{\mcall\m\vs}}$ is a redex
	and $\plugCx\Cx{\evalt\MT\e}$ steps by \WAE{CallLocal}. Similarly, $\evalt\MT{\plugx{\Xx}{\e}} = \evalt\MT{\plugx{\plugx\Xx{\Xx_\e}}{\mcall\m\vs}}$
	is also a redex and steps as $\plugCx\Cx{\evalt\MT\e}$
	\item When $\e$ is $\plugx{\Xx_\e}{\mcallt\m\vs{\MT_i}\t}$ then $\evalt\MT\e = \evalt\MT{\plugx{\Xx_\e}{\mcallt\m\vs{\MT_i}\t}}$ is a redex
	and $\plugCx\Cx{\evalt\MT\e}$ steps by \WAE{CallLocal}. Similarly, $\evalt\MT{\plugx{\Xx}{\e}} = \evalt\MT{\plugx{\plugx\Xx{\Xx_\e}}{\mcallt\m\vs{\MT_i}\t}}$
	is also a redex and steps as $\plugCx\Cx{\evalt\MT\e}$.
	\item When \e is \plugCx{\Cx_\e}{\rdx} then $\evalt\MT\e = \evalt\MT{\plugCx{\Cx_\e}{\rdx}}$ and $\plugCx\Cx{\evalt\MT\e} = \plugCx{\Cx'}\rdx$
	where $\Cx' = \evalt\MT{\Cx_\e}$. Since $\plugCx\Cx{\evalt\MT{\plugx\Xx\e}} = \plugCx{\Cx''}\rdx$ for $\Cx'' = \plugCx\Cx{\evalt\Xx{\Cx_\e}}$,
	$\plugCx\Cx{\evalt\MT\e}$ and $\plugCx\Cx{\evalt\MT{\plugx\Xx\e}}$ step similarly by lemma~\ref{lem:ctx-irrelevance}.
\end{itemize}
\end{proof}

The proof of soundness for \juliette is straightforward by rule induction on the rule used
to derive $\vdash_{\MT; \PT} \e : \t$:
\begin{proof}
\begin{itemize}
	\item \TER{Value}: trivial, case (1) of soundness.
	\item \TER{Var}: impossible, as $\Gamma$ is empty.
	\item \TER{Seq}: Apply the IH to the first typing relation $\vdash_\MT \e_1 : \t_1$:
	\begin{itemize}
		\item If $\e_1$ is a value, then apply \WAE{Seq}; case 2 of soundness applies.
		\item If there is some $\e'_1$ such that $\evalwam{\MT'}{\e_1}{\MT}{\e'_1}{\MT''}$ and $\vdash_{\MT; \PT} \e'_1 : \t_1$ then let $\Xx = \hole; \e_2$ and therefore by simple context irrelevance $\evalwam{\MT'}{\plugx\Xx{\e_1}}{\MT}{\MT''}{\plugx\Xx\e'_1}$;
		let $\e' = \plugx\Xx{\e'_1} = \e'_1; \e_2$ and therefore $\evalwam{\MT'}{\plugx\Xx{\e_1}}\MT{\MT''}{\e'}$ and since $\vdash_{\MT; \PT} \e_2 : \t_2$ it follows that $\vdash_{\MT; \PT} \e' : \t_2$. Case 2 holds.
		\item If $\e_1$ is an ambiguous method call $\plugx{\Xx'}{\mcall{\m}{\obar\v}}$ then $\e$ is stuck at the same ambiguous method call via the $\Xx$ construction as in case 1. Case 3 holds.
		\item If there is some $\Xx'$ and $\e'_1$ such that $\e_1 = \plugx\Xx{\evalchk{\e'_1}\t}$ then construct $\Xx = \Xx'; \e_2$ and therefore $\e = \plugx\Xx{\evalchk{\e'_1}\t}$; case 4 holds.
		\item As in case 4.
	\end{itemize}.
	\item \TER{Sub}: Apply the IH to the typing relation.
	\begin{itemize}
		\item If $\e$ is a value $\v$ and $\jsub{\typeof(\v)}{\t'}$ then by transitivity $\jsub{\typeof(\v)}{\t}$ and case (1) holds.
		\item If $\evalwam{\MT'}{\e}{\MT}{\e'}{\MT''}$ and $\vdash_{\MT; \PT} \e' : \t'$ then $\vdash_{\MT; \PT} \e' : \t$ by \TER{Sub}.
		\item Trivial.
		\item Trivial.
		\item Trivial.
	\end{itemize}.
	\item \TER{Primop}: Apply the IH inductively over the typings of the primop arguments $\obar{\vdash_\MT \e : \t}$. If any argument $i$ steps then I construct a context $\Xx=\delta_l(\obar\v\Xx\e)$ and
	the expression as a whole steps via simple context irrelevance as in \TER{Seq}. Additionally, the context $\Xx$ suffices to show that if either case 2 or 3 apply for any argument $i$ then the respective
	case applies to the expression as a whole. Otherwise if all arguments are values $\obar\v$ then by the IH case 1 of soundness applies and $\jsub{\typeof(\obar\v)}{\obar\t}$. Then, by the definition of the 
	typed primop resolver $\Psi(l, \obar\v) = \v'$ where $\jsub{\typeof(\v')}{\t'}$. Therefore I can apply \WAE{Primop} to find that $\evalwam{\MT'}{\delta_l(\obar\v)}\MT{\MT''}{\v'}$ and case 2 applies.
	\item \TER{Inferred}: I apply the correctness property of inference~\ref{def:inferencesound} to the embedded expression and case analyze.
	\begin{itemize}	
		\item I know that $\exists \v:\;\e = \v$ and $\jsub{\typeof(\v)}{\t}$. Apply \WAE{CallChecked}; case (2) applies.
		\item I know that $\forall \MT'\, \exists \MT'', \evalwam{\MT'}{\e}\MT{\MT''}{\e'}$ and $\infers{}{\MT}{\e'}{\t}$. I can therefore construct $\Xx=\evalchk\hole\t$ and
		therefore by simple context irrelevance $\evalwam{\MT'}{\plugx\Xx\e}\MT{\MT''}{\plugx\Xx{\e'}}$. By rule \TER{Inferred} $\vdash_{\MT; \PT} \evalinf{\e'} : \t$ and thus 
		case (2) applies.
		\item Construct $\Xx$ as in case 2; the ambiguous method definition remains thus and case 3 applies.
		\item Construct $\Xx$ as in case 2 and case 4 applies.
		\item If $\e$ is an expression then construct $\Xx$ as in case 2 and thus case 5 applies. Otherwise, if there is some value $\v$ such that $\e = \v$ then since $\infers{}\MT\v\t$ it follows that $\evalwam{\MT'}{\evalinf\v}\MT{\MT''}{\v}$ since $\jsub{\typeof(\v)}{\t}$ by case (1) of the definition of inference and thus case (2) applies.
	\end{itemize}
	\item \TER{Checked}: Trivial, case 5 of soundness.
	\item \TER{Call}: Apply the IH inductively as before, using the callee form of simple evaluation contexts for the receiver and the argument form for arguments. If any step then I construct the associated simple evaluation contexts and the expression as a whole steps. 

	If the receiver and all arguments are values with types $\m$ and $\obar\s$ respectively, then one of two cases of $\dispatch$ could have applied: normal or protocol invocation.

	\begin{itemize}
	\item \textbf{Normal invocation}: by the definition of $\dispatch$ there is some $\mdefty{\m}{\obar{\jty{\x}{\t'_a}}}{\mu}{\e} \in \MT$ where $\jsub{\obar{\typeof(\v)}}{\t'_a}$. 
		Therefore, I need to case analyze on if there is a unique most specific applicable method:
		\begin{itemize}
			\item $\getmd(\MT,\m,\obar\sigma) = \mdefty{\m}{\obar{\jty{\x}{\t'_a}}}{\mu}{\e'}$. 
			Case analyze on whether $\MT_e$ is empty or not.

			\begin{itemize}
				\item $\MT_e$ is empty. By the definition of $\dispatch$, $\infers{\obar{\jty{\x}{\t_a}}}{\MT}{\e'}{\t_r}$. By substitution for inference, therefore, $\infers{}{\MT}{\e'\subst\xs\vs}{\t_r}$.
				Thus, $\vdash_{\MT; \PT} \evalinf{\e'\subst\xs\vs} : \t_r$. Moreover, since $\evalwam{\MT'}{\mcallt{\m}{\vs}{\MT}{\t_r}}\MT{\MT''}{\evalinf{\e'\subst\xs\vs}}$ by \WAE{CallInferredLocal} case (2) applies.
				\item  $\MT_e$ is nonempty. Then $\evalwam{\MT'}{\mcallt{\m}{\vs}{\MT}{\t_r}}{\MT\bullet\MT_e}{\MT''}{\evalchk{\e'\subst\xs\vs}{\t_r}}$ by \WAE{CallCheckedLocal}. Then, $\vdash_{\MT; \PT} \evalchk{\e'\subst\xs\vs}{\t_r} : \t_r$ by \WAE{ValChecked}.
			\end{itemize}
			\item $\getmd(\MT, \m, \obar\sigma) = \err$ where $\min(\applcbl(\latest(\MT), \m, \gs)) = \err$ due to there being an equivalence class $T$ of most-applicable signatures that are not otherwise related by subtyping. Case (3) of soundness applies as 
			I let $\Xx = \Xx'$ and then $\plugCx\Xx{\mcall{\m}{\obar\v}} = \m(\obar\v)$. 
		\end{itemize}
	\item \textbf{Protocol invocation}: by definition of $\dispatch$ there was some $\pdef\m{\obar{\jty{}{\t'_a}}}{\t_r} \in \PT$ where $\jsub{\obar{\typeof(\v)}}{\t'_a}$.
	By well-formedness of the protocol table I then have that $\checkproto(\m, \obar{\t}, \t_r, \MT)$ and then that $\forall \vs': (\typeof(\vs') = \gs \wedge \gs <: \obar{\t'_a}), \getmd(\MT, \m, \gs) = \mdefty{\m}{\obar{\jty{\x}{\t}}}{\mu}{\e} \wedge \infers{\obar{\jty{\x}{\t}}}{\MT}{\e}{\t'_r} \wedge \jsub{\t'_r}{\t_r}$. Instantiating $\vs'$ with $\vs$ gives
	(1) $\getmd(\MT, \m, \gs) = \mdefty{\m}{\obar{\jty{\x}{\t}}}{\mu}{\e}$, (2) $\infers{\obar{\jty{\x}{\t}}}{\MT}{\e}{\t'_r}$, and (3) $\jsub{\t'_r}{\t_r}$. Protocol invocation then proceeds
	as with normal invocation with the $\m$ provided by (1), the inference result from (2), and subsumption using (3).
	\end{itemize}
\end{itemize}
\end{proof}

By soundness, then, \emph{within the local method table $\MT$}, no return type
checks need be inserted; so long as the inference result is correct it follows
that all statically-determined return types are then correct. If the local
method table is then extended with some $\MT_e$ I must then start dynamically
checking return types.

This soundness theorem ultimately answers my thesis statement. In particular,
\begin{itemize}
	\item The concrete semantics aligns neatly with Julia's existing concept of
	typing; the theorem wholly relies on Julia's inherent definition of type
	while still being able to produce a strong guarantee.
	\item Soundness does not rely on any property of subtyping besides transitivity, as mentioned earlier.
	\item Performance depends on whether the dynamic local method table matches the
	static method table. If the method tables are the same then no overhead is incurred.
	If they differ then return type tests must be performed.
\end{itemize}

These properties make the nature of the type system dependent on whether new
methods have been added or not. If new methods have not been added then typed
code receives guarantees comparable to those of a fully static language: calls
to other functions are guaranteed to go through and their return types are
guaranteed to be correct. If new methods have been added then the guarantee
degrades to a hybrid between the concrete and transient semantics under my
taxonomy as part of KafKa~\cite{chung2018kafka}. As in the transient
semantics, I must check returned values to ensure that they adhere to the
statically determined types. However, unlike the transient semantics, I do so
by checking their entire identity through the type tag rather than merely their
surface-level structure.

Thus, on paper, I have an answer to my original question: yes, it is possible to
design a gradual type system for Julia that matches the philosophy of the language.
How practical is it though, in its present form?

\chapter{Typed Julia in Practice}

With the theory laid out, I can now examine the realization of Typed Julia. In this
chapter I will describe the implementation of the type checker and discuss how it works 
on a few real programs as well as empirically evaluate some of the assumptions that went
into its design.

\section{Implementation}

The type checker is implemented in Julia and relies extensively on existing
Julia language features. The type checker is a standalone Julia program that
analyzes source files and produces warnings and errors.

The type checking pipeline consists of three packages, two of
which were developed for this project. A general schematic is shown in
\figref{fig:typearch}. 

\begin{figure}
\centering
\includegraphics[scale=0.5]{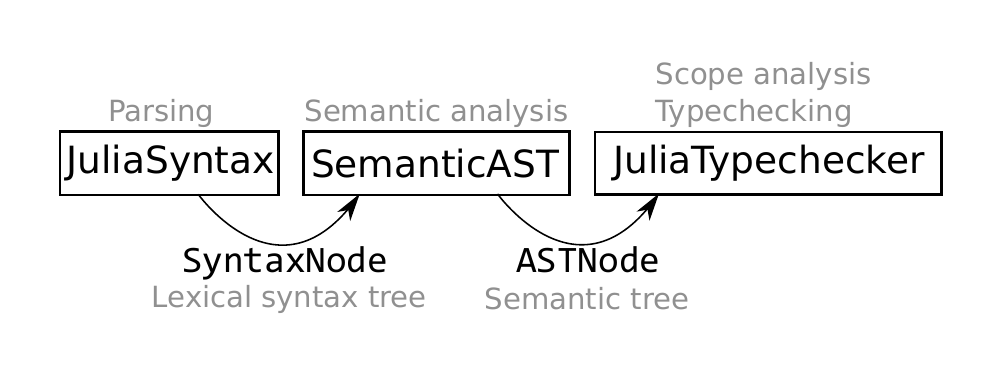}
\caption{Type checker architecture}
\label{fig:typearch}
\end{figure}

Programs go through four phases of type checking:
\begin{itemize}
	\item Syntactic analysis: Julia files are parsed into Julia's expression form.
	\item Semantic analysis: raw Julia ASTs in the form of S-expressions are parsed into a semantic AST.
	\item Module scope analysis: Julia files are traversed to identify the relationship between files and to identify what module each file is logically in.
	\item Type analysis: annotated methods are identified by the type checker and are analyzed.
\end{itemize}

Syntactic and semantic analyses are performed by the dedicated libraries
\texttt{JuliaSyntax} and \texttt{SemanticAST} respectively. Both scope
analysis and type analysis are performed by the package
\texttt{JuliaTypechecker}. \texttt{SemanticAST} and
\texttt{JuliaTypechecker} were developed specifically for type checking.
The broad size of each new package is indicated in
table~\ref{tab:componentloc}; statistics exclude test files.

\begin{table}
\centering
\begin{tabular}{|c|c|c|}\hline
Package & Lines of code & Files \\ \hline
SemanticAST & 1160 & 5 \\ \hline
JuliaTypechecker & 1453 & 7 \\ \hline
\end{tabular}
\caption{Typechecker component packages}
\label{tab:componentloc}
\end{table}

\subsection{Syntactic analysis} Parsing and syntactic analysis is performed using the \texttt{JuliaSyntax}
library, a Julia parser implemented in Julia. \texttt{JuliaSyntax} provides several benefits when compared to
Julia's own parser, including character-precise attribution information as well as a better-described expression
representation. However, \texttt{JuliaSyntax}'s expressions still largely follow Julia's which poses a practical
challenge to semantic analysis. 

Julia expressions (and by extension those produced by \texttt{JuliaSyntax})
are based on S-expressions and are close to the syntax of the input file. As a
result, the same semantic concept can be represented in many different forms. For
instance, function declarations can come with one of five different expression heads
which need to be disambiguated based on their children. 

\subsection{Semantic analysis} My solution to this problem is the Julia
source-level semantic analyzer \texttt{SemanticAST}. \texttt{SemanticAST}
takes the expressions produced by \texttt{JuliaSyntax} and parses them into
semantically meaningful ASTs. In \texttt{SemanticAST} there are only two forms
for function declarations, both of which are fully descriptive at the top level. 

\cbstart
As an example, consider the following definitions:

\begin{footnotesize}
\begin{tabular}{ll}
\begin{normalsize} Source \end{normalsize}& \begin{normalsize} S-expression \end{normalsize}\\
\verb$f(x) = 3$ & \verb$(= (call f x) (block 3))$ \\
\verb$f(x) where T = 3$ & \verb$(= (where (call f x) T) (block 3))$ \\
\verb$f(x)::Int = 3$ & \verb$(= (:: (call f x) Int) (block 3))$ \\
\verb$f(x)::Int where T = 3$ & \verb$(= (:: (call f x) (where Int T)) (block 3))$ \\
\end{tabular}
\end{footnotesize}

All of these definitions describe an identical function. However,
the head and structure of the S-expression differs to match the precise
source syntax. For example, if I specify a type variable then the function's 
head is now \c{where} or if I specify both a type variable and a return type
then the head is \c{::} with two arguments: the ``normal'' function head
followed by a \c{where} structure whose body and type variables are then
interpreted as the return type and the quantifier for the whole function, respectively.

Interpreting Julia's parsed S-expressions is then traditionally difficult: one
must handle a wide range of special cases to align with programmer expectations.
My SemanticAST library provides an abstraction layer over these details: all
function definition have a single representation.

In this example, all four methods are inline definitions of a function named \c{f}.
Inline definitions are handled as any other assignment would be, so all take on the form \c{Assignment([lvalue], [expression])}. The lvalue in these cases are
different \c{FunctionAssignment} instances, while the rvalue is the expression \c{Literal(3)}. \c{FunctionAssignment} then takes on the form \lstinline|FunctionAssignment(name::FunctionName, args_stmts::Vector{FnArg}, kwargs_stmts::Vector{KwArg}, sparams::Vector{TyVar}, rett::Union{Expression, Nothing})|, taking a name, the positional, keyword, and type arguments, followed by the return type.

\begin{footnotesize}
\begin{tabular}{ll}
\begin{normalsize} Source \end{normalsize} & \begin{normalsize} ASTNode for LValue \end{normalsize}\\
\verb|f(x) = 3| & \verb|FunctionAssignment(f,FnArg(x),[],[],nothing)| \\
\verb|f(x) where T = 3| & \verb|FunctionAssignment(f,FnArg(x),[],[TyVar(T)],nothing)|  \\
\verb|f(x)::Int = 3| & \verb|FunctionAssignment(f,FnArg(x),[],[],Var(Int))|  \\
\verb|f(x)::Int where T = 3| & \verb|FunctionAssignment(f,FnArg(x),[],[TyVar(T)],Var(Int))|\\
\end{tabular}
\end{footnotesize}

In all four cases above the L-value is now a consistent instance of \c{FunctionAssignment}. Downstream consumers, such as the type system, need then
only worry about assigning to a function rather than ``what happens if I see a \c{::} in the lvalue of an assignment.''

The implementation of \texttt{SemanticAST} is based on Julia's own semantic analysis step: lowering. Lowering is the first phase of compilation in Julia after parsing and converts parsed S-expressions into a pseudo-SSA representation that's then fed into later stages of analysis. Lowering is, however, unsuitable for use in static tooling as the process loses attribution information and performs numerous other transformations on the program being lowered. In constrast, \texttt{SemanticAST} maintains equivalence to the source program while abstracting parsing details.
\cbend

One challenge for semantic AST analysis is macros. The current implementation
is hard coded; a few common macros are supported but most fall into a
hardcoded exception. Generalized macro handling is a challenge for static
semantic analysis as while the macros could be expanded and their output
analyzed this may lose source-level meaning. A ``method not found'' error
from within a macro expansion is difficult to understand without examining
the macro itself. I will return to this topic when discussing future work.

\subsection{Scope analysis} My next step is module scope analysis. Julia
packages and projects usually consist of multiple files that are all related
by some sort of ``root.'' Julia packages, for instance, are loaded by 
executing their eponymous file (for instance, \texttt{SemanticAST} will be
loaded from the file \texttt{SemanticAST.jl}), which is then responsible for using the \jlinl{include} function to load the remainder of the package.
This root file can use the entire Julia language to decide whether
to \jlinl{include} a file or not. Moreover, \jlinl{include} loads a file into
whatever the \emph{current} module is. Knowing the module that a file was
\jlinl{include}d into is essential in order to determine the module whose
definitions it should be type checked against.

As a brief example, suppose I had three files, as shown in \figref{fig:tinypackage}.
\begin{figure}[H]
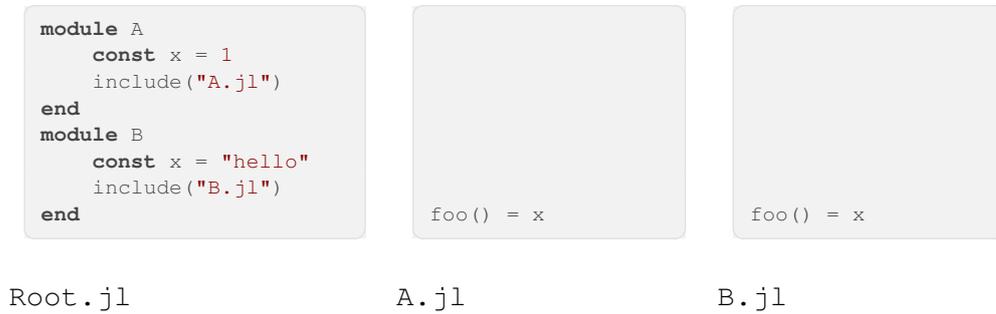

\begin{minipage}{5cm}
\begin{lstlisting}[linewidth=4.9cm]
module A 
	const x = 1
	include("A.jl")
end
module B
	const x = "hello"
	include("B.jl")
end
\end{lstlisting}
\texttt{Root.jl}
\end{minipage}
\begin{minipage}{4.1cm}
\begin{lstlisting}[linewidth=4cm]







foo() = x
\end{lstlisting}
\texttt{A.jl}
\end{minipage}
\begin{minipage}{4.1cm}
\begin{lstlisting}[linewidth=4cm]







foo() = x
\end{lstlisting}
\texttt{B.jl}
\end{minipage}
\caption{A tiny package}
\label{fig:tinypackage}
\end{figure}

This package will load the contents of \texttt{A.jl} into module \texttt{A}
and will load the contents of \texttt{B.jl} into module \texttt{B}. If I call
\jlinl{A.foo()} then I get an integer back; if I call \jlinl{B.foo()} then
I get a string back. I cannot identify which \jlinl{x} is being used without
knowing from where the reference is being \jlinl{include}d from. Julia files
must be analyzed within the context of the entire project.

\cbstart
Scope analysis is then a two-step process: files must first be analyzed
to figure out what ``tree'' that they might be contained in, then their scopes
can be determined based on while file(s) they might be included from.

The first step of scope analysis is ``tree'' identification: determining
which files include what other files. The algorithm begins with a set of 
``roots'' (top-level files from which other files are used) and then 
removes roots when a reference from one preliminary root to another is identified.
In pesudocode the algorithm can be described as
\begin{lstlisting}
roots = the set of all files 
for file in roots
	for referenced_file in references(file)
		if referenced_file in roots
			remove(roots, referenced_file)
		end
	end
end
\end{lstlisting}

Once root identification has occured scope identification can start. Scope
identification starts at each root and runs recursively into each referenced
file, propagating the scope from the referencing file into the referenced. Again,
in pseudocode, this works as follows:
\begin{lstlisting}
function analyze_scope(outer_scope, file)
	for module in file
		inner_scope = combine(outer_scope, module)
		for referenced_file in references(module)
			analyze_scope(inner_scope, referenced_file)
		end
	end
end

for root in roots
	analyze_scope(main scope, root)
end
\end{lstlisting}

The implementation of scope analysis has two major limitations: it miss files
that are added by calling \jlinl{include} with some variable value;
similarly, it may collect files that are included from within a method that
is never used. However, most \jlinl{include} usages by packages are at the
top level with a hardcoded string. Scope analysis takes around 200  lines of
code.

\subsection{Type analysis} Finally, I preform type analysis. Using the scope
 information provided by the previous step the type checker performs a
 recursive descent through each function annotated with the special-cased \c
 {@typed} macro in the set of files provided. 

 The implementation of the type checker follows the typing rules as described in 
 figure~\ref{typesystem}. It recursively descends into expression forms while 
 maintaining a type checking context to determine the type of an expression.
 The code itself is small at 1,100 lines; A number of architectural decisions deserve discussion, however. 

\cbend

\paragraph{Entity-component system} The type checker maintains a side data
structure to the AST in entity-component form. Entity component systems are a
design pattern from video games that support large numbers of lightweight
entities that can have many components. Each component contains
some additional information. This design is used in video games to
decouple various high-level behaviors from one another by making them only
have to consider which entities have which components.

The type checker uses aspects of this entity-component architecture to allow
it to handle semantic information about the AST. Each AST node is associated
by the type checker with an entity; each entity has an AST component that
indicates which node it is attached to and what its parent is. The type
checker then attaches various additional components to these entities as it
type checks the program including what the inferred type of a given
sub-expression is, what methods a given call site could be dispatched to, and
what the errors at a given location were.

Using entities and components to capture this metadata allows very
straightforward use of the information by other parts of the type checker. One
example is that scope information is included as a component. Moreover, it
allows weak coupling between type checker components as analyses only need
to consider the components that are relevant to their task.

\paragraph{Method tables} As in the theory, the implementation performs type
checking and inference against some reference method table. My implementation
prepares a Julia instance with the package loaded to act as this table. In
turn, this Julia instance then is used to 
\begin{itemize}
	\item resolve subtyping queries,
	\item find applicable methods,
	\item infer return types for a given invocation.
\end{itemize}
By using a real running Julia instance the method table used for type checking
is closely aligned to the method table used at runtime by most packages. Definitions
are dynamically reloaded when files are changed using \texttt{Revise.jl}.

Using an actual Julia instance to serve as the ``black box'' for type checking
closely aligns the type checker with Julia's runtime reasoning about types but has
drawbacks. In particular, Julia cannot dynamically reload redefinitions of the same
type and the instance must be reinitialized from scratch each time a type definition
is changed. 

\paragraph{Inference} Again as in the theory I use type inference to determine
the return types for invocations. The implementation uses its Julia instance
to perform this inference, leveraging the same inferencer used for runtime
type specialization of methods. My type checker queries it with the statically
determined arguments for a given method to determine the inferred return type.

\begin{figure}
\centering
\includegraphics[scale=0.6]{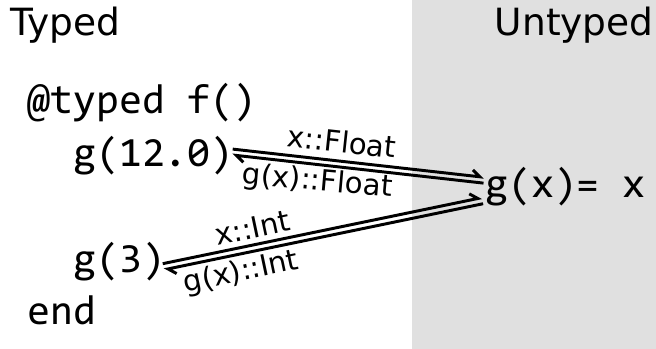}
\caption{Inference example}
\label{fig:inference}
\end{figure}

A consequence of using Julia's inference system to determine the return type for each
invocation site is that different typed calls to the same untyped method may
result in different return types. \figref{fig:inference} illustrates a simple
example. Here, function \jlinl{f} calls the untyped identity function
\jlinl{g} with the floating point number \jlinl{12.0} and the integer
\jlinl{3}. The inferred return type for \jlinl{g} at the first call site is
\jlinl{Float64} while at the latter it is \jlinl{Int} even though the same
function is being called.

Rerunning inference for each call site has several benefits but also
tradeoffs. When used with concrete types the inference returns very precise
types and aligns closely with the dynamic behavior of Julia's optimizer.
However, the inferencer struggles with abstract types. For example, the
inferred return type for \jlinl{+} when applied to two \jlinl{Numbers} is
\jlinl{Any}. As I will see, imprecision of abstract inference can pose
challenges for typing library code.

\section{Evaluation}

The practice of a type system is an inherent part of its design; at several
times in my theoretical treatment I have made design decisions fundamentally
motivated by practical justifications. Moreover, my type system is a basic
framework from which more of Julia's features may be covered, but I do not
know just how comprehensive my treatment of Julia is.

I wanted to answer three questions about the type checker:
\begin{itemize}
	\item What is the right choice of method table? How many programs might use \eval to extend their method table after they are first loaded and thereby require return type checks?
	\item How common are signatures in actual code and are they completely implemented?
	\item What mutations are needed for existing Julia programs to be type checkable?
\end{itemize}

\subsection{Picking a method table}

As mentioned in the theoretical examination of the type system it is important
to carefully chose the method table \MT that I type check against. If the
table \MT is the one actually used at runtime then the type system is
``free'': no dynamic checks are needed. If there are additional methods added
after the fact then this guarantee goes out the window (though practically
only for functions that \emph{actually} have new methods added as a result of
JIT compilation).

As part of my earlier work on world age and method tables in
Julia~\cite{belyakova2020world} I conducted an emperical evaluation of uses of
\eval in the wild. I will summarize the relevant results here.

I wanted to examine how frequently libraries either modify the method table
themselves after initialization or support user programs that do.
Self-modifications of the method table after initialization would need to use
\eval somewhere within a method. If a library wishes to support a user program
that adds new methods after initialization then it must call those methods
using either \eval or the \invokelatest function from inside a method.
Therefore, I statically analyzed a corpus to examine how many packages use
\eval or \invokelatest from withing methods.

My corpus consists of all 4,011 registered Julia packages as of August 2020.
The results of statically analyzing the code base are shown in
\figref{fig:wapkgusage}. The analysis shows that 2,846 of the 4,011 packages
used neither \eval nor \invokelatest, and thus are definitely age agnostic. Of
the remaining packages, 1,094 used \eval only, and so \emph{could} be
impacted. 15 packages used \invokelatest only, and some 56 used both. I can
reasonably presume that at least these latter 71 packages are impacted by
world age because they bypass it using \invokelatest.

\begin{figure}
  \centering
  \subfloat[\eval and \invokelatest use by package]{
  	\label{fig:wapkgusage}
	  \includegraphics[scale=0.5]{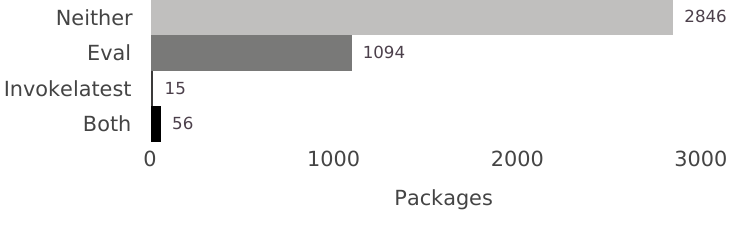}
	  \vspace{-3mm}
  }
  \hfill
  \subfloat[Static use of AST forms in all packages]{
  	\label{fig:all_static_ast}
	  \includegraphics[scale=0.5]{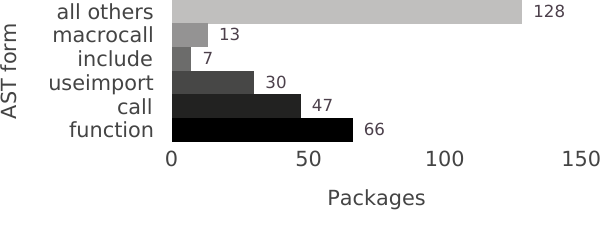}
	  \vspace{-3mm}
  }
  \caption{Dynamic code generation usage metrics}
\end{figure}

To understand if packages that only use \eval dynamically define or use
dynamically defined methods, I statically analyzed the location of calls to
\eval and their arguments by parsing files that contain \eval. For each call,
I  classified the argument ASTs, recursively traversing them and counting
occurrences of relevant nodes. The analysis is conservative: it assumes that
an AST that is not statically obvious (such as a variable) could contain
anything.
\figref{fig:all_static_ast} 
shows how many packages use method table relevant AST forms. Only uses of
\eval from within methods---where the method table could be relevant---are shown;
top-level uses of \eval---which cannot be affected by world age---are filtered out.
The ``all-others'' category encompasses all AST forms not relevant to world age.
While this aggregate is, taken as a whole, more common than any other single
AST form, none of the constituent AST forms is more prevalent than
function calls. Therefore, most common arguments to \eval are function definitions,
followed by function calls and loading of modules and other files.

Using the results of the static analysis, I estimate that about 4--9\% of the
4,011 packages might be affected by world age. The upper bound (360 packages)
is a conservative estimate, which includes 289 packages with potentially
method table-related calls to \eval but without calls to \invokelatest, and
the 71 packages that use \invokelatest. The lower bound (186 packages)
includes 115 \invokelatest-free packages that call \eval with both function
definitions and function calls, and the 71 packages with \invokelatest.

\cbstart
This analysis only considers usages of \eval that are not at the top level---
that is, usages that are within function bodies or similar constructs. Using \eval
at the top level is a common practice in Julia packages to generate boilerplate
definitions. From the perspective of world age, these usages of \eval are identical 
to writing these same definitions explicitly and are visible in the ``as-imported''
method table used by the type checker. The analysis is conservative, however, in
that it considers \eval used in \emph{any} non-top-level context as being
``not-top-level.'' Some of these methods may be helpers that are only used from
the top level at import-time to define methods that are then visible in the as-imported
method table; the analysis will consider these usages of \eval to be not top-level
and it thus overapproximates the number of packages that use \eval below the top level.
\cbend

To validate my static results, I dynamically analyzed 32 packages out of
the~186 identified as possibly affected by world age. These packages were
selected by randomly sampling 49 packages and then removing packages that did
not run, whose tests failed, or that did not call \eval or \invokelatest at
least once during testing. Over this corpus, the dynamic analysis was
implemented by adding instrumentation to record calls to \eval and
\invokelatest, recording the ASTs and functions, respectively, as well as the
stack traces for each invocation.

\begin{figure}
  \centering
  \subfloat[Static use of AST forms used by package]{
  	\label{fig:stasts}
  	\includegraphics[scale=0.5]{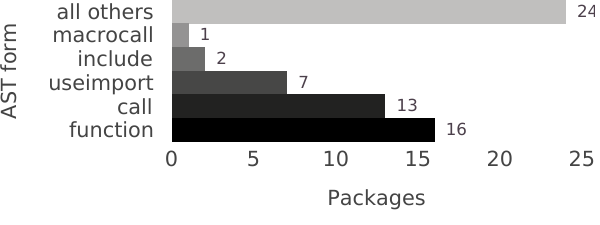}
	\vspace{-3mm}
  }
  \hfill
  \subfloat[Dynamic use of AST forms used by package]{
  	\label{fig:dynasts}
	  \includegraphics[scale=0.5]{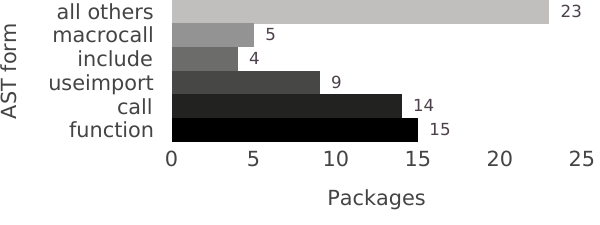}
	  \vspace{-3mm}
  }
  \caption{AST form usage within \eval}
\end{figure}

The results of the static and dynamic analysis of the 32 packages
are given in \figref{fig:stasts} and \figref{fig:dynasts}, respectively.
Both analysis methods agreed that the
most common method table-relevant use of \eval was to define functions, followed
by making function calls and importing other packages. In general, the dynamic
analysis was able to identify more packages that used each AST form, as it can
examine every AST ran through \eval, not only statically declared ones. However,
this accuracy is dependent on test coverage.

As a result of this analysis, I conclude that the method table that results
from evaluating but not using a given package is the same method table that
is used at runtime for between 91\% and 96\% of packages.

\subsection{Protocols}

\vspace{0.1em}\noindent\begin{minipage}{\textwidth}
To evaluate the prevalence of protocols in Julia code I performed a small corpus
analysis of Julia packages to identify how many protocols they defined and what
patterns occurred within those protocol definitions.

I analyzed  the base library as well as the top 10 most-starred Julia
packages on Github as well as their first-order dependencies. I selected this
corpus as many of these packages provide abstraction layers over other systems
(such as JuMP.jl, a numerical optimization abstraction library, or
DifferentialEquations.jl, an abstraction library over numerical integrators).
Moreover, the popularity of these packages indicates that they are widely used
and thus that their protocols are important for the broader Julia ecosystem.
The full list of root packages is provided in \figref{fig:corpuspkgs}; their
first-order dependencies produce a total of 200 corpus packages.

\begin{wrapfigure}{r}{0.45\textwidth}
\begin{tabular}{cc}
DuckDB & Pluto \\
Flux & IJulia \\
DifferentialEquations & Genie \\
JuMP & MakieCore \\
Gadfly & Turing
\end{tabular}
\caption{Corpus packages}
\label{fig:corpuspkgs}
\end{wrapfigure}
The lack of consistent machine-checkable protocol
specifications means that there is no source of truth to check against; the
protocols must be discovered from the source code alone. Moreover, there may
be no use sites for a given protocol in a given package if the protocol is
only intended for external consumption. As a result, the protocol analysis I
used for the corpus exploration is differently constructed than the one that
the type system utilizes.
\end{minipage}

In order to be able to identify protocols from definitions I adopt a
simplified, weaker, version of protocols versus the statically checked one.
The protocol checker requires that there exists an implementation for every
instantiation of the protocol argument typing that returns a value of the
correct type. However, checking this is impossible given definitions alone
for no protocol argument or return types are available. Instead, my definition
of a protocol for the purposes of corpus analysis is a set of methods that:
\begin{itemize}
	\item have the same name,
	\item specialize on every subtype (tag, variables are not considered) of some abstract type $A$ in a consistent position $i$,
	\item is not implemented for $A$ in position $i$ or a supertype thereof
\end{itemize}
For example, 
\begin{lstlisting}
abstract type A end
struct B <: A end; f(::B) = 1
struct C <: A end; f(::C) = "hi"
\end{lstlisting}
satisfies this definition as there is a set of methods with name \jlinl{f} 
that take every subtype of \jlinl{A} despite the return types not being
consistent. Note that the analysis does not consider varargs functions.

The last component of the definition excludes methods that are implemented
abstractly. If I were to define an implementation of \jlinl{f} that takes any
\jlinl{A} like
\begin{lstlisting}
f(::A) = 9
\end{lstlisting}
then that implementation is \emph{not} a protocol; it is only a single
implementation that can handle multiple cases. As a result, while concrete
specializations might exist (as they would in the example of \jlinl{f})
I exclude these cases.

\begin{figure}
\includegraphics[scale=0.6]{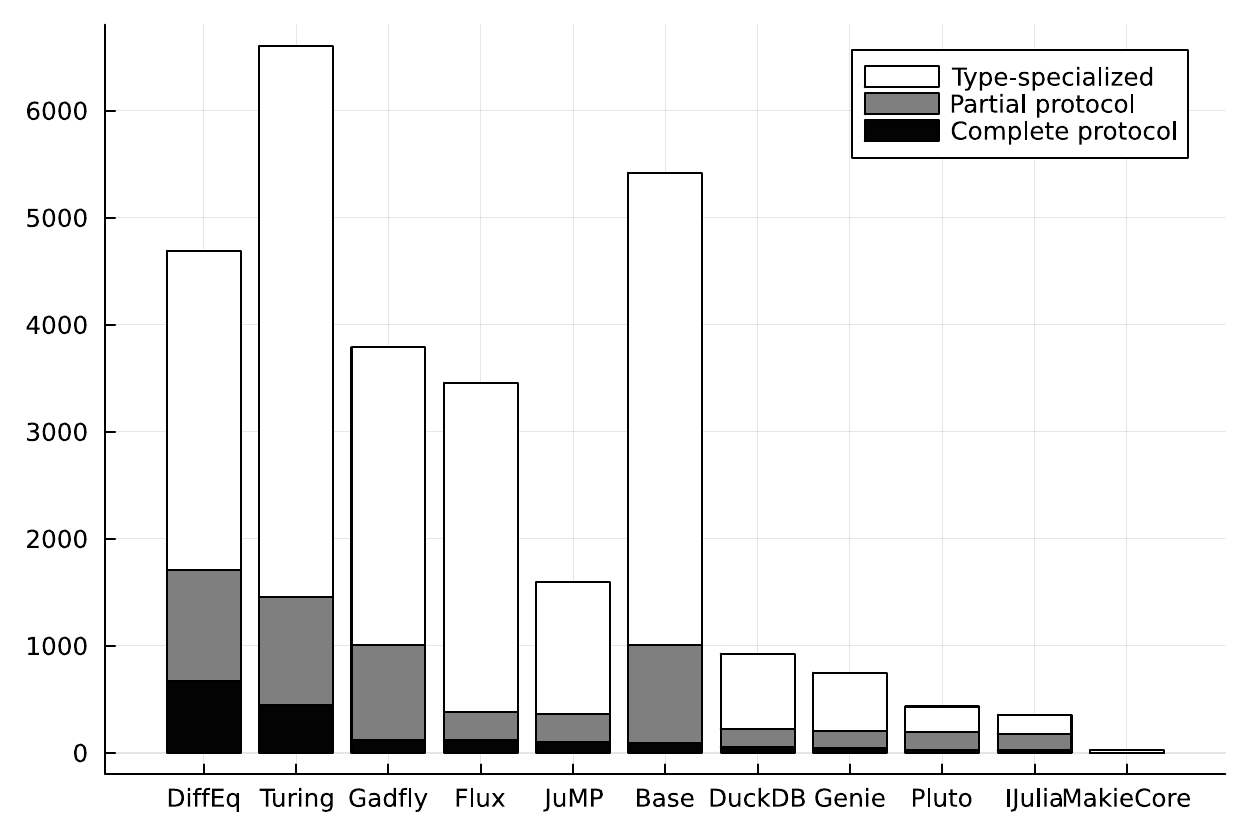}
\caption{Protocols defined by package}
\label{fig:protosbypackage}
\end{figure}

Figure~\ref{fig:protosbypackage} shows the number of protocols defined by each
package.  Sets of definitions that satisfy all three of my properties are
complete protocols. A set of definitions that have a implementation in the
supertype or are missing some implementations are considered partial
protocols. Finally, any two methods with the same name that specialize on
subtypes of the same concrete type are considered ``type-specialized.''

I can clearly see that many protocols are defined by these large packages and
their dependencies. One example is \texttt{DiferentialEquations.jl}
(abbreviated as DiffEq) which defines 200 strict protocols. However, this is
dwarfed by the number of partial protocols and type specializations detected
by the analysis. There may be many more protocols \emph{intended} but that are
not properly implemented.

A simple example of an incomplete protocol implementation can be seen in the
protocol defined by the MathOptInterface library for the
\jlinl{AbstractFunction} type. The documentation says that every
\jlinl{AbstractFunction} must implement \jlinl{constant(::AbstractFunction)}
that returns the constant component of the function, and \emph{nearly} all
implementations do.  The protocol is not fully implemented, however, as a type
exists for which there is no matching implementation of \jlinl{constant}:
\jlinl{MathOptInterface.FileFormats.MOF.Nonlinear}. Thus, the checker
categorizes \jlinl{constant} as a partial protocol since there are concrete
subtypes of \jlinl{AbstractFunction} for which no implementation exists.

Several takeaways can be derived from this observation. In this light the
large number of complete protocols is notable; it reinforces the concept that
protocols are important for Julia programmers and that they work to maintain
them. At the same time, these results emphasize the need for checked protocol
declarations due to the number of intended but faulty protocols. Thus, the
ability to statically enforce protocol adherence may be useful for Julia
package developers and the concept of protocols that I describe here is
a common use case for many packages. 

\subsection{Case Study}

To evaluate how practical the current state of the type system is I considered two case studies:
\begin{itemize}
	\item A obstacle-avoidance trajectory optimizer for quadcopters based on a numerical optimization algorithm, representing ``user'' code.
	\item Julia's \texttt{math.jl} that implements several basic math operations (such as logarithms), representing ``library'' code.
\end{itemize}
In total the two case studies comprise around 2,500 lines of code.

\paragraph{User code} One Julia use case is to write concrete analyses, where
specific values and types are known. To evaluate the utility of the type
checker on such code, I considered an implementation of the PIPG
algorithm~\cite{yu2020proportional}. This program determines the flight path
of a quadcopter that avoids two obstacles in its path using numerical
optimization. I aim to statically type all of the user code; the libraries
that the code uses will remain untyped. 

The trajectory optimization routine depends on several existing libraries:
\begin{itemize}
	\item LinearAlgebra, part of Julia's standard library
	\item StaticArrays, which provides statically-typed sized matrices and vectors
	\item JuMP, an abstraction library over various numerical solvers
	\item Plots, to visualize the output trajectories
\end{itemize}
Invocations into these untyped libraries are resolved as described earlier using inference.

Type checking the program required both adding argument and return type annotations
and removing metaprogramming. I added argument type annotations to all methods
based either on the comments for each function or based on their usage in the
code; I added return type annotations only to those methods that returned
meaningful values. 

Metaprogramming is a larger challenge for the practical type checker. The
optimizer relies on the JuMP abstraction layer, as mentioned, which introduces
a problem description EDSL implemented as macros. A trivial example is
\begin{lstlisting}
@variable(model,x[1:nx,1:N])
\end{lstlisting}
which binds \jlinl{x} to be an array of variables of size \jlinl{nx, N} inside the \jlinl{model}.
I replace introduction forms such as \jlinl{@variable} with binding forms that introduce
variables of the correct types. In the case of \jlinl{x}, the above macro invocation is replaced
with
\begin{lstlisting}
x = nothing::Matrix{JuMP.VariableRef}
\end{lstlisting}
While the type assertion will fail at runtime, the cast to a matrix of \jlinl{JuMP.VariableRef}s
causes \jlinl{x} to be correctly statically typed.

The results of modifying and type checking the program are shown in the first column of table~\ref{tab:typability}.

\begin{table}
\begin{tabular}{|c|c|c|c|c|c|}
\hline \multirow{2}{*}{Program} & \multicolumn{2}{c|}{Functions} & \multicolumn{2}{c|}{Argument annotations} & \multirow{2}{*}{Return types} \\
\cline{2-5} & Typed & Untypable & Concrete & Abstract & \\ \hline
PIPG & 43 & 2 & 106 & 7 & 3 \\ \hline
\texttt{math.jl} & 21 & 59 & 63 & 83 & 69 \\ \hline
Concrete \texttt{math.jl} & 77 & 3 & 134 & 12 & 69 \\ \hline
\end{tabular}
\caption{Typability of methods in case study programs}
\label{tab:typability}
\end{table}

Most functions in this program are typable; only two functions must remain
untyped. The majority of fuctions work by mutating their arguments; only three
functions actually returned values and thus needed return type annotations.

Of the added type annotations almost all were concrete. The only exceptions
were the methods \jlinl{construct_G} and \jlinl{construct_H} which initialize
new problem definition matrices.

The function \jlinl{construct_G} is untypable since the line
\begin{lstlisting}
cat([cat(value.(u)[:,t],value.(x)[:,t+1],dims=1) for t=1:N-1]...,dims=1)
\end{lstlisting}
is rejected as the type checker cannot determine the length of the
array being splatted into \jlinl{cat}. I need to know the length
of the arguments in order to statically determine which method might
be called. With such a variable-length argument the type checker cannot determine a precise return type. 

The second untyped method, \jlinl{construct_H}, is untypable because of the
array index access \jlinl{B[1]}. The root of the problem is that the type of
\jlinl{B} is the abstract type \jlinl{Vector\{<:SMatrix\{M, N, Float64\} where
\{M, N\}\}}, or a vector of statically-sized matrices of currently indeterminate size
whose elements are all \jlinl{Float64}s. Accessing \jlinl{B[1]} should
clearly return a value that is some instance of \jlinl{SMatrix\{M,
N, Float64\} where \{M, N\}}. The Julia type inferencer cannot determine
this fact and instead indicates that the return type is \jlinl{Any}. This
imprecision then causes all downstream operations to fail.

The array access \jlinl{B[1]} going wrong due to imprecise inference is
illustrative of the limitations of the type checker. So long as the program is
concretely typed (that is, works over Julia tag types) programs can frequently
be type checked without needing to modify their implementation even if they use
untyped functionality.

Libraries are more challenging than user code, however. Library code is supposed
to be generic; it cannot simply say that every argument is concretely typed. I will
next examine how readily a small library can be type checked.

\paragraph{Libraries} The \texttt{math.jl} file was the next target, consisting of 
Julia's implementation of several basic math functions. As an example, one function
provided by \texttt{math.jl} is 
\begin{lstlisting}
function clamp!(x::AbstractArray, lo, hi)
    for i in eachindex(x)
        x[i] = clamp(x[i], lo, hi)
    end
    x
end
\end{lstlisting}

Here, I have an implementation of the \jlinl{clamp} function generalized to
clamp each element of an array to be between \jlinl{lo} and \jlinl{hi}. Here
I see the main challenge posed when typing library functions: generic or
otherwise underspecified arguments.

Abstraction and generic arguments are much more common in libraries compared
to user code. The quadcopter knew that the trajectory being optimized was
represented as an array of  \jlinl{Float64} values. In contrast, this
implementation of \jlinl{clamp!} knows nothing about what \jlinl{lo} or
\jlinl{hi} might be at runtime; all it knows is that they are subtypes of
\jlinl{Any} which is not helpful. The only property guaranteed by the original
type annotations is that \jlinl{x} is an \jlinl{AbstractArray}. Adding more
precise types is then difficult without knowing what code uses this function.

A reasonable typing of this \jlinl{clamp!} function might be
\begin{lstlisting}
clamp!(x::AbstractArray{T}, lo::T, hi::T)::T where T<:Number
\end{lstlisting}
I specify here that there must be a single concrete subtype of \jlinl{Number}
such that \jlinl{x} is some sort of array of numbers and that both
\jlinl{lo} and \jlinl{hi} are also instances of said number.

Making a best-effort attempt to type the methods produces the results seen in
the second column of table~\ref{tab:typability}.  The input program was
already heavily annotated. The majority of added types were return types or
specializations of an existing type annotation which are not counted. However,
in spite of these annotations, few definitions were statically typeable. The
most common reason for a method failing to typecheck was some form of
unhandled abstraction.

An example of such abstractions can be seen in the earlier \jlinl{clamp!}
example. Consider, for a moment, what is the type of \jlinl{i}? It should
clearly be the type of the elements produced from \jlinl{eachindex}, but the
documentation says that
\begin{quote}
For array types that have opted into fast linear indexing (like Array), this
is simply the range 1:length(A). For other array types, return a specialized
Cartesian range to efficiently index into the array with indices specified for
every dimension. For other iterables, including strings and dictionaries,
return an iterator object supporting arbitrary index types (e.g. unevenly
spaced or non-integer indices).
\end{quote}
As a result, I do not really know what \jlinl{i} actually is. The only type
I can reasonably give \jlinl{i} is \jlinl{Any}. This then sets the stage for
ensuing chaos with the next problem: \jlinl{x[i]}.

In Julia, the syntax \jlinl{x[i]} entails the multimethod call \jlinl{getindex(x, i)}.
In this case, I know that \jlinl{x} is an \jlinl{AbstractArray{T}}; I only know that
\jlinl{i} is an instance of \jlinl{Any}. Referring to the Julia documentation for
the \jlinl{AbstractArray} set of protocols, I find that there only needs to be
either a method \jlinl{getindex(A, i::Int)} or a method \jlinl{getindex(A, I:Vararg{Int, N})} defined for each \jlinl{AbstractArray}. Therefore,
with our annoyingly-\jlinl{Any} typed \jlinl{i} I cannot invoke any of them. 

The underlying problem is twofold; one an artifact of the implementation and
the other a product of Julia's scale. First, on the implementation end Julia
type inference is optimized to work on concrete types and while it \emph{can}
work on abstract types, sometimes, it tends towards imprecision. Moreover,
Julia's type inference system struggles to deal with type variables, as seen
earlier with the \jlinl{StaticArrays} example, to the point that it simply
cannot infer a return type in some cases.

Even if I were to implement a new type inference algorithm that could better
handle generically-typed code I run into another problem: representation of
abstraction. Protocols, as described earlier, capture the case when there is
an implementation of a function that can handle every concrete instantiation
of the protocol type. However, as seen in the example of \jlinl{eachindex},
real Julia code also has much more complicated abstractions then can be
represented with protocols alone. 

In order to type \jlinl{eachindex} effectively I need \jlinl{AbstractArray}
to specify its iterator type. I could then treat this as a generic type
dependent on the value of \jlinl{x} and ensure that operations on this type
were safe. This approach, effectively a multimethod version of path dependent
types~\cite{amin2014foundations}, would allow this to type check. Similarly, we
could introduce a new existential type variable to \jlinl{AbstractArray} that
encodes the type of the iterator. However, both approaches are substantial
additions to Julia and would break backwards compatibility.

The practically encountered abstractions are thus much more complex than the
simple protocols that my system can check. Moreover, the documentation about
what protocols are necessary to implement for an new instance of a type to
work \emph{correctly} is frequently lacking. Even in this example---a tiny
mathematical function included as part of Julia's own core library---the
documented protocols were insufficient for it to type check correctly.

Representation of abstraction is clearly a key problem for type checking
Julia. Does \texttt{math.jl} pose problems for the type checker besides its
use of abstraction, however? To answer this question, I re-annotated
\texttt{math.jl} with concrete types. For example, the annotations for
\jlinl{clamp!} are now:
\begin{lstlisting}
clamp!(x::Vector{Float64}, lo::Float64, hi::Float64)
\end{lstlisting}

No additional annotations were added to the generically-annotated version; the
only change was that existing generic annotations were replaced with concrete
instantiations of each. As seen in the third column of
table~\ref{tab:typability}, I can see that the same code when given concrete
types generally type checks. Only three methods were then left untyped:
\begin{itemize}
	\item \jlinl{literal_pow}, which dispatches on specific concrete values
	in the type and cannot be practically statically analyzed.
	\item \jlinl{hypot}, which splats arguments for a function call
	and thus the invocation target cannot be resolved.
	\item \jlinl{_hypot} (a helper for \jlinl{hypot}) that uses a
	first class function and is thus not supported.
\end{itemize}

The type checker can thus type almost all operations used in \texttt{math.jl},
but only when concretely typed.

\chapter{Conclusion}

In this dissertation I have described a gradual type system for Julia. My
approach can provide a strong soundness guarantee for typed code and requires
no additional dynamic checks in typed or untyped code so long as no new
methods are added with \eval. 

Providing this strong guarantee is a product of both multiple dispatch and
Julia's design. Multiple dispatch effectively answers one of the key questions
of gradual typing, how to establish type membership, for me. With multiple
dispatch a method will only get called with values that are members of its
argument types. Moreover, Julia's emphasis on type inference allows me to
infer return types for almost any function. These systems then provide
assurance about typed arguments and the result of calling untyped methods from
within a typed context thereby allowing the elimination of runtime checks.

While Julia's design facilitates this core of a type system, it is missing
the primitives needed for typed abstraction. Julia provides no mechanism for
developers to declare common functionality between types, related or otherwise.
As a result, developers frequently create bugs when types do not implement or 
improperly implement some expected functionality.

Julia's subtyping was also a considerable challenge. Julia's extensive usage
of types is both blessing and curse for static analysis for it also begets a complex
subtyping relationship.  I showed that subtyping was undecidable, which is not
a crippling blow in practice, but suggests that firm theoretical results about
subtyping may be hard to come by. My theory treats subtyping (and the related
protocol completeness problem) as a ``black box,'' with the implementation
using Julia's own subtyping system and a naive completeness checker.

Finally, I described a protocol system that solves part of the abstraction
problem. Protocols capture the case where a method exists for every concrete
instantiation of some signature. Protocol declarations then statically enforce
the existence of suitable methods and can be used by the type checker to let
the user call a function they may not otherwise be able to. I additionally
showed a basic algorithm for deciding whether a protocol has been implemented
or not.

\section{Future Work}

My type system for Julia is a foundation. As seen in the case studies it is
very applicable to programs where concrete types are known but struggles with
abstractly typed code. Abstractly typed code stands to benefit the most from
static checking, however. The clearest parts of future work are thus in
building on top of this foundational type system to support a larger set of
the abstractions that have grown up in the Julia community. One simple example is
to support a larger swathe of Julia's type language for use in protocol declarations.

\paragraph{Protocols.}

The protocol system that I describe is able to capture some of the protocols
in Julia code but has several major limitations. In particular, it cannot
handle generic types and type variables. The underlying issue is the same as
was seen with subtyping: dealing with bounded type variables is hard. A
protocol checker for Julia is trivially going to be undecidable; it is easy to
see that a protocol checking procedure could be used to decide subtyping (at
least under the semantic definition) and thus the same proof applies.

Accepting this undecidability one approach would be to apply one of several
more sophisticated pattern matching completeness algorithms. In particular,
Lower Your Guards~\cite{graf2020lower} describes a compositional completeness
checker that may be able to handle the more general Julia type language.
However, several challenges (particularly the number of constructors for some
types such as \jlinl{Any}) make its application to protocol checking not
trivial. 

Protocols lack generality in one key way, however: they only apply to
subtypes of one specific abstract type vector. Frequently users wish
to have some common behavior that is shared between otherwise unrelated types.
Protocols cannot capture this requirement. Instead, a trait system fits this
need better.

\paragraph{Traits.} 

The Julia community has been interested in the concept of traits for some
time. Traits capture some shared behavior that exists outside of the type
hierarchy~\cite{scharli2003traits} allowing the programmer to write code
that is even more general than the nominal type hierarchy would allow.
Julia developers have came up with several ad-hoc ways of defining and using
traits (usually by having some method that exists and returns a sentinel value
for a type that supports the trait), but none are statically checked.

Designing and implementing a trait system for Julia would further the ability
to type abstract code. One problem, for example, is that the seemingly-simple
definition of the \jlinl{+} operator has been of some contention over time
leading to there being no consistent implementation for all \jlinl{Numbers};
in effect, the practice of \jlinl{+} is that ``you know it when you see it.''
Being able to declare, enforce, and use traits would substantially simplify
this problem.

The primary challenge in developing a trait system in Julia is integrating it
into multiple dispatch. Users would like to be able to say that an argument is
an instance of a trait, which is easy enough. More challenging are cases such
as ``this array holds only trait implementations'' or ``this type member is
inhabited with subtypes of this abstract trait.'' A good implementation of
traits needs to be efficient, roughly match the expectations of existing
programmers, and support a large subset of Julia's type language.

\paragraph{First-class functions.}

This work does not type check first-class functions for Julia; its support for
higher-order functions like \jlinl{map} is wholly through special-casing and
hard coding. Introducing first-class functions (and a suitably specific type)
makes the key problem of gradual typing much harder in a multiple dispatch
setting. Consider, for a moment, that you have
\begin{lstlisting}
f(x::Function{String, Int})::Int = x("hello") + 2
f(x::Function{String, String})::String = trim(x("world"))
\end{lstlisting}
and then you call it with \jlinl{f(x -> x)}. How do you decide which implementation
of \jlinl{f} should be called without having to do deep analysis of the lambda?

The multiple dispatch setting makes many gradual typing approaches
impractical. The behavioral semantics, for example, would tell us that we
should wrap \jlinl{x} in a proxy object that enforces the type I expect on
it. That is great, but does not tell us \emph{which} of these implementations
I should call in the first place. Other gradual type systems have analogous problems
as they rely on knowing what type the higher-order functionality \emph{should be}
when the dispatch system needs to know what it \emph{is}.

One potential solution would be to only dispatch on typed function. This
approach is straightforward and would be expressible within Julia's existing
type language. Requiring typed lambdas would, however, prevent untyped code
from being able to invoke specific typed implementations.

\appendix
\cleardoublepage
\cleardoublepage
\bibliography{thesis,bibs/st-dec,bibs/gc,bibs/jlov,bibs/proposal}

@String{DLS    = {Symposium on Dynamic languages (DLS)}}

@String{ECOOP  = {European Conference on Object-Oriented Programming (ECOOP)}}

@String{ICFP   = {International Conference on Functional Programming (ICFP)}}

@String{OOPSLA = {Conference on Object Oriented Programming Systems Languages and Applications (OOPSLA)}}

@String{POPL   = {Symposium on Principles of Programming Languages (POPL)}}

@InProceedings{Kennedy07,
  author    = {Kennedy, Andrew and Pierce, Benjamin C.},
  booktitle = {Workshop on Foundations and Developments of Object-Oriented Languages (FOOL/WOOD)},
  title     = {On Decidability of Nominal Subtyping with Variance},
  url       = {https://www.microsoft.com/en-us/research/publication/on-
                  decidability-of-nominal-subtyping-with-variance/},
  year      = {2007},
}

@InProceedings{Chambers94,
  author    = {Chambers, Craig and Leavens, Gary T.},
  booktitle = {Conference on Object-oriented Programming Systems, Language, and Applications (OOPSLA)},
  title     = {Typechecking and Modules for Multi-methods},
  doi       = {10.1145/191080.191083},
  year      = {1994},
}

@InProceedings{chung2018kafka,
  author    = {Chung, Benjamin and Li, Paley and Nardelli, Francesco Zappa and Vitek, Jan},
  booktitle = {ECOOP},
  title     = {KafKa: Gradual typing for objects},
  year      = {2018},
}

@InProceedings{belyakova2020world,
  author    = {Belyakova, Julia and Chung, Benjamin and Gelinas, Jack and Nash, Jameson and Tate, Ross and Vitek, Jan},
  booktitle = {OOPSLA},
  title     = {World age in Julia: optimizing method dispatch in the presence of eval},
  year      = {2020},
}

@InProceedings{Castagna94,
  author    = {Castagna, Giuseppe and Pierce, Benjamin C.},
  booktitle = {Symposium on Principles of Programming Languages (POPL)},
  title     = {Decidable Bounded Quantification},
  doi       = {10.1145/174675.177844},
  year      = {1994},
}

@InProceedings{clifton2000multijava,
  author    = {Clifton, Curtis and Leavens, Gary T. and Chambers, Craig and Millstein, Todd},
  booktitle = {Conference on Object-oriented Programming, Systems, Languages, and Applications (OOPSLA)},
  title     = {MultiJava: Modular Open Classes and Symmetric Multiple Dispatch for Java},
  doi       = {10.1145/353171.353181},
  year      = {2000},
}

@PhdThesis{Litvinov03,
  author = {Litvinov, Vassily},
  title  = {Constraint-Bounded Polymorphism: an Expressive and Practical Type System for Object-Oriented Languages},
  school = {University of Washington},
  year   = {2003},
}

@InProceedings{10.1007/978-3-642-10672-9_10,
  author    = {Wehr, Stefan and Thiemann, Peter},
  booktitle = {Programming Languages and Systems (ESOP)},
  title     = {On the Decidability of Subtyping with Bounded Existential Types},
  year      = {2009},
}

@InProceedings{Tobin16,
  author    = {Bonnaire-Sergeant, Ambrose and Davies, Rowan and Tobin-Hochstadt, Sam},
  booktitle = {European Symposium on Programming (ESOP)},
  title     = {Practical optional types for {Clojure}},
  doi       = {10.1007/978-3-662-49498-1_4},
  year      = {2016},
}

@InProceedings{Litvinov98,
  author    = {Litvinov, Vassily},
  booktitle = {Addendum to the Conference on Object-oriented Programming, Systems, Languages, and Applications},
  title     = {Constraint-based Polymorphism in {Cecil}: Towards a Practical and Static Type System},
  doi       = {10.1145/346852.346948},
  year      = {1998},
}

@inproceedings{pierce1992bounded,
  title={Bounded quantification is undecidable},
  author={Pierce, Benjamin C},
  booktitle={Proceedings of the 19th ACM SIGPLAN-SIGACT symposium on Principles of programming languages},
  pages={305--315},
  year={1992}
}

@article{maranget2007warnings,
  title={Warnings for pattern matching},
  author={Maranget, Luc},
  journal={Journal of Functional Programming},
  volume={17},
  number={3},
  pages={387--421},
  year={2007},
  publisher={Cambridge University Press}
}

@phdthesis{bezansonthesis,
  author       = {Jeff Bezanson}, 
  title        = {Abstraction in technical computing},
  school       = {Massachusetts Institute of Technology },
  url={http://dspace.mit.edu/handle/1721.1/7582},
  year         = 2015,
}

@inproceedings{yuliasubtyping,
  author = {Julia Belyakova},
  year = {2019}, 
  title = {Decidable Tag-Based Semantic Subtyping for Nominal Types, Tuples, and Unions},
  booktitle = {Proceedings of the 21st Workshop on Formal Techniques for Java-like Programs {FTFJP}}
}

@STRING{dls	= {Symposium on Dynamic languages (DLS)} }

@STRING{ecoop	= {European Conference on Object-Oriented Programming
		  (ECOOP)} }

@STRING{oopsla	= "Conference on Object-Oriented Programming Systems,
		  Languages and Applications (OOPSLA)" }

@STRING{pacmpl	= {Proc. ACM Program. Lang.} }

@STRING{pldi	= {Conference on Programming Language Design and
		  Implementation (PLDI)} }

@STRING{popl	= {Symposium on Principles of Programming Languages (POPL)} }

@Misc{		  hack13,
  author	= {Facebook},
  note		= {\url{http://hacklang.org}},
  title		= {Hack},
  year		= {2016}
}

@InProceedings{	  tf-popl08,
  author	= {Sam Tobin-Hochstadt and Matthias Felleisen},
  booktitle	= {Symposium on Principles of Programming Languages (POPL)},
  title		= {The design and implementation of typed {Scheme}},
  year		= {2008},
  doi		= {10.1145/1328438.1328486}
}

@InProceedings{	  siektaha06,
  author	= {Jeremy Siek},
  booktitle	= {Scheme and Functional Programming Workshop},
  note		= {\url{http://ecee.colorado.edu/~siek/pubs/pubs/2006/siek06_gradual.pdf}},
  title		= {Gradual Typing for Functional Languages},
  year		= {2006}
}

@Article{	  muehlboeck2017,
  author	= {Muehlboeck, Fabian and Tate, Ross},
  journal	= pacmpl,
  number	= {OOPSLA},
  title		= {Sound Gradual Typing is Nominally Alive and Well},
  year		= {2017},
  doi		= {10.1145/3133880},
}

@InProceedings{	  allen11,
  author	= {Allen, Eric and Hilburn, Justin and Kilpatrick,
                  Scott and Luchangco, Victor and Ryu, Sukyoung and
                  Chase, David and Steele, Guy},
  booktitle	= oopsla,
  title		= {Type Checking Modular Multiple Dispatch with Parametric
		  Polymorphism and Multiple Inheritance},
  year		= {2011},
  doi		= {10.1145/2048066.2048140}
}

@InProceedings{	  frisch02,
  author	= {Alain Frisch and Giuseppe Castagna and V{\'{e}}ronique
		  Benzaken},
  booktitle	= {Symposium on Logic in Computer Science {(LICS)}},
  title		= {Semantic Subtyping},
  year		= {2002},
  doi		= {10.1109/LICS.2002.1029823}
}

@InProceedings{	  pierce92,
  author	= {Pierce, Benjamin C.},
  booktitle	= {Symposium on Principles of Programming Languages (POPL)},
  title		= {Bounded Quantification is Undecidable},
  year		= {1992},
  doi		= {10.1145/143165.143228}
}

@InProceedings{	  Grigore:2017:JGT:3009837.3009871,
  author	= {Grigore, Radu},
  booktitle	= {Symposium on Principles of Programming Languages (POPL)},
  title		= {Java Generics Are Turing Complete},
  year		= {2017},
  doi		= {10.1145/3009837.3009871}
}

@InProceedings{	  lisp,
  author	= {McCarthy, John},
  booktitle	= {History of programming languages (HOPL)},
  title		= {History of LISP},
  year		= {1978},
  doi		= {10.1145/960118.808387}
}

@InProceedings{	  hoelzle92,
  author	= {Urs H{\"o}lzle and Craig Chambers and David Ungar},
  booktitle	= {Conference on Programming Language Design and
		  Implementation (PLDI)},
  key		= {PLDI'92},
  title		= {Debugging Optimized Code with Dynamic Deoptimization},
  year		= {1992}
}

@InProceedings{	  lb98,
  author	= {Sheng Liang and Gilad Bracha},
  booktitle	= {Conference on Object-oriented programming, systems,
		  languages, and applications (OOPSLA'98)},
  title		= {Dynamic class loading in the {Java} virtual machine},
  year		= {1998}
}

@InProceedings{	  detlefs99,
  author	= {Detlefs, David and Agesen, Ole},
  booktitle	= {ECOOP},
  editor	= {Guerraoui, Rachid},
  title		= {Inlining of Virtual Methods},
  year		= {1999}
}

@InProceedings{	  nguyen1996interprocedural,
  author	= {Nguyen, Trung and Gu, Junjie and Li, Zhiyuan},
  booktitle	= {Languages and Compilers for Parallel Computing},
  editor	= {Huang, Chua-Huang and Sadayappan, Ponnuswamy and Banerjee,
		  Utpal and Gelernter, David and Nicolau, Alex and Padua,
		  David},
  title		= {An interprocedural parallelizing compiler and its support
		  for memory hierarchy research},
  year		= {1996}
}

@InProceedings{	  glew2005type,
  author	= {Glew, Neal and Palsberg, Jens and Grothoff, Christian},
  booktitle	= {International Static Analysis Symposium},
  title		= {Type-safe optimisation of plugin architectures},
  year		= {2005}
}

@InProceedings{	  lee,
  author	= {Cook, Robert P. and Lee, Insup},
  booktitle	= {Proceedings of the Symposium on High-Level Debugging},
  title		= {DYMOS: A Dynamic Modification System},
  year		= {1983},
  doi		= {10.1145/1006147.1006188}
}

@Article{	  hicks,
  author	= {Stoyle, Gareth and Hicks, Michael and Bierman, Gavin and
		  Sewell, Peter and Neamtiu, Iulian},
  journal	= {ACM Trans. Program. Lang. Syst.},
  number	= {4},
  title		= {Mutatis Mutandis: Safe and Predictable Dynamic Software
		  Updating},
  volume	= {29},
  year		= {2007},
  doi		= {10.1145/1255450.1255455}
}

@InProceedings{	  politz12,
  author	= {Politz, Joe Gibbs and Carroll, Matthew J. and Lerner,
		  Benjamin S. and Pombrio, Justin and Krishnamurthi,
		  Shriram},
  booktitle	= {Symposium on Dynamic Languages},
  series	= {DLS},
  title		= {A Tested Semantics for Getters, Setters, and Eval in
		  JavaScript},
  year		= {2012},
  doi		= {10.1145/2384577.2384579},
}

@Article{	  glew2005method,
  author	= {Glew, Neal},
  journal	= {Journal of Object Technology},
  number	= {8},
  title		= {Method Inlining, Dynamic Class Loading, and Type
		  Soundness.},
  volume	= {4},
  year		= {2005},
  doi		= {10.5381/jot.2005.4.8.a2}
}

@Article{	  matthews2008operational,
  author	= {MATTHEWS, JACOB and FINDLER, ROBERT BRUCE},
  journal	= {Journal of Functional Programming},
  number	= {1},
  pages		= {47–86},
  publisher	= {Cambridge University Press},
  title		= {An operational semantics for Scheme},
  volume	= {18},
  year		= {2008},
  doi		= {10.1017/S0956796807006478}
}

@InProceedings{	  agrawal1991multimethods,
  author	= {Agrawal, Rakesh and Demichiel, Linda G. and Lindsay, Bruce
		  G.},
  booktitle	= {OOPSLA},
  title		= {Static Type Checking of Multi-Methods},
  year		= {1991}
}

@Article{	  chambers1995typechecking,
  author	= {Chambers, Craig and Leavens, Gary T.},
  title		= {Typechecking and Modules for Multimethods},
  year		= {1995}
}

@InProceedings{	  litvinov1998contraintpolymorphism,
  author	= {Litvinov, Vassily},
  booktitle	= {OOPSLA},
  title		= {Contraint-Based Polymorphism in Cecil: Towards a Practical
		  and Static Type System},
  year		= {1998},
}

@InProceedings{	  siek2015refined,
  title		= {Refined criteria for gradual typing},
  author	= {Siek, Jeremy G and Vitousek, Michael M and Cimini, Matteo
		  and Boyland, John Tang},
  booktitle	= {SNAPL},
  year		= {2015}
}

@InProceedings{	  takikawa2016sound,
  title		= {Is sound gradual typing dead?},
  author	= {Takikawa, Asumu and Feltey, Daniel and Greenman, Ben and
		  New, Max S and Vitek, Jan and Felleisen, Matthias},
  booktitle	= {POPL},
  year		= {2016}
}

@InProceedings{	  chambers1992object,
  title		= {Object-oriented multi-methods in Cecil},
  author	= {Chambers, Craig},
  booktitle	= {ECOOP’},
  year		= {1992}
}

@Article{	  oopsla18a,
  author	= {Jeff Bezanson and Jiahao Chen and Ben Chung and Stefan
		  Karpinski and {Viral B.} Shah and Jan Vitek and Lionel
		  Zoubritzky},
  year		= 2018,
  title		= "Julia: Dynamism and Performance Reconciled by Design",
  journal	= pacmpl,
  number	= {{OOPSLA}},
  volume	= 2,
  doi		= {10.1145/3276490}
}

@Article{	  oopsla18b,
  author	= {Francesco {Zappa Nardelli} and Julia Belyakova and Artem
		  Pelenitsyn and Benjamin Chung and Jeff Bezanson and Jan
		  Vitek},
  year		= 2018,
  title		= "Julia Subtyping: A Rational Reconstruction",
  journal	= pacmpl,
  number	= {{OOPSLA}},
  volume	= 2,
  doi		= {10.1145/3276483}
}

@Article{	  popl18,
  author	= {Olivier Fl{\"{u}}ckiger and Gabriel Scherer and Ming{-}Ho
		  Yee and Aviral Goel and Amal Ahmed and Jan Vitek},
  title		= {Correctness of speculative optimizations with dynamic
		  deoptimization},
  journal	= pacmpl,
  number	= {{POPL}},
  volume	= {2},
  year		= {2018},
  doi		= {10.1145/3158137}
}

@InProceedings{	  ecoop17,
  author	= {Todd A. Anderson and Hai Liu and Lindsey Kuper and Ehsan
		  Totoni and Jan Vitek and Tatiana Shpeisman},
  title		= {Parallelizing Julia with a Non-Invasive {DSL}},
  booktitle	= ecoop,
  year		= 2017,
  doi		= {10.4230/LIPIcs.ECOOP.2017.4}
}

@InProceedings{	  vee14,
  author	= {Kalibera, Tomas and Maj, Petr and Morandat, Floreal and
		  Vitek, Jan},
  title		= {A Fast Abstract Syntax Tree Interpreter for R},
  booktitle	= {Conference on Virtual Execution Environments (VEE)},
  year		= {2014},
  doi		= {10.1145/2576195.2576205},
}

@InProceedings{	  oopsla09,
  title		= "Thorn: Robust, Concurrent, Extensible Scripting on the
		  {JVM}",
  author	= {Bard Bloom and John Field and Nathaniel Nystrom and Johan
		  \"Ostlund and Gregor Richards and Rok Strnisa and Jan Vitek
		  and Tobias Wrigstad},
  booktitle	= oopsla,
  doi		= {10.1145/1639950.1640016},
  year		= 2009
}

@InProceedings{	  riposte,
  booktitle	= {2012 21st International Conference on Parallel
		  Architectures and Compilation Techniques (PACT)},
  title		= {Riposte: A trace-driven compiler and parallel VM for
		  vector code in R},
  year		= {2012},
  doi		= {10.1145/2370816.2370825 },
}

@InProceedings{	  mat1,
  author	= {Maxime Chevalier{-}Boisvert and Laurie J. Hendren and
		  Clark Verbrugge},
  title		= {Optimizing Matlab through Just-In-Time Specialization},
  booktitle	= {Conference on Compiler Construction (CC)},
  year		= {2010},
  doi		= {10.1007/978-3-642-11970-5\_4}
}

@Article{	  mat2,
  author	= {De Rose, Luiz and Padua, David},
  title		= {Techniques for the Translation of MATLAB Programs into
		  Fortran 90},
  journal	= {ACM Trans. Program. Lang. Syst.},
  issue_date	= {March 1999},
  volume	= {21},
  number	= {2},
  year		= {1999},
  doi		= {10.1145/316686.316693}
}

@InProceedings{	  nool17,
  title		= "Towards Typing {Julia}",
  year		= 2017,
  author	= "Benjamin Chung and Paley Li",
  booktitle	= "The -2th Workshop on New Object-Oriented Languages (NOOL)"
}

@InProceedings{	  holzle94,
  author	= {H\"{o}lzle, Urs and Ungar, David},
  title		= {Optimizing Dynamically-dispatched Calls with Run-time Type
		  Feedback},
  booktitle	= {Conference on Programming Language Design and
		  Implementation (PLDI)},
  year		= {1994},
  doi		= {10.1145/773473.178478}
}

@InProceedings{	  hotspot,
  author	= {Paleczny, Michael and Vick, Christopher and Click, Cliff},
  title		= {The Java {HotSpot} Server Compiler},
  booktitle	= {Symposium on Java Virtual Machine Research and Technology
		  (JVM)},
  year		= {2001},
  url		= {http://dl.acm.org/citation.cfm?id=1267847.1267848}
}

@InProceedings{	  graal,
  author	= {W\"{u}rthinger, Thomas and Wimmer, Christian and
		  W\"{o}\ss, Andreas and Stadler, Lukas and Duboscq, Gilles
		  and Humer, Christian and Richards, Gregor and Simon, Doug
		  and Wolczko, Mario},
  title		= {One VM to Rule Them All},
  booktitle	= {Symposium on New Ideas, New Paradigms, and Reflections on
		  Programming \& Software (Onward!)},
  year		= {2013},
  doi		= {10.1145/2509578.2509581}
}

@InProceedings{	  bobrow86,
  author	= {Bobrow, Daniel G. and Kahn, Kenneth and Kiczales, Gregor
		  and Masinter, Larry and Stefik, Mark and Zdybel, Frank},
  title		= {CommonLoops: Merging {Lisp} and Object-oriented
		  Programming},
  booktitle	= {Conference on Object-oriented Programming Systems,
		  Languages and Applications (OOPSLA)},
  year		= {1986},
  doi		= {10.1145/28697.28700}
}

@InProceedings{	  gabriel87,
  author	= {Linda DeMichiel and Richard Gabriel},
  title		= {The {Common {Lisp} Object System}: An Overview},
  booktitle	= {European Conference on Object-Oriented Programming
		  (ECOOP)},
  year		= 1987,
  doi		= {10.1007/3-540-47891-4_15}
}

@Book{		  randal03,
  author	= "Allison Randal and Dan Sugalski and Leopold Toetsch",
  title		= "Perl 6 and Parrot Essentials",
  publisher	= "O'Reilly",
  year		= 2003
}

@Article{	  chambers14,
  author	= "John Chambers",
  title		= "Object-Oriented Programming, Functional Programming and R",
  journal	= " Statistical Science",
  issue		= 29,
  year		= 2014,
  number	= 2,
  doi		= {10.1214/13-STS452}
}

@InProceedings{	  llvm,
  author	= {Chris Lattner and Vikram Adve},
  title		= {{LLVM}: A Compilation Framework for Lifelong Program
		  Analysis and Transformation},
  booktitle	= {Symposium on Code Generation and Optimization f(CGO)},
  year		= 2004,
  doi		= {10.1109/CGO.2004.1281665}
}

@Article{	  x10,
  author	= {Charles, Philippe and Grothoff, Christian and Saraswat,
		  Vijay and Donawa, Christopher and Kielstra, Allan and
		  Ebcioglu, Kemal and von Praun, Christoph and Sarkar,
		  Vivek},
  title		= {X10: An Object-oriented Approach to Non-uniform Cluster
		  Computing},
  booktitle	= {Conference on Object-oriented Programming, Systems,
		  Languages, and Applications (OOPSLA)},
  year		= 2005,
  doi		= {10.1145/1103845.1094852}
}

@InCollection{	  fortress,
  author	= {Guy Steele and Eric Allen and David Chase and Christine
		  Flood and Victor Luchangco and Jan{-}Willem Maessen and
		  Sukyoung Ryu},
  title		= {Fortress (Sun {HPCS} Language)},
  booktitle	= {Encyclopedia of Parallel Computing},
  pages		= {718--735},
  year		= {2011},
  doi		= {10.1007/978-0-387-09766-4_190}
}

@Book{		  matlab,
  year		= {2018},
  author	= {MATLAB},
  title		= {version 9.4},
  publisher	= {The MathWorks Inc.},
  address	= {Natick, Massachusetts}
}

@Manual{	  r,
  title		= {R: A Language and Environment for Statistical Computing},
  author	= {{R Core Team}},
  organization	= {R Foundation for Statistical Computing},
  year		= {2008},
  url		= {http://www.R-project.org}
}

@InProceedings{	  lit98,
  author	= {Litvinov, Vassily},
  title		= {Constraint-based Polymorphism in Cecil: Towards a
		  Practical and Static Type System},
  booktitle	= {Addendum to Conference on Object-oriented Programming,
		  Systems, Languages, and Applications (OOPSLA Addendum)},
  year		= {1998},
  doi		= {10.1145/346852.346948}
}

@InProceedings{	  miles13,
  author	= {Lubin, Miles and Dunning, Iain},
  title		= {Computing in Operations Research using Julia},
  booktitle	= {INFORMS Journal on Computing},
  year		= 2013,
  doi		= {10.1287/ijoc.2014.0623}
}

@PhDThesis{	  fagan1991soft,
  title		= {Soft typing: an approach to type checking for dynamically
		  typed languages},
  author	= {Fagan, Mike},
  year		= {1991},
  school	= {Rice University}
}

@InProceedings{	  wright1994practical,
  title		= {A practical soft type system for Scheme},
  author	= {Wright, Andrew and Cartwright, Robert},
  booktitle	= {ACM SIGPLAN Lisp Pointers},
  volume	= {7},
  number	= {3},
  pages		= {250--262},
  year		= {1994}
}

@Article{	  chaudhuri2017fast,
  author	= {Chaudhuri, Avik and Vekris, Panagiotis and Goldman, Sam
		  and Roch, Marshall and Levi, Gabriel},
  title		= {Fast and Precise Type Checking for JavaScript},
  journal	= {Proc. ACM Program. Lang.},
  volume	= {1},
  number	= {OOPSLA},
  year		= {2017},
  doi		= {10.1145/3133872},
}

@Misc{		  typescript13,
  title		= "{T}ype{S}cript -- Language Specification",
  author	= "Microsoft",
  year		= 2016
}

@InProceedings{	  tf-dls06,
  author	= {Sam Tobin-Hochstadt and Matthias Felleisen},
  title		= {Interlanguage migration: from scripts to programs},
  year		= {2006},
  doi		= {10.1145/1176617.1176755},
  booktitle	= dls
}

@InCollection{	  shiv90,
  author	= {Olin Shivers},
  booktitle	= {Topics in Advanced Language Implementation},
  publisher	= {MIT Press},
  pages		= {47-88},
  title		= {Data-flow Analysis and Type Recovery in Scheme},
  year		= 1990
}

@Misc{		  plbg,
  author	= "Isaac Gouy",
  title		= "The Computer Language Benchmarks Game",
  year		= 2018,
  url		= "https://benchmarksgame-team.pages.debian.net/benchmarksgame"
}

@InProceedings{	  rata,
  author	= "Logozzo, Francesco and Venter, Herman",
  editor	= "Gupta, Rajiv",
  title		= "RATA: Rapid Atomic Type Analysis by Abstract
                  Interpretation -- Application to JavaScript
                  Optimization",
  booktitle	= "Compiler Construction",
  year		= "2010",
  doi		= "10.1007/978-3-642-11970-5_5"
}

@Article{	  hackett2012fast,
  title		= {Fast and Precise Hybrid Type Inference for JavaScript},
  author	= {Hackett, Brian and Guo, Shu-yu},
  booktitle	= {Proceedings of the 33rd ACM SIGPLAN Conference on
		  Programming Language Design and Implementation},
  series	= {PLDI '12},
  year		= {2012},
  doi		= {10.1145/2254064.2254094},
}

@InProceedings{	  aiken1993type,
  author	= {Aiken, Alexander and Wimmers, Edward L.},
  title		= {Type Inclusion Constraints and Type Inference},
  booktitle	= {Proceedings of the Conference on Functional Programming
		  Languages and Computer Architecture (FPCA)},
  year		= {1993},
  doi		= {10.1145/165180.165188},
}

@Article{	  garciadyn,
  author	= {Garcia, Miguel and Ortin, Francisco and Quiroga, Jose},
  title		= {Design and Implementation of an Efficient Hybrid Dynamic
		  and Static Typing Language},
  journal	= {Softw. Pract. Exper.},
  volume	= {46},
  number	= {2},
  year		= {2016},
  doi		= {10.1002/spe.2291},
}

@article{DBLP:journals/pacmpl/ParkHSR19,
  author    = {Gyunghee Park and
               Jaemin Hong and
               Guy L. Steele Jr. and
               Sukyoung Ryu},
  title     = {Polymorphic symmetric multiple dispatch with variance},
  number    = {{POPL}},
  year      = {2019},
  doi       = {10.1145/3290324},
}

@mastersthesis{hu2019decidability,
  title={Decidability and Algorithmic Analysis of Dependent Object Types (DOT)},
  author={Hu, Zhong Sheng},
  year={2019},
  school={University of Waterloo}
}

@inproceedings{mehnert2010extending,
  title={Extending Dylan's type system for better type inference and error detection},
  author={Mehnert, Hannes},
  booktitle={Proceedings of the 2010 international conference on Lisp},
  pages={1--10},
  year={2010}
}

@article{amin2014foundations,
  title={Foundations of path-dependent types},
  author={Amin, Nada and Rompf, Tiark and Odersky, Martin},
  journal={ACM Sigplan Notices},
  year={2014},
}

@article{yu2020proportional,
  title={Proportional-integral projected gradient method for model predictive control},
  author={Yu, Yue and Elango, Purnanand and A{\c{c}}{\i}kme{\c{s}}e, Beh{\c{c}}et},
  journal={IEEE Control Systems Letters},
  year={2020},
  publisher={IEEE}
}

@article{graf2020lower,
  title={Lower your guards: a compositional pattern-match coverage checker},
  author={Graf, Sebastian and Peyton Jones, Simon and Scott, Ryan G},
  journal={Proceedings of the ACM on Programming Languages},
  number={ICFP},
  year={2020},
}

@inproceedings{scharli2003traits,
  title={Traits: Composable units of behaviour},
  author={Sch{\"a}rli, Nathanael and Ducasse, St{\'e}phane and Nierstrasz, Oscar and Black, Andrew P},
  booktitle={ECOOP},
  year={2003},
}

@inproceedings{tate2011taming,
  title={Taming wildcards in Java's type system},
  author={Tate, Ross and Leung, Alan and Lerner, Sorin},
  booktitle={PLDI},
  year={2011}
}

@article{bagnara2022coding,
  title={Coding Guidelines and Undecidability},
  author={Bagnara, Roberto and Bagnara, Abramo and Hill, Patricia M},
  journal={arXiv preprint arXiv:2212.13933},
  year={2022}
}

@inproceedings{gil2019fling,
  title={Fling-A fluent API generator},
  author={Gil, Yossi and Roth, Ori},
  booktitle={ECOOP},
  year={2019}
}

\end{document}